\documentclass[a4paper,onecolumn,10pt]{quantumarticle}
\pdfoutput=1
\usepackage[utf8]{inputenc}
\usepackage[english]{babel}
\usepackage[T1]{fontenc}
\usepackage{amsmath,amssymb,amsfonts}
\usepackage{hyperref}

\usepackage{tikz}
\usepackage{lipsum}
\usepackage{algorithmic}
\usepackage{algorithm}
\usepackage{cite}
\usepackage{amsmath,amssymb,amsfonts}
\usepackage{halloweenmath}
\usepackage{algorithmic}
\usepackage{graphicx}
\usepackage{textcomp}
\usepackage{xcolor}
\usepackage{float}
\usepackage{physics}
\usepackage{placeins}
\usepackage{subfigure}
\usepackage{amsthm,amssymb}
\numberwithin{figure}{section}
\newtheorem{theorem}{Theorem}
\newtheorem{lemma}{Lemma}
\newtheorem{corollary}{Corollary}
\newtheorem{proposition}{Proposition}
\newtheorem{example}{Example}
\newtheorem{remark}{Remark}
\newtheorem{definition}{Definition}

\usepackage[numbers]{natbib}
\begin{document}

\title{Quantum Constacyclic BCH Codes over Qudits: A Spectral-Domain Approach}

\author{Shikha Patel}
\affiliation{Department of Electronic Systems Engineering, Indian Institute of Science, Bengaluru, India, 560012}
\email{shikhapatel@iisc.ac.in}
\author{Shayan Srinivasa Garani}
\email{shayangs@iisc.ac.in}

\maketitle
	\begin{abstract}
		We characterize constacyclic codes in the spectral domain using the finite field Fourier transform (FFFT) and propose a reduced complexity method for the spectral-domain decoder.  Further, we also consider repeated-root constacyclic codes and characterize them in terms of symmetric and asymmetric $q$-cyclotomic cosets. Using zero sets of classical self-orthogonal and dual-containing codes, we derive quantum error correcting codes (QECCs) for both constacyclic Bose-Chaudhuri-Hocquenghem (BCH) codes and repeated-root constacyclic codes. We provide some examples of QECCs derived from repeated-root constacyclic codes and show that constacyclic BCH codes are more efficient than repeated-root constacyclic codes. Finally, quantum encoders and decoders are also proposed in the transform domain for Calderbank-Shor-Steane CSS-based quantum codes. Since constacyclic codes are a generalization of cyclic codes with better minimum distance than cyclic codes with the same code parameters, the proposed results are practically useful.
	\end{abstract}
\section{Introduction}
 One of the powerful generalizations of linear cyclic codes over finite fields is constacyclic codes \cite{peterson1972error, chen2012constacyclic} with a rich algebraic structure, useful towards practice. One can specify $\lambda$-constacyclic codes of length $n$ over a finite field $\mathbb{F}_q$ by ideals of the polynomial ring $\frac{\mathbb{F}_q[x]}{\langle x^n-\lambda\rangle}$. Moreover, for the same code parameters 
 one can find better lower bounds for the minimum distance of constacyclic codes over their cyclic counterparts, motivating the investigation of constacyclic codes in this paper. Many prior works in literature are based on simple-root codes where the code length is relatively prime to the characteristic of $\mathbb{F}_{p^m}$. However, when the code length $n$ is divisible by the characteristic $p$ of the field, they are referred to repeated-root codes. For more details on repeated-root codes, we refer the reader to (\cite{chen2014repeated,dinh2014repeated,castagnoli1991repeated}). In this paper, we investigate  both simple root and repeated-root constacyclic codes for quantum error correction. 
 
Fourier transforms are widely used in many engineering applications, such as in signal processing and communications. Spectral techniques based on Fourier transforms over finite fields have powerful applications in coding theory and practice. Cyclic codes over finite fields in the transform domain were originally introduced by Blahut \cite{blahut1979algebraic}, \cite{blahut1983theory}. In \cite{blahut1979transform}, Blahut considered cyclic BCH codes and their decoding via the syndrome computations in the frequency domain. In this sequel, many researchers considered special linear codes like abelian codes, quasi-cyclic codes, etc. in the transform domain \cite{rajan1992transform,rajan1994generalized,rajan1994transform,massey1998discrete,dey2003dft}. Unlike cyclic codes, to the best of our knowledge, no one has hitherto investigated spectral decoding methods for constacyclic codes, which is the generalization of Blahut's work. Transform domain coding has advantages for building efficient encoding and decoding circuits since parity positions are based on the spectral nulls of the codes  \cite{mondal2021efficient,mondal2020efficient}. Motivated by these considerations, we extend the transform domain approach developed for cyclic codes to the constacyclic case, with an eye to get better code properties and efficient circuits towards practice. 

With the evolution of quantum technologies for sensing, computing, communications and security, the role of quantum error correction is indispensable in such systems. Any computing device dealing with fragile state like single photons, etc. must overcome decoherence due to the coupling of the state with the environment. The stabilizer framework by Gottesman \cite{gottesman1997stabilizer} and the CSS sub-frameworks \cite{calderbank1996good,steane1996error} based on the dual-containing property play an important role for deriving quantum codes from well-known classical codes with good properties. In the recent years, many researchers have focussed on QECCs over qudits \cite{galindo2019entanglement,nadkarni2021non,ashikhmin2001nonbinary,luo2017non,ketkar2006nonbinary}. Compared to qubits, fewer qudits are required to store the same amount of information, leading to overall  efficiencies in computations and in overcoming noise within error correction circuits over qudit-based systems. Similar to classical codes used in almost all transmission and data storage channels to overcome random errors, burst erasures and defects, quantum codes will analogously be useful for realizing reliable quantum transmission and memory systems. 

 

Among the class of quantum algebraic cyclic codes, quantum BCH codes and Reed-Solomon codes are particularly interesting due to their burst erasure correcting property and guaranteed error correction performance. The quantum BCH code was originally constructed from classical cyclic BCH codes by Grassl \cite{grassl1999quantum}. Quantum BCH codes derived from cyclic BCH codes in the time domain has been well-investigated by many authors \cite{aly2007quantum,steane1996simple,grassl2000cyclic}. Since the spectral null patterns over the coded states \cite{ying2006quantum} are useful for determining the bounds on the minimum distance and  also facilitate better encoding/decoding, we explore quantum constacyclic codes in the spectral domain. 

Our contributions in this paper are as follows: 
\begin{enumerate}
    \item We propose a novel spectral decoding method for systematic constacyclic Bose–Chaudhuri–\\Hocquenghem (BCH) codes with reduced computational complexity by a linear factor over well-known algorithms, such as the Peterson–Gorenstein–Zierler (PGZ) and Berlekamp-Massey (BM) algorithm. We use this result for the syndrome computation unit while decoding quantum constacyclic codes. 
    \item We characterize the spectral properties of constacyclic codes and design quantum maximum-distance-separable (MDS)\footnote{An $[[n,k,d]]_q$ quantum MDS code attaining the quantum Singleton bound (QSB) $2d\leq n-k+2$, with equality \cite{calderbank1998quantum}.}  constacyclic codes with better minimum distance than the prior work on spectral-based quantum cyclic codes in~\cite{chowdhury2009quantum}.
    \item We construct some examples to show that BCH constacyclic codes are more efficient than repeated-root constacyclic codes.
    \item We provide an algorithmic framework to
construct the encoding operator for quantum constacyclic BCH codes, completing the entire code design.
\end{enumerate}

The rest of the article is organized as follows: In Section \ref{sec2}, we derive some important algebraic properties of constacyclic codes in the transform domain using finite field Fourier transforms. We provide the construction of classical constacyclic BCH codes that are important for the deriving constacyclic QECCs. In Section \ref{sec 3}, we propose a spectral\footnote{The reader must note that the words spectral domain and transform domain are used interchangeably.} decoding algorithm for these codes with reduced computational complexity than well-known decoding algorithms, such as the Berlekamp-Massey and Peterson–Gorenstein–Zierler (PGZ) algorithms. We also construct some examples to validate the derived results. In Section \ref{sec4}, we construct QECCs from both self-orthogonal and dual-containing constacyclic BCH codes in the transform domain. Further, we consider repeated-root constacyclic codes and use them for QECCs. We also show that constacyclic BCH codes are more efficient than repeated-root constacyclic codes. In Section \ref{secA}, we provide an encoding circuit architecture for quantum constacyclic BCH codes. In Section \ref{decoding}, we provide the error correction procedure, error deduction and recovery procedures for quantum constacyclic BCH codes along with the syndrome computation circuits, followed by conclusions in Section \ref{conclu}. 

 \section{ Constacyclic Codes in the Transform Domain}\label{sec2}
 \begin{definition}
Let $\mathbb{F}_q$ be a finite field of order $q=p^s$ where $p$ is a prime and $s$ is a positive integer. Let $\lambda$ be a unit in $\mathbb{F}_q$. Then a linear code $C$ of length $n$ over $\mathbb{F}_q$ is said to be a $\lambda$-constacyclic (or simply constacyclic) code if $\tau_{\lambda}(c):=(\lambda c_{n-1}, c_0,\dots, c_{n-2})\in C$ whenever $c=(c_0,c_1,\dots, c_{n-1})\in C$. Note that $\tau_{\lambda}$ is known as a $\lambda$-constacyclic shift operator. When $\lambda=1$, it is a \textit{cyclic} code, and for $\lambda=-1$, it is a \textit{negacyclic} code.
\end{definition}
Any $\lambda$-constacyclic code $C$ of length $n $ over $\mathbb{F}_q$ is identified with an ideal of the quotient ring $\frac{\mathbb{F}_q[x]}{\langle x^n-\lambda\rangle}$ where $\langle  x^n-\lambda \rangle$ denotes the ideal generated by  $x^n-\lambda$ of the polynomial ring $\mathbb{F}_q[x]$. Also, $C$ is generated by a factor polynomial of $x^n-\lambda$, called the polynomial generator of the $\lambda$-constacyclic code $C$. In order to obtain all the $\lambda$-constacyclic codes of length $n$ over $\mathbb{F}_q$, we need to determine all the irreducible factors of $x^n-\lambda$ over $\mathbb{F}_q$.  There is a criterion on irreducible polynomials over  $\mathbb{F}_{q}$, which was given by Serret in 1866 (see \cite{wan2003lectures,lidl1997finite}).
\begin{theorem}[\cite{wan2003lectures}, Theorem 10.7]
     Let $n \geq 2$. For any $\lambda\in \mathbb{F}_{q}^*=\mathbb{F}_{q}\setminus \{0\}$ with $ord(\lambda) = \kappa$, the polynomial $x^n-\lambda$ is irreducible
over $\mathbb{F}_q$ if and only if both the following two conditions are satisfied:
\begin{enumerate}
    \item  Every prime divisor of $n$ divides $\kappa$, but does not divide $\frac{q-1}{\kappa}$;
    \item If $4|n$, then $4|(q-1)$.
\end{enumerate}
\end{theorem}
The key idea to derive the transform domain properties of constacyclic codes lies in factorization of $x^n-\lambda$. To do this, we consider a positive integer $n$ such that $n|(q^m-1)$, where $m$ is a  positive integer. Let $\Phi: \frac{\mathbb{F}_{q^m}[x]}{\langle x^{n}-1\rangle} \longrightarrow \frac{\mathbb{F}_{q^m}[x]}{\langle x^{n}-\lambda\rangle}$ be a map defined by
$\Phi(f(x)+\langle x^n-1\rangle )= f(\beta^{-1}x)+\langle x^n-\lambda \rangle$, where $\beta\in \mathbb{F}_{q}^*$ and  $\lambda=\beta^n$. Let $f(x), g(x)\in \frac{\mathbb{F}_{q^m}[x]}{\langle x^{n}-\lambda\rangle}$. Then,
\begin{align*}
\Phi((f(x)+\langle x^n-1\rangle)+ (g(x)+\langle x^n-1\rangle))
   = &\Phi((f(x)+g(x))+\langle x^n-1\rangle)\\
   =&(f(\beta^{-1}x)+g(\beta^{-1}x))+\langle x^n-\lambda\rangle\\
 = &f(\beta^{-1}x)+\langle x^n-\lambda\rangle+ g(\beta^{-1}x)+\langle x^n-\lambda\rangle\\
  =&\Phi(f(x)+\langle x^n-1\rangle)+\Phi(g(x)+\langle x^n-1\rangle )\\
\text{and }~~~   \Phi((f(x)+\langle x^n-1\rangle)(g(x)+\langle x^n-1\rangle))= &\Phi((f(x)g(x))+\langle x^n-1\rangle)
   \\=&(f(\beta^{-1}x)g(\beta^{-1}x))+\langle x^n-\lambda\rangle\\
  = &(f(\beta^{-1}x)+\langle x^n-\lambda\rangle) (g(\beta^{-1}x)+\langle x^n-\lambda\rangle)\\
  =&\Phi(f(x)+\langle x^n-1\rangle )\Phi(g(x)+\langle x^n-1\rangle ).
\end{align*}
Next, we consider 
\begin{align*}
    \ker(\Phi)&=\{f(x)\in \frac{\mathbb{F}_{q^m}[x]}{\langle x^{n}-\lambda\rangle} ~|~ \Phi(f(x)+\langle x^n-1\rangle)=0\}\\
    &=\{f(x)\in \frac{\mathbb{F}_{q^m}[x]}{\langle x^{n}-\lambda\rangle }~|~ \Phi(f(x)+\langle x^n-1\rangle)=x^n-\lambda\}\\
    &=\{f(x)\in \frac{\mathbb{F}_{q^m}[x]}{\langle x^{n}-\lambda\rangle} ~|~ f(\beta^{-1}x)+\langle x^n-\lambda\rangle=x^n-\lambda\}\\
    &=\{f(x)\in \frac{\mathbb{F}_{q^m}[x]}{\langle x^{n}-\lambda\rangle} ~|~ f(\beta^{-1}x)\in \langle x^n-\lambda\rangle\}
    =\{0\}.
\end{align*} Thus, $\Phi$ is one-one and also onto. Therefore, $\Phi$ is a ring  isomorphism, implying
$$\frac{\mathbb{F}_{q^m}[x]}{\langle x^{n}-1\rangle } \cong \frac{\mathbb{F}_{q^m}[x]}{\langle x^{n}-\lambda\rangle}.$$ In this sequel, first we will factorize $x^n-1$. Let $\xi$ be a primitive $n^{\mathrm{th}}$ root of unity in $\mathbb{F}_{q^m}$. Then $\xi^n=1$ and for any $1\leq i\leq n-1$ we have $(\xi^i)^n= (\xi^n)^i=1$. Thus, $\xi^i$ is a root of $x^n-1$, and $x^n-1=(x-1)(x-\xi)(x-\xi^2)\cdots(x-\xi^{n-1}).$ Therefore, by the ring isomorphism we express
\begin{equation}\label{eq1}
    x^n-\lambda=\prod_{j=1}^{n}(x-\beta\xi^j).
\end{equation}
Next, we consider $\mathbb{F}_{q^m}$-algebra isomorphism $\Phi': \frac{\mathbb{F}_{q^m}[x] }{\langle x^{n}-\lambda \rangle} \longrightarrow \prod_{j=0}^{n-1}\frac{\mathbb{F}_{q^m}[x]}{\langle x-\beta\xi^j\rangle}$ defined by 
\begin{equation}\label{eq2}
    \Phi'\bigg(\sum_{i=0}^{n-1}a_ix^i\bigg)\triangleq\bigg(\sum_{i=0}^{n-1}a_i\beta^i,\sum_{i=0}^{n-1}a_i(\beta\xi)^i,\dots,\sum_{i=0}^{n-1}a_i(\beta\xi^{n-1})^i\bigg).
\end{equation}
Now, we can define an isomorphism to relate the combinatorial and algebraic structures of constacyclic codes. For this, we consider $\Phi'': \mathbb{F}_{q^m}^n \longrightarrow  \frac{\mathbb{F}_{q^m} [x]}{\langle x^{n}-\lambda \rangle}$ defined by 
\begin{equation}\label{eq3}
    \Phi''(a_0,a_1,\dots,a_{n-1})=\sum_{i=0}^{n-1}a_ix^i.
\end{equation}
From equations (\ref{eq2}) and
 (\ref{eq3}), we can define 
$A_j=\sum_{i=0}^{n-1}a_i(\beta\xi^j)^i ~~~\text{for }~~~j=0,1,\dots,n-1.$

To extend Blahut's approach of finite field Fourier transforms for cyclic codes towards the constacyclic case, we begin with some definitions.
\begin{definition}
    Let $\boldsymbol{a}=(a_0,a_1,\dots,a_{n-1})$ be a vector over $\mathbb{F}_q$ and $n|(q^m-1)$. Let $\xi$ be an element of $\mathbb{F}_{q^m}$ of order $n$ and $\lambda=\beta^n$. The finite field Fourier transform (FFFT) of the vector $\boldsymbol{a}$ is the vector over $\mathbb{F}_{q^m}$, $\boldsymbol{A}=\{A_j|j=0,1,\dots,n-1\}$, defined by 
    $A_j=\sum_{i=0}^{n-1}a_i(\beta\xi^j)^i ~~~\text{for }~~~j=0,1,\dots,n-1.$
\end{definition}
We next validate the inverse of FFFT for constacyclic codes.  
\begin{lemma}\label{lem1}
    Let $\xi$ be an element of $\mathbb{F}_{q^m}$ of order $n$, $n>1$ and $\lambda=\beta^n$. A $\lambda$-constacyclic code vector over $\mathbb{F}_q$ and its corresponding spectrum are related by: 
$A_j=\sum_{i=0}^{n-1}a_i(\beta\xi^j)^i,$ 
$a_i=\frac{1}{n\beta^i}\sum_{j=0}^{n-1}\xi^{-ij}A_j.$
\end{lemma}
\begin{proof}
    From equation $(\ref{eq1})$, we rewrite  
    \begin{align*}       
       x^n-\beta^n=&(x-\beta)(x^{n-1}+\beta x^{n-2}+\cdots+\beta^{n-2}x+\beta^{n-1})\\
       =&(x-\beta)\bigg(\sum_{i=0}^{n-1}\beta^ix^{n-1-i}\bigg).
    \end{align*} From the definition of $\xi$, $\beta\xi^r$ is a root of the above equation for all $r=0,1,...,n-1$. Hence, $\beta\xi^r$ is a root of the last term for all $r \neq 0 \mod(n)$. This is also written as 
    $\sum_{i=0}^{n-1}\beta^i(\beta\xi^{r})^{n-1-i}=\sum_{i=0}^{n-1}(\beta\xi^r)^{n-1}\xi^{-ri}=0,$ $r \neq 0 \mod(n)$. If $r=0$, then $\sum_{i=0}^{n-1}(\beta\xi^r)^{n-1}\xi^{-ri}=n\beta^{n-1}\mod(p)$. This is not zero if $n$ is not a multiple of the field characteristic $p$. Next, $\sum_{j=0}^{n-1}(\beta\xi^j)^{-i}\sum_{k=0}^{n-1}(\beta\xi^{j})^ka_k\\=\sum_{k=0}^{n-1}a_k\sum_{j=0}^{n-1}\beta^{k-i}\xi^{(k-i)j}=n\mod(p)a_i$.
Since $n$ is not a multiple of $p$, $n\mod(p)\neq 0$, proving the lemma.
\end{proof}
\subsection{Properties of constacyclic codes } \label{sub2.1}
In this subsection, we discuss the properties of $\lambda$-constacyclic codes in the frequency domain. In \cite{boztas1998constacyclic}, only some properties are stated without proofs. Here, we will consider each of those in detail. First, we will define a shift in the transform domain.
 \begin{theorem} (Constacyclic shift property in the transform domain)
     If $\boldsymbol{A} = \text{FFFT}(\boldsymbol{a})$, $\boldsymbol{b}\in\mathbb{F}_q^n$
such that $b_i = a_{i-1}$ for $i=1,2,\dots,n-1$ and $b_0=\lambda a_{n-1}$, and
$\boldsymbol{B} = \text{FFFT} (\boldsymbol{b})$, then $B_j =  \beta\xi^j A_j$.
 \end{theorem}
 \begin{proof}
Let $\boldsymbol{A} = \text{FFFT} (\boldsymbol{a})$ and $\boldsymbol{B }= \text{FFFT} (\boldsymbol{b})$. Then, 
\begin{align*}
    B_j=&\sum_{i=0}^{n-1}b_i(\beta\xi^j)^i
    =b_0+\sum_{i=1}^{n-1}a_{i-1}(\beta\xi^j)^i\\
    =&\lambda a_{n-1}+\sum_{i=1}^{n-1}a_{i-1}(\beta\xi^j)^i.
    \end{align*}
    Also,
    \allowdisplaybreaks
    \begin{align*}\beta\xi^j A_j=&\beta\xi^j\sum_{i=0}^{n-1}a_i(\beta\xi^j)^i
     =\sum_{i=0}^{n-1}a_i(\beta\xi^j)^{i+1}\\
=&\beta\xi^ja_0+\beta^2\xi^{2j}a_1+\beta^3\xi^{3j}a_2+\cdots+\beta^{n-1}\xi^{(n-1)j}a_{n-2}+\beta^n\xi^{nj}a_{n-1}\\
     =&\lambda a_{n-1}+ \beta\xi^ja_0+(\beta\xi^{j})^2a_1+(\beta\xi^{j})^3a_2+\cdots+(\beta\xi^{j})^{n-1}a_{n-2}\\
     =&\lambda a_{n-1}+\sum_{i=1}^{n-1}a_{i-1}(\beta\xi^j)^i=B_j.
\end{align*}
\end{proof}
 The FFFT previously defined for constacyclic codes has the following additional properties.
\begin{theorem} (Convolutional Property)
Let $\boldsymbol{a},\boldsymbol{b}$ and $\boldsymbol{c}$ are vectors over $\mathbb{F}_q$ such that $c_i=a_ib_i$ for $i=0,1,\dots,n-1$. Then their FFFT coefficients satisfy the relation 
$C_j= \frac{1}{n}\sum_{k=0}^{n-1}B_k\mathcal{A}_{j-k}$, where $j=0,1,\dots,n-1$ and conversely.
    
\end{theorem}
\begin{proof}
Let $c_i=a_ib_i$ for $i=0,1,\dots,n-1$.
Now, taking the FFFT of $(c_0,c_1,\dots,c_{n-1})$ we get
\begin{align*}
    C_j=&\sum_{i=0}^{n-1}c_i(\beta\xi^j)^i
    =\sum_{i=0}^{n-1}(\beta\xi^j)^ia_i\frac{1}{n}\sum_{k=0}^{n-1}(\beta\xi^k)^{-i}B_k\\
    =&\frac{1}{n}\sum_{k=0}^{n-1}B_k\sum_{i=0}^{n-1}\beta^{-i}\beta^i(\xi)^{i(j-k)}a_i\\
    =&\frac{1}{n}\sum_{k=0}^{n-1}B_k\sum_{i=0}^{n-1}\xi^{i(j-k)}a_i
    =\frac{1}{n}\sum_{k=0}^{n-1}B_k\mathcal{A}_{j-k},
\end{align*} where $\mathcal{A}_{j-k}$ denotes the finite field Fourier transform of cyclic codes (in this case $\beta=1$).

Conversely, let $C_j=A_jB_j$, $j=0,1,\dots,n$. Now, by taking the inverse FFFT of  $(C_0,C_1,\dots,C_{n-1})$, we get
\begin{align*}
    c_i=&\frac{1}{n}\sum_{j=0}^{n-1}(\beta\xi^j)^{-i}C_j
    =\frac{1}{n}\sum_{j=0}^{n-1}(\beta\xi^j)^{-i}A_j\sum_{k=0}^{n-1}(\beta\xi^j)^kb_k\\
    =&\sum_{k=0}^{n-1}b_k\bigg(\frac{1}{n}\sum_{j=0}^{n-1}(\beta\xi^j)^{k-i}A_j\bigg)
    =\sum_{k=0}^{n-1}b_ka_{i-k}.
\end{align*}
Hence, the result follows.
\end{proof}
The following property identifies those vectors over $\mathbb{F}_{q^m}$ that are images of vectors over $\mathbb{F}_q$ under FFFT.
\begin{theorem} (Conjugate Symmetry Property)\label{th3}
Let $A_j$ for $j=0,1,\dots,n-1$ take elements in $\mathbb{F}_{q^m}$ and $n|(q^m-1)$. Then, $a_i=0$ for $i=0,1,\dots,n-1$ are all elements of  $\mathbb{F}_{q}$ if and only if the following
equations are satisfied:
$A_j^q=A_{qj\mod(n)}$ $j=0,1,\dots,n-1$. 
    \end{theorem}
\begin{proof}
By the definition of the FFFT of constacyclic codes, we have
    $A_j=\sum_{i=0}^{n-1}a_i(\beta\xi^j)^i$ for $j=0,1,\dots,n-1.$
    It is well-known that for a field of characteristic $p$ and any $c,d\in \mathbb{F}_q$, $(c+d)^q=c^q+d^q$. Therefore,
$$A_j^q=  \bigg(\sum_{i=0}^{n-1}a_i(\beta\xi^j)^i\bigg)^q
     =\sum_{i=0}^{n-1}a_i^q\beta^{iq}\xi^{ijq}.$$
    Since $a_i$ is an element of $\mathbb{F}_q$, $a_i^{q-1}=1$. Also, $\beta^{iq}=(\beta^q)^i=\beta^i$. Thus, we have
     $A_j^q=\sum_{i=0}^{n-1}a_i\beta^{i}\xi^{ijq}
     =\sum_{i=0}^{n-1}a_i(\beta\xi^{jq})^i=A_{qj\mod(n)}$.
     
    Conversely, assume that for all $j=0,1\dots,n-1$, $A_j^q=A_{qj\mod n}$. Then, $\sum_{i=0}^{n-1}a_i^q\beta^{iq}\xi^{ijq}=\sum_{i=0}^{n-1}a_i(\beta\xi^{jq})^i$, $j=0,1\dots,n-1$. Let $k=qj$. Since $(q,n)=1$, as $j$ ranges from $0$ to $n-1$, so $k$ also varies from $0$ to $n-1$. Therefore,
    $\sum_{i=0}^{n-1}a_i^q\beta^{i}\xi^{ik}=\sum_{i=0}^{n-1}a_i(\beta\xi^{k})^i$ for $k=0,1\dots,n-1$. Further, by the uniqueness of the FFFT, we have $a_i^q=a_i$ for all $i$. Thus, $a_i$ is a root of the polynomial $x^q-x$ for all $i$ and these roots are the elements of $\mathbb{F}_q.$
\end{proof}
For the characterization of the dual of constacyclic codes in the transform domain, we need the following property.
\begin{theorem} (Reversal Preserving Property)
Let $\boldsymbol{a}=(a_0,a_1,\dots,a_{n-1})$ and  $\boldsymbol{b}=(b_0,b_1,\dots,b_{n-1})$ be two vectors over $\mathbb{F}_q$. Let $\boldsymbol{A}=(A_0,A_1,\dots,A_{n-1})$ and  $\boldsymbol{B}=(B_0,B_1,\dots,B_{n-1})$ be their transform vectors, and $b_i=a_{n-i}$ for all $i=0,1,\dots,n-1$. Then, $B_j=A_{n-j}$ for all $j=0,1,\dots,n-1$. 
\end{theorem}
\begin{proof}
    For any $j\in\{0,1,\dots,n-1\}$, we have
    \begin{align}\label{eq5}
        B_j=&\sum_{i=0}^{n-1}b_i(\beta\xi^j)^i
        =\sum_{i=0}^{n-1}a_{n-i}(\beta\xi^j)^i
        =\sum_{k=0}^{n-1}a_k(\beta\xi^j)^{n-k}\nonumber\\
        =&\sum_{i=0}^{n-1}a_k\beta^n\beta^{-k}\xi^{-jk}
        =\lambda\sum_{k=0}^{n-1}a_k\beta^{-k}\xi^{-jk}.
    \end{align} 
    Next, we consider
    \begin{align}\label{eq6}
        A_{n-j}=&\sum_{i=0}^{n-1}a_i(\beta\xi^{n-j})^i
        =\sum_{i=0}^{n-1}a_i(\beta\xi^{-j})^i
        =\lambda\sum_{k=0}^{n-1}a_k\beta^{-k}\xi^{-jk},
    \end{align} where $n-k=i$. From equations (\ref{eq5}) and (\ref{eq6}), we get the required result.
\end{proof}
\subsection{Classical  Constacyclic BCH Codes}\label{BCH}
Binary BCH codes were discovered by Hocquenghem \cite{hocquenghem1959codes} and independently by Bose and Ray-Chaudhuri \cite{bose1960class}. In $1961$, Gorenstein and Zierler \cite{gorenstein1961class}
extended  BCH codes over finite fields. Over the past sixty years, a lot of progress on the study of BCH codes has been
made, including 2D BCH cyclic codes and their circuit architectures  \cite{mondal2020efficient,mondal2021efficient}. The goal of this paper is to extend BCH cyclic codes to BCH constacyclic codes, including their quantum versions.

Let $\xi$ be a primitive element in $\mathbb{F}_{q^m}$. The generator polynomial $g(x)$
of the $t$-error-correcting BCH code of length $q^m-1$ is the lowest-degree polynomial over $\mathbb{F}_{q^m}$ having 
$$\beta\xi,\beta\xi^2,\dots,\beta\xi^{2t}$$ as its roots. Also, $g(x)$ has $\beta\xi,\beta\xi^2,\dots,\beta\xi^{2t}$ and their conjugates as all its roots. Let $\phi_i(x)$ be the minimal polynomial of $(\beta\xi)^i$. Then g(x) must be the least common multiple of  $\phi_1(x),\phi_2(x),\dots\phi_{2t}(x)$, i.e., 
$$g(x)=\text{LCM}\{\phi_1(x),\phi_2(x),\dots,\phi_{2t}(x)\}.$$
Since $\xi^i=(\xi^{i'})^{2^l}$ is a conjugate of $\xi^{i'}$, $\xi^i$ and $\xi^{i'}$ have the same minimal polynomial. Thus, every even power of $\xi$ in the sequence above has the same minimal polynomial as some preceding odd power of $\xi$ in the sequence. Therefore,  the generator polynomial $g(x)$ of the  $t$-error-correcting BCH code of length $q^m-1$ can be reduced to 
$$g(x)=\text{LCM}\{\phi_1(x),\phi_3(x),\dots,\phi_{2t-1}(x)\}.$$
Also, let $b$ and $\delta$ be integers where $0\leq \delta \leq n$. Given those parameters, a $t$-error-correcting BCH code $\mathcal{C}$ consists of all polynomials $c(x) = c_0 + c_1x+\dots+c_{n-1}x^{n-1}\in \frac{\mathbb{F}_q[x] }{x^n-\lambda}$ such that
$c(\beta\xi^l) = 0$ for $l = b,b+1,\dots,b+\delta-2$. The sequence $\beta\xi^b,\beta\xi^{b+1},\dots,\beta\xi^{b+\delta-2}$ which consists of the roots of the underlying constacyclic BCH code, is called the consecutive root sequence of $\mathcal{C}$.

The classical constacyclic BCH code is obtained as the null space of the parity check matrix
\begin{align}
H = \begin{bmatrix}
1 & \beta\xi & (\beta\xi)^2 & \dots & (\beta\xi)^{n-1}\\
1 & \beta\xi^2 & (\beta\xi^2)^2 & \dots & (\beta\xi^2)^{(n-1)}\\
\vdots & \vdots & \vdots & \ddots & \vdots\\
1 & \beta\xi^{2t} & (\beta\xi^{2t})^{2} & \dots & (\beta\xi^{2t})^{(n-1)}\\
\end{bmatrix},\vspace{-0.15cm}\label{eqn:H_RS}
\end{align}
where $t\geq 1$ and $\delta= 2t+1$ is the designed distance of the constacyclic BCH code and $n|q^m-1$. This matrix can be reduced to the following form
\begin{align}
H = \begin{bmatrix}
1 & \beta\xi & (\beta\xi)^2 & \dots & (\beta\xi)^{n-1}\\
1 & \beta\xi^3 & (\beta\xi^3)^2 & \dots & (\beta\xi^3)^{(n-1)}\\
\vdots & \vdots & \vdots & \ddots & \vdots\\
1 & \beta\xi^{2t-1} & (\beta\xi^{2t-1})^{2} & \dots & (\beta\xi^{2t-1})^{(n-1)}\\
\end{bmatrix}.\vspace{-0.15cm}\label{eqn:H_RS}
\end{align}

 A more general form of the classical constacyclic BCH code is obtained as the null space of the parity check matrix\vspace{-0.15cm}
  \begin{align}
\!\!H_b\! =\!\!\! \begin{bmatrix}
1 & \!\beta\xi^b & \!(\beta\xi^b )^{2} & \!\!\dots & (\beta\xi^b)^{(n-1)}\!\\
1 & \!\beta\xi^{(b+1)} & \!(\beta\xi^{b+1} )^{2} & \!\!\dots & (\beta\xi^{b+1} )^{(n-1)}\!\\
\vdots & \!\vdots & \!\vdots & \!\!\ddots & \vdots\!\\
1 & \!\beta\xi^{b+\delta-2} & \!(\beta\xi^{b+\delta-2} )^{2} & \!\!\dots & (\beta\xi^{b+\delta-2} )^{(n-1)}\!
\end{bmatrix}\!\!,\label{eqn:H_b_RS}
  \end{align}
  where $b \in \{0,\dots,n-1\}$.

For the constacyclic BCH  codes obtained from parity check matrices $H$ and $H_b$ in equations \eqref{eqn:H_RS} and \eqref{eqn:H_b_RS}, respectively, the constacyclic BCH  code is over $n$ symbols. The Vandermonde matrix of $m$ rows is an $m\times m$ matrix over $\mathbb{F}_{p^{k'}}$ with the $(i,j)^{\mathrm{th}}$ element being $(\beta\xi^{j})^{(i-1)}$. Any $\delta-1$ columns of $H$ or $H_b$ form a full rank matrix as its transpose is a scaled version of a square Vandermonde matrix, which is a full rank matrix \cite{shu2011error}. As the columns of $H$ or $H_b$ comprise $\delta-1$ symbols, any $\delta$ columns are linearly dependent; hence, the minimum distance of the $t$-error-correcting constacyclic BCH code of designed distance $\delta$ is atleast $\delta$. 
\begin{theorem}[BCH bound \cite{hartmann1972generalizations}]
    Let $\mathcal{C}$ be a constacyclic BCH code of length $n$ over $\mathbb{F}_q$ with designed  distance $\delta$. Then  minimum distance $d$ of $\mathcal{C}$ is at least $\delta$, i.e., $d\geq \delta$.
\end{theorem}

We note that, when $b=1$, $H_b=H$. For classical BCH codes, the codes obtained from $H$ and $H_b$ have the same code parameters. However, in the case of quantum constacyclic BCH codes, we exploit the form of $H_b$ to choose the value of $b$ to obtain quantum codes with different code parameters. We note that $A_{b},A_{(b+1)},\dots,A_{(b+\delta-2)}$ are spectral nulls of the constacyclic BCH code because for any codeword $\boldsymbol{a}$ of the  constacyclic BCH code obtained from $H_b$, $H_b\boldsymbol{a}^{\mathrm{T}}=0$. We note that the indexing of the positions of the bits in the transformed codeword start from $1$. Thus, the positions $b+1,b+2,\dots,b+\delta-1$ form the indexed set of \emph{spectral nulls} of the code as they take value $0$ in the transform-domain  equivalent of every codeword of the constacyclic BCH  code. 
 
Using the concept of spectral nulls, the spectral-domain encoding of the constacyclic BCH  code involves inserting the message symbols in the positions that are not spectral nulls and performing the inverse $\mathrm{FFFT}$ to obtain the codeword.

The decoding and error correction procedures involve first performing $\mathrm{FFFT}$, followed by obtaining the syndrome bits from the spectral null positions $b+1,b+2,\dots,b+\delta-1$. Based on the syndrome, the error is deduced using algorithms like the Berlekamp-Massey algorithm, Euclidean algorithm, etc., and the error is subtracted from the erroneous state to obtain the codeword. In the next section, we will also provide a syndrome-less spectral decoding algorithm for constacyclic BCH codes which has less computational complexity than the well-known algorithms like PGZ and BM algorithms.

Moreover, the error correction procedure of the  constacyclic BCH code involves the following steps. Let $\boldsymbol{e}$ be the error that occurs on the codeword. Then, the erroneous codeword is $\boldsymbol{r=c+e}$. Performing $\mathrm{FFFT}$ on $\boldsymbol{r}$, we get $\mathrm{FFFT}(\boldsymbol{r}) = \mathrm{FFFT}(\boldsymbol{c+e}) = \mathrm{FFFT}(\boldsymbol{c})+\mathrm{FFFT}(\boldsymbol{e})$. We note that the elements of $\mathrm{FFFT}(\boldsymbol{r})$ and $\mathrm{FFFT}(\boldsymbol{e})$ in the positions of the spectral nulls are the same, i.e, $H_b\boldsymbol{e}^{\mathrm{T}}$, as the elements of $\mathrm{FFFT}(\boldsymbol{c})$ in the spectral null positions are zero. Therefore, the elements of $\mathrm{FFFT}(\boldsymbol{r})$ in the positions of the spectral nulls provide the syndrome. Based on the syndrome obtained, the error can be deduced, and the inverse error operation is performed to obtain the codeword $\boldsymbol{c}$.

\section{The Spectral Decoding Algorithm }\label{sec 3}
 To study constacyclic error-correcting codes over finite fields in the frequency domain we have to deal with polynomials in the quotient ring $\frac{\mathbb{F}_{q^m}[x]}{\langle x^n-\lambda\rangle}$. Thus, we associate the polynomial $a(x)=\sum_{i=0}^{n-1}a_ix^i$ with a polynomial $A(x)=\sum_{j=0}^{n-1}A_jx^j$ by means of the finite field Fourier transform. The polynomial $A(x)$ is known as  the
spectrum polynomial or the associated polynomial of $a(x).$ The following useful result relates the roots of these polynomials to the properties of the spectrum. 
\begin{theorem}\label{th6}
\begin{enumerate}
    \item The polynomial $a(x)$ has a root at $\beta \xi^j$ if and only if the $j^{\mathrm{th}}$ spectral component $A_j$ equals zero.
    \item  The polynomial $A(x)$ has a root at $ \xi^{-i}$ if and only if the $i^{\mathrm{th}}$ time component $a_i$ equals zero.
\end{enumerate}    
\end{theorem}
\begin{proof}
\begin{enumerate}
    \item Since $a(x)=\sum_{i=0}^{n-1}a_ix^i$,        \begin{align}\label{eq7}
        a(\beta \xi^j)=\sum_{i=0}^{n-1}a_i(\beta\xi^j)^i=A_j.
    \end{align} From equation (\ref{eq7}), we get the required result.
    \item Since $A(x)=\sum_{j=0}^{n-1}A_jx^j$, $ A( \xi^{-i})=\sum_{j=0}^{n-1}A_j(\xi^{-i})^j=n\beta^ia_i$, proving the required result.
\end{enumerate}
    \end{proof}
    
Let $\mathcal{C}$ be an $(n,k)$ constacyclic BCH code, where $n$ is the blocklength of the code and $k$ is the dimension of the code. The spectrum of this code lies in $\mathbb{F}_{q^m}$ with $n|(q^m-1)$. For the special case with a \textit{design} minimum distance $d=2t+1$, we consider $g(x)$ to be a generator polynomial with $2t$ roots $\{\beta\xi,\beta\xi^2\dots,\beta\xi^{2t}\}$. We assume that all the polynomials associated with the coding scheme are in the transform domain. 

Let $R$ be the received vector\footnote{The bold letters indicate that the vectors/polynomials are in the spectral domain.} and its corresponding polynomial is represented as
$R(x)=\sum_{j=0}^{n-1}R_jx^j=C(x)+E(x)=\sum_{j=0}^{n-1}C_jx^j+\sum_{j=0}^{n-1}E_jx^j$, where $C$ is the codeword and $E$ is the error vector. The $j^{\mathrm{th}}$ error in the error vector  $E$ has a locator $\sigma_j\in \{\beta\xi,\beta\xi^2,\dots,\beta\xi^{n-1}\}$.  The error locator polynomial is
$\Gamma(x)=\prod_{j=1}^{\tau}(1-x\beta\xi^j)$, where $\tau$ is the actual number of errors with $\tau\leq t$. If there are no errors, then $\Gamma(x)=1$.
  
Let $C(x) = \sum_{j=0}^{n-1}C_jx^j=\text{FFFT}(c(x))$ 
be a code
polynomial of the constacyclic BCH code~$C$, where
$C_j\in\mathbb{F}_{q^m}$, $j = 0,1,\dots,n-1$, and the conjugacy
constraints are satisfied according to Theorem \ref{th3}.  Let $m(x) = \sum_{i=0}^{k-1}m_ix^i$ be the message polynomial corresponding to the systematic constacyclic code and $\deg (m(x))<k$. 

Recall that $C(x)=\text{FFFT}(c(x))$ and $C(x)=\sum_{j=0}^{n-1}C_jx^j$. Also,  $c_i=\frac{1}{n\beta^i}\sum_{j=0}^{n-1}\xi^{-ij}C_j$ and $c(x)=\sum_{i=0}^{n-1}c_ix^i$. Then, $c(\beta\xi^j)=\sum_{i=0}^{n-1}c_i(\beta\xi^j)^i=C_j$, $j=0,1,\dots,n-1$. The message polynomial $m(x)$ and its transformed version $M(x)$ can be obtained from the first $k$ components of $c(x)$ or $C(x)$, respectively as we consider a systematic code. Since $m_j=0$ for $j=k,..,n-1$, we have for $j=0,1,...,n-1$
\begin{equation}\label{eq8}
	\begin{cases}
R_j=C_j=m(\beta\xi^j), & \mbox if ~R_j=C_j\\
\Gamma(\beta\xi^j)=0, & \mbox if~ R_j\neq C_j.
	\end{cases}
\end{equation} From equation (\ref{eq8}) it follows that 
$\Gamma(\beta\xi^j)R_j=\Gamma(\beta\xi^j)m(\beta\xi^j)$, $j=0,1,\dots,n-1$. Next, let $P(x)=\Gamma(x)m(x)$. Then 
$\Gamma(\beta\xi^j)R_j=P(\beta\xi^j)$, $j=0,1,\dots,n-1$. Now, we will construct an interpolating polynomial $\mu(x)=\text{IFFFT}(R(x))$ such that 
$\mu(\beta\xi^j)=R_j$ (from Theorem \ref{th6}) where $\deg (\mu(x)) < n$. Also, 
\begin{align*}
\Gamma(\beta\xi^j)\mu(\beta\xi^j)=&P(\beta\xi^j), j=0,1,\dots,n-1\\
    \Gamma(x)\mu(x)-P(x)=&(x-\beta\xi^j)q_j(x), j=0,1,\dots,n-1\\
     \Gamma(x)\mu(x)-P(x)=&(x^n-\lambda)q(x)
\end{align*} where $q_j(x)$ and $q(x)$ are quotient polynomials. Finally, we obtain the key equation as follows:
\begin{equation}
	\begin{cases}
		\Gamma(x)\mu(x)\cong P(x) \mod (x^n-\lambda)\\
  \deg(\Gamma(x)) \leq t\\
  \text{maximize} \deg (\Gamma(x)).
  
	\end{cases}
\end{equation} Since $\deg (P(x))= \deg (m(x))+ \deg (\Gamma(x))\leq k-1+t<n-t $, we have 
\begin{equation}\label{eq9}
	\begin{cases}
\Gamma(x)\mu(x)\cong P(x) \mod (x^n-\lambda)\\
  \deg (P(x))< n-t\\
  \text{maximize} \deg(P(x)).
  \end{cases}
\end{equation} 
Next, we solve the key equation by using  the extended
Euclidean algorithm for polynomials (EEAP) to $x^n-\lambda$ and
$\mu(x)$, and we obtain polynomials $P(x)$ and $\Gamma(x)$. Finally, we get the message polynomial $m(x)=\frac{P(x)}{\Gamma(x)}$. Post zero-padding $m(x)$ and taking FFFT, we get the spectral form of the message polynomial $D(x)=\text{FFFT}(m(x))$.
\begin{algorithm}
\begin{algorithmic}
 \caption{\textbf{Spectral decoding algorithm for systematic $\lambda$-constacyclic codes of length $n$}}
 \label{algo:1}
 \renewcommand{\algorithmicrequire}{\textbf{Input:}}
 \renewcommand{\algorithmicensure}{\textbf{Output:}}
 \REQUIRE 
  ~\\The received vector $R(x)=\sum_{j=0}^{n-1}R_jx^j\longleftrightarrow R=(R_0,R_1,\dots,R_{n-1})$.\\
  
 \ENSURE ~\\ The decoded message polynomial in the spectral domain $D(x)=\sum_{i=0}^{n-1}D_jx^j\longleftrightarrow D=(D_0,D_1,\dots,D_{n-1})$.\\
\textbf{Step 1:} $\mu=\text{IFFFT}(R)\longleftrightarrow \mu(x)$.\\

\textbf{Step 2:}  Solve the congruence such that $\mu(\beta\xi^j)=R_j$ for $j=0,1,...,n-1$.
\begin{equation*}
	\begin{cases}
		\Gamma(x)\mu(x)\cong P(x) \mod x^n-\lambda\\
  \deg P(x) < n-t\\
  \text{maximize} \deg \text{of }  P(x).
  \end{cases}
  \end{equation*}\\
\textbf{Step 3:} $m(x)=\frac{P(x)}{\Gamma(x)}$, where $m(\beta\xi^j)=C_j$ valid for $j=0,1,...,\le n-2t-1$ and $t$ is the $\deg(\Gamma(x))$ based on the actual number of errors. If $\Gamma(x) \nmid P(x)$, then report a decoding failure. \\
\textbf{Step 4:} $D=\text{FFFT}(m(x))$ post zero-padding $m(x)$ from Step~3 and taking the FFFT.\\
\end{algorithmic}
\end{algorithm}

The following theorem provides the correctness of the proposed algorithm.
\begin{theorem}
    For decoding up to the designed error-correcting capability $t$ for the constacyclic BCH code the decoding algorithm produces a unique decoded polynomial $D(x)$.
\end{theorem}
\begin{proof}
     We have to show that the key equation (\ref{eq9}) has a unique solution which produces a unique decoded polynomial $D(x)$. Suppose we obtain the polynomials $P'(x)$ and $\Gamma'(x)$, from equation (\ref{eq9}) by applying EEAP to $x^n-\lambda$ and $\mu(x)$, which terminates when  $\deg(P'(x))<n-t$. If $P'(x)$ is divisible by $\Gamma'(x)$, then we get the message polynomial $m'(x)$ which leads to the decoded polynomial $D'(x)$. 
     
     Now, there are $2$ cases.\\
\textbf{Case 1:} When the received vector $R(x)$ is zero, then $\mu(x)=0$. Now, by applying EEAP, we get $(x^n-\lambda)0+\mu(x)1=0$. This gives $\Gamma'(x)=1$ and $P'(x)=m'(x)=D'(x)=0$. Thus, $D(x)=D'(x)$.\\
\textbf{Case 2:} When there are errors, i.e., $E(x)\neq 0$ and $R(x)\neq 0$. First, we consider
\begin{align*}
P(x)\Gamma'(x)\cong&(\Gamma(x)\mu(x))\Gamma'(x)\\
\cong&\Gamma(x)(\mu(x)\Gamma'(x))\\
\cong&\Gamma(x)P'(x)~\mod(x^n-\lambda) 
\end{align*}
 Since the degree of each side of this congruence is less than $n$, $P(x)\Gamma'(x)=\Gamma(x)P'(x)$. From this, we get $P'(x)=\frac{P(x)}{\Gamma(x)}\Gamma'(x)=m(x)\Gamma'(x)$ and $m(x)=\frac{P'(x)}{\Gamma'(x)}=m'(x)$; hence, $D(x)=D'(x)$, proving the required result.
\end{proof}
 \subsection{Computational Complexity of the  Decoding Algorithm } \label{sub3}
 In this subsection, we provide the complexity of our proposed spectral decoding algorithm in terms of the number
of arithmetic operations in the computation field $\mathbb{F}_{q^m}$. \\
\textbf{The spectral decoding algorithm:}
\begin{enumerate}
    \item In Steps 1 and 4, we have to multiply a vector by a matrix to find the FFFT and IFFFT of a vector. So, the total number of arithmetic operations required in this step is $2n^2-2n$ (Mult: $n(n-1)$, Add: $n(n-1)$).
    \item In Step 2, to solve the key equation we use the extended Euclidean algorithm for polynomials~\cite{moenck1973fast}. The the total number of arithmetic operations required is $t(4t+4n)=4t^2+4tn$ (Mult: $t(2t+2n)$, Add: $t(2t+2n)$).
    \item In Step 3, to find the message polynomial  we use the fast division algorithm for polynomials~\cite{von2013modern}. The the total number of arithmetic operations required in this step is $(n-2t)2t$ (Mult: $t(n-2t)$, Add: $t(n-2t)$).
\end{enumerate} 
Hence, our proposed spectral decoding algorithm for $\lambda$-constacyclic codes requires $O_{Spec}=2n^2-2n+4t^2+4tn+2tn-4t^2+2n^2-2n=4n^2-4n+6tn$ arithmetic operations in the computational field $\mathbb{F}_{q^m}$. Also, the total no. of multiplications required for our method are $O_{Spec}^{M}=2n^2-2n+3tn$. Further, we compare our proposed algorithm with the syndrome-based decoding algorithms. In syndrome-based decoding algorithm (when Berlekamp-Massey (BM) algorithm for finding the error-location polynomial is used) the complexity \cite{chen2008complexity}  is about $O_{Syn}=6nt+16t^2-n+6t$ arithmetic operations and $O_{Syn}^{M}=3tn+10t^2-n+6t$. Also, the complexity of
the classical Peterson–Gorenstein–Zierler (PGZ) decoding algorithm [\cite{blahut1983theory}, Fig. 7.1] is about
$O_{PGZ}=6tn+\frac{1}{2}t^4+\frac{16}{3}t^3+5t^2+\frac{1}{6}t$ arithmetic operations. Table $\ref{tabl2}$ presents the complexity in terms of the number of arithmetic operations in the computation field $\mathbb{F}_{q^m}$ of our proposed spectral decoding algorithm and the classical decoding algorithms.
 \subsection{Numerical illustration}\label{sub32}
 In this subsection, first we will provide an example to check that $\lambda$-constacyclic codes have a better minimum distance, useful for error-correction. 
  \begin{example}
     Let $\mathbb{F}_9$ be a field of order $9$ where $\mathbb{F}_9=\mathbb{F}_3(w)$ and $w^2+2w+2=0$. Let $n=50$ and $\lambda=w^5$. In $\mathbb{F}_9[x]$, we have 
     \begin{align*}
        x^n-\lambda = x^{50}-w^5=& (x^2+w)(x^2+x+w)(x^2+w^3x+w)(x^2+2x+w)(x^2+w^7x+w)\\&(x^{10}+wx+w^5)(x^{10}+w^2x+w^5)(x^{10}+w^5x+w^5)(x^{10}+w^6x+w^5).
     \end{align*}
 Now, let $g(x)=(x^2+w)(x^2+x+w)(x^2+w^3x+w)(x^2+2x+w)(x^2+w^7x+w)$ be the generating polynomial of constacyclic code, i.e., $C=\langle g(x) \rangle$ be a constacyclic code of length $50$ over $\mathbb{F}_9$ having parameters $[50,40,2]$. 
 
 If we take $g(x)=(x^{10}+wx+w^5)~\text{or} ~(x^{10}+w^2x+w^5)~\text{or} ~(x^{10}+w^5x+w^5)~\text{or} ~(x^{10}+w^6x+w^5)$, then corresponding code parameters are $[50,40,3]$. 
 
Further, in the case of cyclic code we have 
\begin{align*}
    x^{50}-1=&(x+1)(x+2)(x^2+wx+1)(x^2+w^3x+1)(x^2+w^5x+1)(x^2+w^7x+1)(x^{10}+wx^5+1)\\& (x^{10}+w^3x^5+1) (x^{10}+w^5x^5+1) (x^{10}+w^7x^5+1).
\end{align*}
    Here, either we take $g(x)=(x+1)(x+2)(x^2+wx+1)(x^2+w^3x+1)(x^2+w^5x+1)(x^2+w^7x+1)~\text{or} ~(x^{10}+wx^5+1)~\text{or} ~ (x^{10}+w^3x^5+1) ~\text{or} ~(x^{10}+w^5x^5+1) ~\text{or} ~(x^{10}+w^7x^5+1)$.  In every case code parameters are $[50,40,2]$.
    
    Hence, we can get improved minimum distance using constacyclic codes than cyclic codes for a given length $n$ and dimension $k$. 
 \end{example}
 \begin{example}\label{exam2}
Let $\mathbb{F}_{27}$ be a field of order $27$ where $\mathbb{F}_{27}=\mathbb{F}_3(w)$ and $w^3+2w+1=0$. Let $n=13$, $\beta=-1$, $\lambda=w^{13}$ and $\xi=w^2$. Let $C$ be a $(13,10)$ constacyclic code over $\mathbb{F}_{3}$, where $n=13$ is the
blocklength of the code and $k = 10$ is the dimension of the code. The spectrum of this code lies in the extension field $\mathbb{F}_{27}$. In $\mathbb{F}_{27}[x]$, we have
\begin{align*}
    x^{13}-w^{13}=& (x+1)(x+w^2)(x+w^4)(x+w^6)(x+w^8)(x+w^{10})(x+w^{12})(x+w^{14})\\&(x+w^{16})(x+w^{18})(x+w^{20})(x+w^{22})(x+w^{24}).
\end{align*}  Let the
designed error-correction capability for this code be $t=1$. Let the received vector be $R=(0,0,0,0,0,0,0,0,0,w,0,0,0)$.
\FloatBarrier

		\begin{table}[htbp]
  \renewcommand{\arraystretch}{1}
			\caption{ The number of arithmetic operations for spectral decoding algorithm
for some constacyclic codes}
\begin{center}

			\begin{tabular}
          {|c|c|c|c|c|c|}
					\hline
				$q^m$ & $\lambda$& $n$  & $t$ &  Classical algorithms & Spectral algorithm\\
					\hline
					$27$ & $w^{14}$ & $2$ & $1$ & $32~~ (O_{Syn})$& $20~~(O_{Spec})$\\
					 \hline
      $125$ & $w^{62}$ & $62$ & $28$ & $13154
 ~~ (O_{{Syn}}^{M})$& $12772~~ (O_{Spec}^{M})$\\
      \hline
      $125$ & $w^{62}$ & $62$ & $29$ & $13916
 ~~ (O_{{Syn}}^{M})$& $12958~~ (O_{Spec}^{M})$\\
      \hline
      
      $125$ & $w^{62}$ & $62$ & $13$ & $31681
 ~~ (O_{PGZ})$& $19964~~(O_{Spec})$\\
					 \hline
      $125$ & $w^{62}$ & $62$ & $15$ & $50020~~ (O_{PGZ})$& $20708~~(O_{Spec})$\\
					 \hline
            $625$ & $w^{312}$ & $156$ & $70$ & $82024
 ~~ (O_{Syn}^{M})$& $81120~~ (O_{Spec}^{M})$\\
      \hline
            $625$ & $w^{312}$ & $156$ & $29$ & $91644
 ~~ (O_{Syn}^{M})$& $61932~~ (O_{Spec}^{M})$\\
      \hline
				\end{tabular}
				\label{tabl2}
   \end{center}
		\end{table}

		\FloatBarrier

\subsection*{Computational evaluations of the spectral decoding algorithm:}
\begin{enumerate}
    \item $\mu=\mathrm{IFFFT}(R)=(0,0,0,0,0,0,0,0,0,w,0,0,0)$
    $\times$$\begin{pmatrix}
        \beta^{-0}&\beta^{-1}&\cdots &\beta^{-12}\\
        (\beta^{-1}\xi^{-1})^0& (\beta^{-1}\xi^{-1})^1&\cdots&(\beta^{-1}\xi^{-1})^{12}\\
        \vdots & \vdots& \cdots& \vdots\\
         (\beta^{-1}\xi^{-12})^0& (\beta^{-1}\xi^{-12})^1&\cdots&(\beta^{-1}\xi^{-12})^{12}
    \end{pmatrix}\\=(w,w^{22},w^{17},w^{12},w^7,w^2,w^{23},w^{18},w^{13},w^8,w^3,w^{24},w^{19})\longleftrightarrow w+w^{22}x+w^{17}x^2+w^{12}x^3+w^7x^4+w^2x^5+w^{23}x^6+w^{18}x^7+w^{13}x^8+w^8x^{9}+w^3x^{10}+w^{24}x^{11}+w^{19}x^{12}=\mu(x)$.
    \item Solve the key equation:
    $\Gamma(x)\mu(x)\cong P(x) \mod x^n-\lambda$.\\
    Step $\mathrm{EEAP}~0:$ $(x^n-w^{13})0+\mu(x)1=\mu(x)$\\
    Step $\mathrm{EEAP}~1:$ $(x^n-w^{13})=\mu(x)(w^7x+w^{25})$. \\
    Now, we get $\Gamma(x)=w^7x+w^{25}$ and $P(x)=0(x)$, and $(x^n-w^{13})1-\mu(x)\Gamma(x)=-P(x)$.
    \item $m(x)=\frac{P(x)}{\Gamma(x)}=0(x)$
    $\longleftrightarrow (0,0,0,0,0,0,0,0,0,0,0,0,0)=m.$
    \item $D=\mathrm{FFFT}(m(x))=(0,0,0,0,0,0,0,0,0,0,0,0,0).$
\end{enumerate}
Hence, the transmitted vector was the zero vector.
 \end{example}
 \begin{example}\label{exam3}
Let $\mathbb{F}_{9}$ be a field of order $9$ where $\mathbb{F}_{9}=\mathbb{F}_3(w)$ and $w^2+2w+2=0$. Let $n=4$, $\beta=-w$, $\lambda=w^{4}$ and $\xi=w^2$. Let $C$ be a constacyclic BCH code of blocklength $4$ over $\mathbb{F}_{3}$, the spectrum of this code lies in the extension field $\mathbb{F}_{9}$. In $\mathbb{F}_{9}[x]$, we have
\begin{align*}
    x^{4}-w^{4}=& (x+w)(x+w^3)(x+w^5)(x+w^7).
\end{align*} Let the received vector be $R=(0,1,1,1)$.
\subsection*{Computational evaluations of the spectral decoding algorithm:}
\allowdisplaybreaks
\begin{enumerate}
    \item\begin{align*}
       \mu=\mathrm{IFFFT}(R)=&(0,1,1,1)\times \begin{pmatrix}
        \beta^{-0}&\beta^{-1}&\cdots &\beta^{-3}\\
        (\beta^{-1}\xi^{-1})^0& (\beta^{-1}\xi^{-1})^1&\cdots&(\beta^{-1}\xi^{-1})^{3}\\
        \vdots & \vdots& \cdots& \vdots\\
         (\beta^{-1}\xi^{-3})^0& (\beta^{-1}\xi^{-3})^1&\cdots&(\beta^{-1}\xi^{-3})^{3}
    \end{pmatrix}\\=&
    (0,1,1,1)
    \times\begin{pmatrix}
       1&w^3&w^6&w\\
       1&w&w^2&w^3\\
       1&w^7&w^6&w^5\\
       1&w^5&w^2&w^7\\
    \end{pmatrix}\\=&(0,w^7,w^2,w^5)\longleftrightarrow w^7x+w^2x^2+w^5x^3=\mu(x).
    \end{align*}
    \item Solve the key equation:
    $\Gamma(x)\mu(x)\cong P(x) \mod x^n-\lambda$.\\
    Step $\mathrm{EEAP}~0:$ $(x^4-w^{4})0+\mu(x)1=\mu(x)$\\
    Step $\mathrm{EEAP}~1:$ $(x^4-w^{4})=\mu(x)(w^3x+2)+w^7x+1$.\\ Now, we get $\Gamma(x)=w^3x+2$ and $P(x)=-(w^7x+1)=w^3x+2$, and $(x^4-w^{4})1-\mu(x)\Gamma(x)=-P(x)$.
    \item $m(x)=\frac{P(x)}{\Gamma(x)}=1$
    $\longleftrightarrow (1,0,0,0)=m.$
    \item \begin{align*}
        D=\mathrm{FFFT}(m(x))=&(1,0,0,0)\times 
    \begin{pmatrix}
        \beta^{0}&(\beta\xi^{1})^0&\cdots &(\beta\xi^{3})^0\\
         \beta^{1}&(\beta\xi^{1})^1&\cdots &(\beta\xi^{3})^1\\
        \vdots & \vdots& \cdots& \vdots\\
          \beta^{3}&(\beta\xi^{1})^3&\cdots &(\beta\xi^{3})^3\\
    \end{pmatrix}\\=&
    (1,0,0,0)
    \times\begin{pmatrix}
       1&1&1&1\\
       -w&-w^3&-w^5&-w^7\\
       w^2&w^6&w^2&w^6\\
       -w^3&-w&-w^7&-w^5\\
    \end{pmatrix}\\=&(1,1,1,1).
    \end{align*}
\end{enumerate}
Hence, the transmitted vector was $(1,1,1,1)$.
 \end{example}
 \section{Quantum BCH code in the spectral domain}\label{sec4}
In this section, we first briefly review the mathematical representation of the quantum states, the operators acting on these states, and then we define quantum constacyclic BCH codes in the spectral domain. Throughout the paper, we consider a finite field $\mathbb{F}_q$ of order $q=p^s$ and its extension field of order $q^m=p^{ms}=p^{k'}$.
\subsection{Quantum states and operators over qudits}
\label{sub41}
 For a quantum system with $\zeta$ levels, the state of a unit system, a qudit, is a superposition of $\zeta$ basis states of the system given by  
 \begin{align*}
\ket{\psi}_\zeta &= \sum_{i=0}^{\zeta-1} a_i\ket{i}_\zeta, ~~ \text{ where } a_i\in \mathbb{C} \text{ and } \sum_{i=0}^{\zeta-1} |a_i|^2 = 1, 
 \end{align*}
where the subscript $\zeta$ refers to dimension of the unit quantum system.  Also, $\ket{\psi}_\zeta=[a_0~a_1~\dots ~a_{\zeta-1}]^{\mathrm{T}}$ and  $\ket{i}_\zeta = \mathbf{e}_{{(i+1)}}^{{(\zeta)}}$, where $\mathbf{e}_{{(i+1)}}^{{(\zeta)}}$ is a vector in $\mathbb{C}^\zeta$ with the $(i+1)^{\mathrm{st}}$ element being $1$ and rest of the elements being $0$.

From the second postulate of quantum mechanics, the operators acting on a quantum system belong to the unitary group $\mathrm{U}(\zeta)$, which is a subset of $\mathbb{C}^{\zeta \times \zeta}$. As the cardinality of $\mathrm{U}(\zeta)$ is infinite, we represent its elements in terms of a basis of $\mathbb{C}^{\zeta \times \zeta}$.

The generalized version of the Pauli group for arbitrary $\zeta$, known as the Weyl-Heisenberg group, is defined by
\begin{align}  
 \mathcal{G}_{\zeta}^{(g)} = \left\{\omega_\zeta^l\mathrm{X}_\zeta(a)\mathrm{Z}_\zeta(b)|a,b,l \in\mathbb{Z}_\zeta\right\},\nonumber
\end{align}
 where $\omega_\zeta \!=\! \mathrm{e}^{\!\frac{\mathrm{i}2\pi}{\zeta}}$, $\mathrm{X}_\zeta(a)\!\ket{c}_\zeta \!:=\!\! \ket{(a\!+\!c)~\mathrm{mod}~\zeta}_\zeta$, and $\mathrm{Z}_\zeta(b)\!\ket{c}_\zeta := \omega_\zeta^{bc}\!\ket{c}_\zeta$ for every $c \in \mathbb{Z}_\zeta$.
 The generalized Pauli basis $\mathcal{P}$ \cite{mielnik1968geometry}
  \begin{align}
 \mathcal{G}_\zeta = \left\{\mathrm{X}_\zeta(a)\mathrm{Z}_\zeta(b)|a,b \in\mathbb{Z}_\zeta\right\}.\label{eqn:Gen_Pauli_Basis}
\end{align}
 is obtained by neglecting the phase $\omega_\zeta^l$ in $\mathcal{G}_\zeta$. The basis operator of the form $\mathrm{X}_\zeta(a)\mathrm{Z}_\zeta(b)$ is uniquely represented by a vector of length $2$ defined over ring $\mathbb{Z}_\zeta$, namely $[a|b]_\zeta$ as follows:
\begin{align*}
\mathrm{X}_\zeta(a)\mathrm{Z}_\zeta(b) \equiv  [a|b]_\zeta
\end{align*}

Next, we define a trace operation over the field elements as follows: 
\begin{definition}[\cite{ketkar2006nonbinary}]
The field trace $\mathrm{Tr}_{p^{k'}/p}(\cdot)$ is an $\mathbb{F}_p$-linear function $\mathrm{Tr}_{p^{k'}/p}: \mathbb{F}_{p^{k'}} \rightarrow \mathbb{F}_{p}$, given by $\mathrm{Tr}_{p^{k'}/p}(\alpha)=\sum_{i=0}^{k'-1}\alpha^{p^i}$, where $\alpha \in \mathbb{F}_{p^{k'}}$.
\end{definition}
The function $\mathrm{Tr}_{p^{k'}/p}(\cdot)$ is said to be $\mathbb{F}_p$-linear as $\mathrm{Tr}_{p^{k'}/p}(a\alpha+b\gamma) = a\hspace{0.07cm}\mathrm{Tr}_{p^{k'}/p}(\alpha)+b\hspace{0.07cm}\mathrm{Tr}_{p^{k'}/p}(\gamma)$, for all $a,b \in \mathbb{F}_p$ and $\alpha, \gamma \in \mathbb{F}_{p^{k'}}$. We note that for an element $b \in \mathbb{F}_p$, $\mathrm{Tr}_{p^{k'}/p}(b) = b$.

The group that generates the operator basis for $\mathbb{C}^{p^{k'} \times p^{k'}}$ defined in terms of the field based representation of basis states is \cite{lidar2013quantum}
 \begin{align}
 \mathcal{G}_{p^{k'}}^{(g)} \!=\! \begin{cases} \left\{\omega^l\mathrm{X}^{(p^{k'})}(\alpha)\mathrm{Z}^{(p^{k'})}(\gamma)\bigg|\alpha,\!\gamma \in \mathbb{F}_{p^{k'}} \text{ and }l \!\in\!\mathbb{Z}_p\right\}, & \text{when characteristic $p$ is odd},\\
 \left\{\mathrm{i}^{g}\omega^l\mathrm{X}^{(p^{k'})}(\alpha)\mathrm{Z}^{(p^{k'})}(\gamma)\bigg|\alpha,\!\gamma \in \mathbb{F}_{p^{k'}} \text{ and }g,l \!\in\!\mathbb{Z}_p\right\}, & \text{when characteristic $p$ is even},\\
 \end{cases} \nonumber
\end{align}
where $\omega = \mathrm{e}^{\frac{\mathrm{i}2\pi}{p}}$, $\mathrm{i}=\sqrt{-1}$,
\begin{align}
\mathrm{X}^{(p^{k'})}(\alpha)\ket{\theta}_{p^{k'}} &:= \ket{\alpha+\theta}_{p^{k'}},~~~~~\forall \theta \in \mathbb{F}_{p^{k'}},\label{eqn:X_Definition_Qudit}\\
\mathrm{Z}^{(p^{k'})}(\gamma)\ket{\theta}_{p^{k'}} &:= \omega^{\mathrm{Tr}_{p^{k'}/p}(\gamma\theta)}\ket{\theta}_{p^{k'}},~~~\forall \theta \in \mathbb{F}_{p^{k'}}. \label{eqn:Z_Definition_Qudit}
\end{align}  
We note that the factor $\mathrm{i}^{g}$ is included in the basis $\mathcal{G}_{p^{k'}}$ when the characteristic $p$ is even as $\mathrm{iI}$ belongs to $\mathcal{P}$ and $\mathcal{P} = \mathcal{G}_{p^{k'}}$ for $p=2$ and $k'=1$.

The operator basis for $\mathbb{C}^{p^{k'} \times p^{k'}}$ is
\begin{align}
 \mathcal{G}_{p^{k'}} \!=\! \left\{\mathrm{X}^{(p^{k'})}(\alpha)\mathrm{Z}^{(p^{k'})}(\gamma)\bigg|\alpha,\!\gamma \in \mathbb{F}_{p^{k'}}\right\}, \label{eqn:FieldErrorBasis}
\end{align}

 From equations \eqref{eqn:X_Definition_Qudit} and \eqref{eqn:Z_Definition_Qudit}, $\mathrm{X}^{(p^{k'})}(\alpha)$ and $\mathrm{Z}^{(p^{k'})}(\gamma)$ are given by
 \begin{align}
 \mathrm{X}^{(p^{k'})}(\alpha) &= \underset{\theta \in \mathbb{F}_{p^{k'}}}{\sum}\ket{\alpha + \theta}\bra{\theta}\label{eqn:X_Expression_Qudit},\\
  \mathrm{Z}^{(p^{k'})}(\gamma) &= \underset{\theta \in \mathbb{F}_{p^{k'}}}{\sum}\omega^{\mathrm{Tr}_{p^{k'}/p}(\gamma\theta)}\ket{\theta}\bra{\theta}.\label{eqn:Z_Expression_Qudit}
 \end{align}
 The basis operator of the form $\mathrm{X}^{(p^{k'})}(\alpha)\mathrm{Z}^{(p^{k'})}(\gamma)$ is uniquely represented by a vector of length $2$ defined over field $\mathbb{F}_{p^{k'}}$, namely $[\alpha|\gamma]_{p^{k'}}$.
\begin{align*}
\mathrm{X}^{(p^{k'})}(\alpha)\mathrm{Z}^{(p^{k'})}(\gamma) \equiv  [\alpha|\gamma]_{p^{k'}}
\end{align*} 

Calderbank, Shor, and Steane \cite{calderbank1996good} \cite{steane1996error} proposed a framework to construct quantum error correction codes over qubits from two classical binary codes $C_1$ and $C_2$ that satisfy $C_1^{\perp} \subset C_2$. This class of codes are called the CSS codes. The condition $C_1^{\perp} \subset C_2$ is called the \textit{dual-containing condition} of CSS codes. 
By considering $C_1$ and $C_2$ to be the same code, i.e., $C_1=C_2$, we can construct quantum codes from dual-containing classical codes as $C_1^{\perp} \subset C_2 = C_1$.

The CSS codes form a class of stabilizer codes. Let $H_1$ and $H_2$ be the parity check matrices of the classical codes $C_1[n,k_1,d_1]$ and $C_2[n,k_2,d_2]$, respectively. As $C_1^{\perp} \subset C_2$, the elements of $C_1^{\perp}$ are codewords of $C_2$; hence, $H_2H_1^{\mathrm{T}} = {0}$. Now, using the CSS construction over qudits, we have the following result
\begin{theorem}[\cite{ketkar2006nonbinary}]
    Let $\mathcal{C}_1$ and $\mathcal{C}_2$  be two classical codes over finite field $\mathbb{F}_q$ having parameters $[n, k_1, d_1]$ and $[n, k_2, d_2]$ with
parity check matrices $H_1$ and $H_2$, respectively, with $\mathcal{C}_1^{\perp}\subset \mathcal{C}_2$. Then, there exists an
$[[n, k_1 + k_2 - n, d\geq \mathrm{min}(d_1, d_2)]]$ quantum stabilizer code.
\end{theorem}
 The following result is helpful towards obtaining the check matrix of the CSS code.
\begin{lemma} [\cite{nadkarni2021entanglement}]\label{lem:Trace_FieldElementProdBasis_Zero}
For $\gamma \in \mathbb{F}_{p^{k'}}$, $\mathrm{Tr}_{p^{k'}/p}(\xi^i \gamma) = 0$ $\forall i \in \{0,1,\dots ,k'-1\}$ if and only if $\gamma = 0$. 
\end{lemma}
\begin{lemma}[\cite{nadkarni2021entanglement}]\label{lem3}
   For $\gamma\in \mathbb{F}_{p^{k'}}$, $\underset{\theta\in \mathbb{F}_{p^{k'}}}\sum\omega^{\mathrm{Tr}_{p^{k'}/p}(\gamma\theta)}=\delta_{\gamma,0}p^{k'}$. 
\end{lemma}
By using Lemma \ref{lem:Trace_FieldElementProdBasis_Zero}, Nadkarni and Garani \cite{nadkarni2021entanglement} proposed the check matrix of the CSS code which is given below.
\begin{theorem} (\cite{nadkarni2021entanglement})\label{CheckCSS}
 The check matrix for the CSS code obtained from $\mathcal{C}_1$ and $\mathcal{C}_2$ that satisfy $\mathcal{C}_1^{\perp}\subset \mathcal{C}_2$, whose basis codewords are provided in Definition \ref{def5}, is given by,\vspace{-0.15cm}
 \begin{align*}
\mathcal{H}_{\mathrm{CSS}}^{(p^{k'})} =  \left[\begin{array}{c|c}
\begin{matrix}
H_{d_1}\\\xi H_{d_1} \\ \vdots \\ \xi^{k'-1}H_{d_1}
\end{matrix} & \mathbf{0}\\
\mathbf{0} & \begin{matrix}
H_{d_2}\\\xi H_{d_2} \\ \vdots \\ \xi^{k'-1}H_{d_2}
\end{matrix}
\end{array}\right],
 \end{align*}
 where $\xi$ is the primitive element of $\mathbb{F}_{p^{k'}}$. 
\end{theorem}
Next, we consider two classical constacyclic BCH codes $\mathcal{C}_1$ and $\mathcal{C}_2$ over $\mathbb{F}_{q^m}$ with parity check matrices $H_{b_1}$ and $H_{b_2}$ as defined in equation \eqref{eqn:H_b_RS}, where $n|q^m-1$, $b_i\in \{0,1,\dots,n-1\}$, $i\in \{1,2\}$. Then from Theorem \ref{CheckCSS}, we obtain the matrix $\mathcal{H}_{\mathrm{BCH}}$ 
 \begin{align}\label{eqHBCH}
\mathcal{H}_{\mathrm{BCH}} =  \left[\begin{array}{c|c}
\begin{matrix}
H_{b_1}\\\xi H_{b_1} \\ \vdots \\ \xi^{k'-1}H_{b_1}
\end{matrix} & \mathbf{0}\\
\mathbf{0} & \begin{matrix}
H_{b_2}\\\xi H_{b_2} \\ \vdots \\ \xi^{k'-1}H_{b_2}
\end{matrix}
\end{array}\right],
 \end{align}
 where $\xi$ is the primitive element of $\mathbb{F}_{q^m}$ and $q^m=p^{sm}=p^{k'}$.
\subsection{Weakly self-dual quantum constacyclic BCH codes}\label{IVB}
A classical constacyclic BCH code $\mathcal{C}$ is specified by a zero set of the generator polynomial as we described in Subsection \ref{BCH}. In the sequel, we focus on the definition and the computation of the parameters of quantum codes based on the examples in the previous section. The following result provides the relation between the zero sets of the code and its dual. 
 \begin{theorem}\label{th8}
     Let $Z_{\mathcal{C}}$ denotes the zero set of the constacyclic BCH code $\mathcal{C}$ over $\mathbb{F}_q$ and the generator polynomial of the code is given by $g(x)=\prod_{s\in Z_{\mathcal{C}}}(x-\beta\xi^s)$. Then the generator polynomial of the dual code $\mathcal{C}^{\perp}$ is given by $h^{\dagger}(x)=\prod_{s\in Z_{\mathcal{C}^{\perp}}=\mathbb{Z}_n\setminus Z_{\mathcal{C}}}(x-\beta^{-1}\xi^{-s})$.
\end{theorem}
\begin{proposition}\label{prop1}
Let $\boldsymbol{v}\in \mathbb{F}_q^n$ having weight $d$ or less with its spectrum having at least $d$ consecutive zeros. Then $\boldsymbol{v} =0$. Furthermore, let $\mathbb{T}$ be the transformed vector set of $\mathcal{C}$. If the zero set $Z_{\mathcal{C}}$ contains $d -1$ consecutive numbers, then the minimum distance of $\mathcal{C}$ is $d$.
\end{proposition}
From Theorem \ref{th8}, we can say that the parity spectra of the constacyclic BCH code $\mathcal{C}$ and its dual $\mathcal{C}^{\perp}$ are the subsets of $Z_{\mathcal{C}}$ and $Z_{\mathcal{C}^{\perp}}$, respectively. 
In this subsection, we consider weakly self-dual quantum constacyclic BCH codes. First, we will derive a necessary and sufficient condition for the classical weakly self-dual constacyclic codes ($C\subseteq C^{\perp}$).
\begin{proposition}\label{pro1}
    Let $C$ be a constacyclic code of length $n$ over $\mathbb{F}_q$ which is generated by $g(x)=\sum_{i=0}^{n-k} g_ix^i$. Then $C$ is weakly self-dual, i.e., $C\subseteq C^{\perp}$  if and only if $x^n-\lambda$ divides $g(x)g^{\dagger}(x)$ where  $g^{\dagger}(x)=x^{n-k}g(1/x)$  denotes the reciprocal of $g(x)$ and $\lambda=\lambda^{-1}$.
\end{proposition}
\begin{proof}
   Let  $C$ be a weakly self-dual code, i.e., $C\subseteq C^{\perp}$ and $x^n-\lambda=h(x)g(x)$. Then $g(x)=h^{\dagger}(x)p(x)$ for some $p(x)\in \frac{\mathbb{F}_q[x]}{\langle x^n-\lambda\rangle}$. This gives 
   \begin{align*}
g(x)=&x^kh(1/x)p(x)=x^k\bigg(\frac{(1/x)^n-\lambda}{g(1/x)}\bigg)p(x)\\
       =&x^k\bigg(\frac{(1-\lambda x^n)}{x^ng(1/x)}\bigg)p(x)\\
       =&\bigg(-\lambda\frac{(-\lambda+x^n)}{x^{n-k}g(1/x)}\bigg)p(x)\\
       =&\frac{x^n-\lambda}{g^\dagger(x)}(-\lambda p(x))
   \end{align*}
 This implies $g(x)g^\dagger(x)=(x^n-\lambda)(-\lambda p(x))$. Therefore, $x^n-\lambda$ divides $g(x)g^{\dagger}(x)$.  
   Converse follows easily from the first part of the proof. 
\end{proof}
    Hence, any constacyclic code of length $n$ over $\mathbb{F}_q$ can be characterized by its generator polynomial which may be specified in terms of their zero spectrum over $\mathbb{F}_{q^m}$. Proposition \ref{pro1} can also be describe in terms of zero set of the code and its dual which will be useful for the quantum codes construction.
\begin{theorem}\label{th9}
    Let $\mathcal{C}$ be a constacyclic BCH code of length $n$ over $\mathbb{F}_q$. Let $Z_{\mathcal{C}}$ and $Z_{\mathcal{C}^{\perp}}$ are two zero sets of $\mathcal{C}$ and $\mathcal{C}^{\perp}$, respectively. Then $\mathcal{C}$ is weakly self-dual if and only if $Z_{\mathcal{C}^{\perp}}\subseteq Z_{\mathcal{C}}$.
\end{theorem}
\begin{proof}
   Let  $\mathcal{C}$ be weakly self-dual. Since $g(x)=\prod_{s\in Z_{\mathcal{C}}}(x-\beta\xi^s)$ and $h^{\dagger}(x)=\prod_{s\in Z_{\mathcal{C}^{\perp}}=\mathbb{Z}_n\setminus Z_{\mathcal{C}}}(x-\beta^{-1}\xi^{-s})$. Let $g(x)\in \mathcal{C}\subseteq \mathcal{C}^{\perp}=\langle h^\dagger(x)\rangle$. This implies $g(x)=h^\dagger(x)p(x)$ for some $p(x)\in \frac{\mathbb{F}_q[x]}{\langle x^n-\lambda\rangle}$. Hence, $h^\dagger(x)$ is a divisor of $g(x)$. Thus, $Z_{\mathcal{C}^{\perp}}\subseteq Z_{\mathcal{C}}$.

   Conversely, assume that $Z_{\mathcal{C}^{\perp}}\subseteq Z_{\mathcal{C}}$. This implies $g(x)=h^\dagger(x)p(x)$ for some $p(x)\in \frac{\mathbb{F}_q[x]}{\langle x^n-\lambda\rangle}$. Hence, $h^\dagger(x)$ is a divisor of $g(x)$. Therefore, from Proposition \ref{pro1}, we have $\mathcal{C}$ is weakly self-dual.
\end{proof}
Next, we recall the construction of QECCs from weakly self-dual codes in the time domain over $\mathbb{F}_{q}$. 
\begin{definition}\label{def3}
    Let $\mathcal{C}=[n,k,d_1]$ be a weakly self-dual ($\mathcal{C}\subseteq \mathcal{C}^{\perp}=[n,n-k,d_2]$) linear code over $\mathbb{F}_q$. Let $\{u_j: 0\leq j\leq q^{n-2k}\}$ be a set of coset representatives of $\mathcal{C}^\perp\slash \mathcal{C}$. Then $q^{n-2k}$ mutually orthogonal states
    $$\ket{\psi_j}=\frac{1}{\sqrt{|\mathcal{C}|}}\sum_{c\in \mathcal{C}}|c+u_j\rangle$$ span a quantum error-correcting code $\mathcal{C}$  with parameters $[[n, n -2k]]$. Based on classical decoding algorithms for the code $\mathcal{C}^{\perp}$, up to $(d_2 -1)/2$ errors can be corrected. Furthermore, the code can correct errors up to $(d'-1)/2$ where $d'=\mathrm{min}\{\text{wt}~c: c\in \mathcal{C}^\perp \setminus \mathcal{C}\}\geq d_2$.
\end{definition}
Next, we have the main result of this section.
\begin{theorem}\label{th10}
    Let $\mathcal{C}$ be an $[n,k,d_1]$ weakly self-dual constacyclic BCH code over $\mathbb{F}_q$, and the finite field Fourier transform of the vector $\boldsymbol{a}\in \mathcal{C}$ and $\boldsymbol{a}^\perp\in \mathcal{C}^\perp=[n,n-k,d_2]$ be the vector $\boldsymbol{A}$ and $\boldsymbol{A}^\perp$ in $\mathbb{F}_{q^m}$, such that $A_j^q=A_{qj\mod(n)}$ and $(A_j^q)^\perp=A_{qj\mod(n)}^\perp$ respectively. Then, there exists a quantum constacyclic BCH code in the spectral domain with at least, $d_1 -1$ consecutive zeros included in the code states, which can be described as 
    $\frac{1}{\sqrt{|\mathbb{T}|}}\sum_{j_1,\dots,j_n\in \mathbb{T}}|j_1j_2\cdots j_{k_1}\underbrace{0\cdots 0}_{d_1-1}j_{k_1+d}\cdots j_n+j_1'j_2'\cdots j_{k_2}'\underbrace{0\cdots 0}_{d_2-1}j_{k_2+d}'\cdots j_n'\rangle$, where  $j_1',\dots,j_n'\in \mathbb{T}^\perp\setminus\mathbb{T}$; $\mathbb{T}$ and $\mathbb{T}^\perp$ are the sets of Fourier transforms of $\mathcal{C}$ and $\mathcal{C}^\perp$, respectively. The code can correct errors up to $(d'-1)/2$ where $d'=\mathrm{min}\{\text{wt}~c: c\in \mathcal{C}^\perp \setminus \mathcal{C}\}\geq d_2$.
\end{theorem}
\begin{proof}
    From Definition \ref{def3}, any basis state of a quantum code can be  represented as 
    $$|\psi_i\rangle=\frac{1}{\sqrt{|\mathcal{C}|}}\sum_{c\in \mathcal{C}}|c+u_i\rangle.$$  From Proposition \ref{prop1} and Theorem \ref{th8}, FFFTs of $\mathcal{C}$  and $\mathcal{C}^\perp$ are two set of vectors over $\mathbb{F}_{q^m}$ with $d_1-1$ and $d_2-1$ consecutive spectral components in their parity spectrums, respectively. In the transform domain the states of the quantum code can be written as 
        $$|\Psi_j\rangle=\frac{1}{\sqrt{|\mathbb{T}|}}\sum_{C\in \mathbb{T}}|C+U_j\rangle,$$  where $U_j\in \mathbb{T}^\perp\setminus\mathbb{T}$. Now, computing the FFFT of the states  in this equation, we get our required result. From Theorem \ref{th9}, there are atleast $d_1-1$ and $d_2-1$ consecutive spectral components in their parity spectrums of $\mathcal{C}$  and $\mathcal{C}^\perp$ such that $Z_{\mathcal{C}^{\perp}}\subseteq Z_{\mathcal{C}}$. Thus, for the code $\mathcal{C}$ up to $(d_1-1)/2$ and 
    for the code $\mathcal{C}^\perp$ up to $(d_2-1)/2$ errors  can be corrected. Moreover,  since $d'=\text{min}\{\text{wgt}~c: c\in \mathcal{C}^\perp \setminus \mathcal{C}\}$, from Proposition \ref{prop1} and Definition \ref{def3}, there are at least $d'-1$ consecutive  spectrum components in the parity spectrum of $\mathcal{C}^\perp \setminus \mathcal{C}$. 
\end{proof}
\begin{example}
    In Example \ref{exam2}, let the generator polynomial of $\mathcal{C}$ in the time domain is $g(x)=(x+1)(x+w^2)(x+w^4)(x+w^6)(x+w^8)(x+w^{10})(x+w^{12})(x+w^{14})(x+w^{16})(x+w^{18})$. Then $\mathcal{C}$ with parameters $[13,3,11]$ has $10$ consecutive spectrum components in the parity spectrum. Also, the check polynomial of $\mathcal{C}$  or generator polynomial of $\mathcal{C}^\perp$ is given by $h^\dagger(x)=(x+w^{-20})(x+w^{-22})(x+w^{-24})=(x+w^6)(x+w^4)(x+w^2)$. Then $\mathcal{C}^\perp$ with parameters $[13,10,4]$ has $3$ consecutive spectrum components in the parity spectrum. It can be easily seen that $Z_{\mathcal{C}^{\perp}}\subseteq Z_{\mathcal{C}}$. Now, from Theorem \ref{th10}, we obtain the MDS (maximum-distance-separable) quantum constacyclic BCH code with parameters $[[13,7,4]]$.
\end{example}
\begin{example}
    Let $\mathbb{F}_{64}$ be a field of order $64$ where $\mathbb{F}_{64}=\mathbb{F}_2(w)$ and $w^6 + w^4 + w^3 + w + 1=0$. Let $n=7$ and $\xi=w^9$. In $\mathbb{F}_{64}[x]$, we have
    $x^{7}-1= (x+1)(x+w^{9})(x+w^{18})(x+w^{27})(x+w^{36})(x+w^{45})(x+w^{54})$. Let $g(x)=(x+1)(x+w^{9})(x+w^{18})(x+w^{27})$. Then $\mathcal{C}$ with parameters $[7,3,5]$ has $4$ consecutive spectrum components in the parity spectrum. Also, generator polynomial of $\mathcal{C}^\perp$ is given by $h^\dagger(x)=(x+w^{-36})(x+w^{-45})(x+w^{-54})=(x+w^{27})(x+w^{18})(x+w^9)$. Then $\mathcal{C}^\perp$ with parameters $[7,4,4]$ has $3$ consecutive spectrum components in the parity spectrum. It can be easily seen that $Z_{\mathcal{C}^{\perp}}\subseteq Z_{\mathcal{C}}$. Now, from Theorem \ref{th10}, we obtain the MDS quantum constacyclic BCH code with parameters $[[7,1,4]]$ which has the same length and dimension but better minimum distance than the QECC given by \cite{chowdhury2009quantum}.
\end{example}
\begin{remark}
   In Theorem \ref{th10}, for the construction of quantum constacyclic BCH code with $d'-1$ consecutive zeros, we have to choose spectral vectors carefully such that conjugate symmetry property holds. Also, to obtain the code in the time domain, we have to perform the IFFFT on the states of the code over $\mathbb{F}_{q^m}$ in the transform domain. Moreover, codes in the time domain and corresponding codes in the transform domain have the same error correcting capability. 
\end{remark}
\subsection{Dual-containing quantum constacyclic BCH codes}\label{sub43}
In this subsection, we consider dual-containing quantum constacyclic BCH codes. First, we will derive a necessary and sufficient condition for the classical dual-containing constacyclic codes ($C^\perp\subseteq C$).
\begin{proposition}\label{pro3}
    Let $C$ be a constacyclic code of length $n$ over $\mathbb{F}_q$ generated by $g(x)=\sum_{i=0}^{n-k} g_ix^i$. Then $C$ is a dual-containing  code, i.e., $C^\perp\subseteq C$  if and only if $g(x)g^{\dagger}(x)$  divides $x^n-\lambda$ where $g^{\dagger}(x)=x^{n-k}g(1/x)$ denotes the reciprocal of $g(x)$ and $\lambda=\lambda^{-1}$.
\end{proposition}
\begin{proof}
   Let  $C$ be dual-containing, i.e., $C^{\perp}\subseteq C$ and $x^n-\lambda=h(x)g(x)$. 
   The $\lambda^{-1}$ constacyclic code $C^\perp$ is generated by the polynomial $h^{\dagger}(x)=x^{k}h(1/x)$. Then $h^{\dagger}(x)=g(x)p(x)$ for some $p(x)\in \frac{\mathbb{F}_q[x]}{\langle x^n-\lambda\rangle}$. This gives $x^{k}h(1/x)=g(x)p(x)$.  This implies that $x^{k}\frac{(1-\lambda x^n)}{x^ng(1/x)}=g(x)p(x)$. This implies $x^n-\lambda=g(x)g^\dagger(x)(-\lambda p(x))$. Therefore, $g(x)g^{\dagger}(x)$ divides $x^n-\lambda$.  
   Converse follows easily from the first part of the proof. 
\end{proof}
    Hence, any constacyclic code of length $n$ over $\mathbb{F}_q$ can be characterized by its generator polynomial which may be specified in terms of their zero spectrum over $\mathbb{F}_{q^m}$. Proposition \ref{pro1} can also be described in terms of zero set of the code and its dual which will be useful for the dual-containing quantum code construction.
\begin{theorem}\label{th12}
    Let $\mathcal{C}$ be a constacyclic BCH code of length $n$ over $\mathbb{F}_q$. Let $Z_{\mathcal{C}}$ and $Z_{\mathcal{C}^{\perp}}$ are two zero sets of $\mathcal{C}$ and $\mathcal{C}^{\perp}$, respectively. Then $\mathcal{C}$ is dual-containing  if and only if $Z_{\mathcal{C}}\subseteq Z_{\mathcal{C}^\perp}$.
\end{theorem}
\begin{proof}
   Let  $\mathcal{C}$ be a dual-containing code. Since $g(x)=\prod_{s\in Z_{\mathcal{C}}}(x-\beta\xi^s)$ and $h^{\dagger}(x)=\prod_{s\in Z_{\mathcal{C}^{\perp}}=\mathbb{Z}_n\setminus Z_{\mathcal{C}}}(x-\beta^{-1}\xi^{-s})$. Let $h^{\dagger}(x)\in \mathcal{C}^\perp\subseteq \mathcal{C}=\langle g(x)\rangle$. This implies $h^\dagger(x)=g(x)p(x)$ for some $p(x)\in \frac{\mathbb{F}_q[x]}{\langle x^n-\lambda\rangle}$. Hence, $g(x)$ is a divisor of $h^\dagger(x)$. Thus, $Z_{\mathcal{C}}\subseteq Z_{\mathcal{C}^{\perp}}$.

   Conversely, assume that $Z_{\mathcal{C}}\subseteq Z_{\mathcal{C}^{\perp}}$. This implies $h^\dagger(x)=g(x)p(x)$ for some $p(x)\in \frac{\mathbb{F}_q[x]}{\langle x^n-\lambda\rangle}$. Hence, $g(x)$ is a divisor of $h^\dagger(x)$. Therefore, from Proposition \ref{pro3}, we have $\mathcal{C}$ is dual-containing, i.e., $\mathcal{C}^\perp\subseteq \mathcal{C}$.
\end{proof}
Next, we recall the construction of QECCs from dual-containing codes in the time domain over $\mathbb{F}_{q}$.
\begin{definition}\label{def5}
    Let $\mathcal{C}=[n,k,d_1]$ be a  dual-containing ($\mathcal{C}^\perp\subseteq \mathcal{C}=[n,n-k,d_2]$) linear code over $\mathbb{F}_q$. Let $\{v_j: 0\leq j\leq q^{2k-n}\}$ be a set of coset representatives of $\mathcal{C}\slash \mathcal{C}^\perp$. Then $q^{2k-n}$ mutually orthogonal states
    $$\psi_j=\frac{1}{\sqrt{|\mathcal{C}^\perp|}}\sum_{c\in \mathcal{C}^\perp}|c+v_j\rangle$$ span a quantum error-correcting code $\mathcal{C}$  with parameters $[[n, 2k-n]]$. Based on classical decoding algorithms for the code $\mathcal{C}$, up to $(d_1 -1)/2$ errors can be corrected. Furthermore, the code can correct errors up to $(d'-1)/2$ where $d'=\mathrm{min}\{\text{wt}~c: c\in \mathcal{C} \setminus \mathcal{C}^\perp\}\geq d_1$.
\end{definition}
Next, we have the main result of this subsection.
\begin{theorem}\label{th14}
    Let $\mathcal{C}$ be an $[n,k,d_1]$ dual-containing constacyclic BCH code over $\mathbb{F}_q$, and the finite field Fourier transform of the vector $\boldsymbol{a}\in \mathcal{C}$ and $\boldsymbol{a}^\perp\in \mathcal{C}^\perp=[n,n-k,d_2]$ be the vector $\boldsymbol{A}$ and $\boldsymbol{A}^\perp$ in $\mathbb{F}_{q^m}$, such that $A_j^q=A_{qj\mod(n)}$ and $(A_j^q)^\perp=A_{qj\mod(n)}^\perp$ respectively. Then, there exists a quantum constacyclic BCH code in the spectral domain with at least, $d_2 -1$ consecutive zeros included in the code states, which can be described as 
    $\frac{1}{\sqrt{|\mathbb{M}^\perp|}}\sum_{s_1,\dots,s_n\in \mathbb{M}^\perp}|s_1s_2\cdots s_{k_1}\underbrace{0\cdots 0}_{d_2-1}s_{s_1+d}\cdots s_n+s_1's_2'\cdots s_{k_2}'\underbrace{0\cdots 0}_{d_1-1}s_{k_2+d}'\cdots s_n'\rangle$, where  $s_1',\dots,s_n'\in \mathbb{M}\setminus\mathbb{M}^\perp$; $\mathbb{M}$ and $\mathbb{M}^\perp$ are the sets of Fourier transforms of $\mathcal{C}$ and $\mathcal{C}^\perp$, respectively. The code can correct errors up to $(d'-1)/2$ where $d'=\mathrm{min}\{\text{wt}~c: c\in \mathcal{C} \setminus \mathcal{C}^\perp\}\geq d_1$.
\end{theorem}
\begin{proof}
    From Definition \ref{def5}, any basis state of a quantum code can be  represented as 
    $$|\psi_i\rangle=\frac{1}{\sqrt{|\mathcal{C}^\perp|}}\sum_{c\in \mathcal{C}^\perp}|c+v_i\rangle.$$  From Proposition \ref{prop1} and Theorem \ref{th8}, FFFTs of $\mathcal{C}$  and $\mathcal{C}^\perp$ are two sets of vectors over $\mathbb{F}_{q^m}$ with $d_1-1$ and $d_2-1$ consecutive spectral components in their parity spectrums, respectively. In the transform domain, the states of the quantum code can be written as 
        $$|\Psi_j\rangle=\frac{1}{\sqrt{|\mathbb{M}^\perp|}}\sum_{C\in \mathbb{M}^\perp}|C+V_j\rangle,$$  where $V_j\in \mathbb{M}\setminus\mathbb{M}^\perp$. Now, computing the FFFT of the states  in this equation, we get our required result. From Theorem \ref{th12}, there are at least $d_1-1$ and $d_2-1$ consecutive spectral components in their parity spectra of $\mathcal{C}$  and $\mathcal{C}^\perp$ such that $Z_{\mathcal{C}}\subseteq Z_{\mathcal{C}^\perp}$. Thus, for the code $\mathcal{C}$, up to $(d_1-1)/2$ and 
    for the code $\mathcal{C}^\perp$ up to $(d_2-1)/2$ errors can be corrected. Moreover,  since $d'=\text{min}\{\text{wt}~c: c\in \mathcal{C} \setminus \mathcal{C}^\perp\}$, from Proposition \ref{prop1} and Definition \ref{def5}, there are at least $d'-1$ consecutive  spectral components in the parity spectra of $\mathcal{C} \setminus \mathcal{C}^\perp$. 
\end{proof}
\begin{example}
    In Example \ref{exam2}, let the generator polynomial of $\mathcal{C}$ in the time domain is $$g(x)=(x+w^{18})(x+w^{20})(x+w^{22})(x+w^{24}).$$ Then $\mathcal{C}$ with parameters $[13,9,5]$ has $4$ consecutive spectral components in the parity spectrum. Also, the check polynomial of $\mathcal{C}$  or generator polynomial of $\mathcal{C}^\perp$ is given by 
    \begin{align*}
        h^\dagger(x)=&(x+1)(x+w^{-2})(x+w^{-4})(x+w^{-6})(x+w^{-8})(x+w^{-10})(x+w^{-12})(x+w^{-14})(x+w^{-16})\\=&(x+1)(x+w^{24})(x+w^{22})(x+w^{20})(x+w^{18})(x+w^{16})(x+w^{14})(x+w^{12})(x+w^{10}).
    \end{align*} Then $\mathcal{C}^\perp$ with parameters $[13,4,10]$ has $9$ consecutive spectrum components in the parity spectrum. It can be easily seen that $Z_{\mathcal{C}}\subseteq Z_{\mathcal{C}^{\perp}}$. Now, from Theorem \ref{th14}, we obtain the MDS (maximum-distance-separable) quantum constacyclic BCH code with parameters $[[13,5,5]]$.
\end{example}
\subsection{Construction of QECCs from repeated-root constacyclic codes}\label{repeatedroot}
In this subsection, we construct quantum codes in the spectral domain using the repeated-root constacyclic codes with block length $L=p^\eta n$ with $t\geq 0$ and $p\nmid n$. It is well-known that there is a relation between cyclotomic cosets and constacyclic codes \cite{chen2012constacyclic}. Here, we use the $q$-cyclotomic cosets modulo $n$ and generator polynomials to describe constacyclic codes satisfying  the CSS-based quantum construction.

For an integer $r$ with $0\leq r\leq n-1$, the 
$q$-cyclotomic coset of $r$ modulo $n$  is defined by

$$\mathbb{C}_r=\{r.q^j(\mathrm{mod}~n)|j=0,1,\dots\}.$$
A subset $\{r_1,r_2,\dots,r_\rho\}$ of $\{0,1,\dots,n -1\}$ is called a complete set of representatives of all $q$-cyclotomic
cosets modulo $n$ if $\mathbb{C}_{r_1},\mathbb{C}_{r_2},\dots,\mathbb{C}_{r_\rho}$ are distinct and 
$$\cup_{i=1}^{\rho}\mathbb{C}_{r_i}=\{0,1,\dots,n -1\}.$$ In particular, 
$$\mathbb{C}_r=\{r,qr,q^r,\dots,q^{e-1}s\},$$ where $e$ is the smallest positive integer such that $q^{e-1}s\cong s~(\mod n)$. The $q$-cyclotomic coset $\mathbb{C}_r$ is said to be symmetric if and only if $n-r\in \mathbb{C}_r$, and asymmetric if otherwise. Also, the asymmetric cosets always appear in pairs denoted by $\mathbb{C}_r$ and $\mathbb{C}_{-r}=\mathbb{C}_{n-r}$. Let the number of symmetric cosets be $\rho_1$ and asymmetric pairs be $\rho_2$. Since, $\xi$ is the primitive $n^{\mathrm{th}}$ root of unity in some extension field of $\mathbb{F}_q$. Let $\mathcal{M}_\xi(x)$ denotes the minimal polynomial of $\xi$ over $\mathbb{F}_q$. For cyclic codes, it is well-known that (refer \cite{huffman2010fundamentals}) 
$$x^n-1=\mathcal{M}_{\xi^{r_1}}(x)\mathcal{M}_{\xi^{r_2}}(x)\cdots\mathcal{M}_{\xi^{r_\rho}}(x)$$
with $\mathcal{M}_{\xi^{r_i}}=\prod_{j\in \mathbb{C}_r}(x-\xi^j),$ $i=1,2\dots,\rho$, all being irreducible over $\mathbb{F}_q[x]$. Hence, 
\begin{equation}\label{eq20}
    x^L-1=(x^n-1)^{p^\eta}=\mathcal{M}_{\xi^{r_1}}(x)^{p^\eta}\mathcal{M}_{\xi^{r_2}}(x)^{p^\eta}\cdots\mathcal{M}_{\xi^{r_\rho}}(x)^{p^\eta}
\end{equation}
is the irreducible decomposition of $x^n-1$ in $\mathbb{F}_q[x]$. 

Now, from Section \ref{sec2} and equation \eqref{eq20}, we can find the irreducible decomposition of constacyclic codes in terms of minimal polynomials. Here, we assume that $\beta^L=\lambda$
\begin{equation}\label{eq21}
    (\beta^{-1}x)^L-1=((\beta^{-1}x)^n-1)^{p^\eta}=\mathcal{M}_{\xi^{r_1}}(\beta^{-1}x)^{p^\eta}\mathcal{M}_{\xi^{r_2}}(\beta^{-1}x)^{p^\eta}\cdots\mathcal{M}_{\xi^{r_\rho}}(\beta^{-1}x)^{p^\eta}.
\end{equation}
This gives
\begin{equation}\label{eq22}
    x^L-\lambda=\lambda\mathcal{M}_{\xi^{r_1}}(\beta^{-1}x)^{p^\eta}\mathcal{M}_{\xi^{r_2}}(\beta^{-1}x)^{p^\eta}\cdots\mathcal{M}_{\xi^{r_\rho}}(\beta^{-1}x)^{p^\eta}.
\end{equation}
For constacyclic codes, we consider $\mathcal{M}_r=\prod_{i\in \mathbb{C}_r}(x-\beta\xi^i).$ From this we can write equation \eqref{eq22} as follows:
$$x^L-\lambda=\lambda\bigg(\prod_{l=1}^{\rho_1}(\mathcal{M}_{i_l}(\beta^{-1}x))^{p^\eta}\prod_{l=1}^{\rho_2}(\mathcal{M}_{j_l}(\beta^{-1}x)\mathcal{M}_{-j_l}(\beta^{-1}x))^{p^\eta}\bigg),$$
where $\mathbb{C}_{i_l}$ is symmetric ($1\leq l\leq \rho_1$) and $\mathbb{C}_{j_l}$, $\mathbb{C}_{-j_l}$ are asymmetric pairs ($1\leq l\leq \rho_2$). 

Let $C$ be a constacyclic code of length $L$ with generator polynomial $g(x)$ which is a divisor of $x^L-\lambda=h(x)g(x)$. Then the dual of $C$ has the generator polynomial $h^\dagger(x)$, as defind in subsection \ref{IVB}. Further, if $C_r$ is symmetric, then $\mathcal{M}_r=\mathcal{M}_r^\dagger$, however if $C_r$ and $C_{-r}$ are asymmetric pairs, then $\mathcal{M}_r=\mathcal{M}_{-r}^\dagger$. Now, we consider constacyclic BCH codes and assume that $\mathcal{M}_r$ has $d$ consecutive zeros, then $\mathcal{M}^b$ has $bd$ consecutive zeros. Next, we will find the necessary and sufficient condition for weakly self-dual constacyclic codes with length $L=p^\eta n$. 
\begin{theorem}\label{th15}
    Let $C$ be a constacyclic code of length $L$ with generator polynomial 
    $$g(x)=\bigg(\prod_{l=1}^{\rho_1}(\mathcal{M}_{i_l}(\beta^{-1}x))^{a^l}\prod_{l=1}^{\rho_2}(\mathcal{M}_{j_l}(\beta^{-1}x))^{b^l}(\mathcal{M}_{-j_l}(\beta^{-1}x)^{c^l})\bigg).$$ Then $C\subseteq C^\perp$ if and only if $2a_l\geq p^\eta$ and $b_l+c_l\geq p^\eta$. Suppose that the number of consecutive zeros of the polynomial set $\bigg\{(\mathcal{M}_{i_{l_1}}(\beta^{-1}x))^{a^l},\mathcal{M}_{j_{l_2}}^{b^l}(\beta^{-1}x),\mathcal{M}_{-j_{l_2}}^{c^l}(\beta^{-1}x)|1\leq l_1\leq \rho_1, 1\leq l_2\leq \rho_2 \bigg\}$ is $d-1$. Then there exists a $L$-qudit quantum error-correcting code with distance $d$.
\end{theorem}
\begin{proof}
    Let $C$ be a constacyclic code of length $L$ with generator polynomial 
    $$g(x)=\bigg(\prod_{l=1}^{\rho_1}(\mathcal{M}_{i_l}(\beta^{-1}x))^{a^l}\prod_{l=1}^{\rho_2}(\mathcal{M}_{j_l}(\beta^{-1}x))^{b^l}(\mathcal{M}_{-j_l}(\beta^{-1}x)^{c^l})\bigg).$$ Then the generator polynomial of its dual code $C^\perp$ is given by
    $$h^\dagger(x)=\bigg(\prod_{l=1}^{\rho_1}(\mathcal{M}_{i_l}(\beta^{-1}x))^{p^\eta-a^l}\prod_{l=1}^{\rho_2}(\mathcal{M}_{-j_l}(\beta^{-1}x))^{p^\eta-b^l}(\mathcal{M}_{j_l}(\beta^{-1}x))^{p^\eta-c^l})\bigg).$$ From Theorem \ref{th8}, $C$ is weakly self-dual code if and only if $Z_{C^\perp}\subseteq Z_C$. Therefore, $C\subseteq C^\perp$ if and only if $h^\dagger(x)$ divides $g(x)$. Hence, $p^\eta-a^l\leq a^l$,  $p^\eta-b^l\leq c^l$ and  $p^\eta-c^l\leq b^l$. Thus,  $C\subseteq C^\perp$ if and only if $2a_l\geq p^\eta$ and $b_l+c_l\geq p^\eta$. From Proposition \ref{pro3}, there are $d-1$ consecutive zeros in the parity spectrum of $C$. Now, From Theorems \ref{th10}, we get the required result.
\end{proof}
Next, we have a direct consequence of the Theorem \ref{th15}.
\begin{corollary}\label{Cor1}
    Let $C$ be a constacyclic code of length $L$ with generator polynomial 
    $$g(x)=\bigg(\prod_{l=1}^{\rho_1}(\mathcal{M}_{i_l}(\beta^{-1}x))^{a^l}\bigg),$$ i.e., all $q$-cyclotomic cosets are symmetric. Also, assume that $0 \leq 2a_l\geq p^\eta$. Then there exists a quantum error-correcting code of length $L$, dimension $k$ and minimum distance $d$, where $d-1$ is the number of consecutive zeros of the polynomial set $\bigg\{(\mathcal{M}_{i_{l}}(\beta^{-1}x))^{a^l}|1\leq l\leq \rho_1 \bigg\}$.
\end{corollary}

Now, we will find the necessary and sufficient condition for dual-containing constacyclic codes with length $L=p^\eta n$. 

\begin{theorem}\label{th16}
    Let $C$ be a constacyclic code of length $L$ with generator polynomial 
    $$g(x)=\bigg(\prod_{l=1}^{\rho_1}(\mathcal{M}_{i_l}(\beta^{-1}x))^{a^l}\prod_{l=1}^{\rho_2}(\mathcal{M}_{j_l}(\beta^{-1}x))^{b^l}(\mathcal{M}_{-j_l}(\beta^{-1}x)^{c^l})\bigg).$$ Then $C^\perp\subseteq C$ if and only if $2a_l\leq p^\eta$ and $b_l+c_l\leq p^\eta$. Suppose that the number of consecutive zeros of the polynomial set $\bigg\{(\mathcal{M}_{i_{l_1}}(\beta^{-1}x))^{a^l},\mathcal{M}_{j_{l_2}}^{b^l}(\beta^{-1}x),\mathcal{M}_{-j_{l_2}}^{c^l}(\beta^{-1}x)|1\leq l_1\leq \rho_1, 1\leq l_2\leq \rho_2 \bigg\}$ is $d-1$. Then there exists a $L$-qudit quantum error-correcting code with distance $d$.
\end{theorem}
\begin{proof}
    Let $C$ be a constacyclic code of length $L$ with generator polynomial 
    $$g(x)=\bigg(\prod_{l=1}^{\rho_1}(\mathcal{M}_{i_l}(\beta^{-1}x))^{a^l}\prod_{l=1}^{\rho_2}(\mathcal{M}_{j_l}(\beta^{-1}x))^{b^l}(\mathcal{M}_{-j_l}(\beta^{-1}x)^{c^l})\bigg).$$ Then the generator polynomial of its dual code $C^\perp$ is given by
    $$h^\dagger(x)=\bigg(\prod_{l=1}^{\rho_1}(\mathcal{M}_{i_l}(\beta^{-1}x))^{p^\eta-a^l}\prod_{l=1}^{\rho_2}(\mathcal{M}_{-j_l}(\beta^{-1}x))^{p^\eta-b^l}(\mathcal{M}_{j_l}(\beta^{-1}x))^{p^\eta-c^l})\bigg).$$ From Theorem \ref{th12}, $C$ is dual-containing code if and only if $Z_{C}\subseteq Z_{C^\perp}$. Therefore, $C^\perp\subseteq C$ if and only if $g(x)$ divides $h^\dagger(x)$. Hence, $p^\eta-a^l\geq a^l$,  $p^\eta-b^l\geq c^l$ and  $p^\eta-c^l\geq b^l$. Thus,  $C^\perp\subseteq C$ if and only if $2a_l\leq p^\eta$ and $b_l+c_l\leq p^\eta$. From Proposition \ref{pro3}, there are $d-1$ consecutive zeros in the parity spectrum of $C$. Now, from Theorem \ref{th14}, we get the required result.
\end{proof}
Next, we have the following result as a direct consequence of Theorem \ref{th16}.
\begin{corollary}
    Let $C$ be a constacyclic code of length $L$ with generator polynomial 
    $$g(x)=\bigg(\prod_{l=1}^{\rho_1}(\mathcal{M}_{i_l}(\beta^{-1}x))^{a^l}\bigg),$$ i.e., all $q$-cyclotomic cosets are symmetric. Also, assume that $0 \leq 2a_l\leq p^\eta$. Then there exists a quantum error-correcting code of length $N$, dimension $k$ and minimum distance $d$, where $d-1$ is the number of consecutive zeros of the polynomial set $\bigg\{(\mathcal{M}_{i_{l}}(\beta^{-1}x))^{a^l}|1\leq l\leq \rho_1 \bigg\}$.
\end{corollary}

\begin{example}
Let $L=3^2.13=117$, $p^\eta=9$ and $q=3^3$. Then in $\mathbb{F}_{27}[x]$, we have
\begin{align*}
    x^{117}-w^{13}=(x^{13}-w^{13})^{9}=& (x+1)^9(x+w^2)^9(x+w^4)^9(x+w^6)^9(x+w^8)^9\\&(x+w^{10})^9(x+w^{12})^9(x+w^{14})^9(x+w^{16})^9(x+w^{18})^9\\&(x+w^{20})^9(x+w^{22})^9(x+w^{24})^9.
\end{align*} 
Here, one can easily check that only $C_0=\{0\}$ is a symmetric $q$-cyclotomic coset, as $n-r\in C_r$. Also, $C_2=\{2\}$, $C_4=\{4\}$, $C_6=\{6\}$, $C_8=\{8\}$, $C_{10}=\{10\}$, $C_{12}=\{12\}$, $C_{14}=\{1\}=C_1$, $C_{16}=\{3\}=C_3$, $C_{18}=\{5\}=C_5$, $C_{20}=\{7\}=C_7$, $C_{22}=\{9\}=C_9$, $C_{24}=\{11\}=C_{11}$. Here, $\{C_{2},C_{11}\}$, $\{C_{4},C_{9}\}$, $\{C_{6},C_{7}\}$, $\{C_{8},C_{5}\}$, $\{C_{10},C_{3}\}$ and $\{C_{12},C_{1}\}$ are asymmetric $q$-cyclotomic 
 cosets that appeared in pairs, as $n-r\notin C_r$. Let 
\begin{align*}
    g(x)=& (x+1)^5(x+w^2)^5(x+w^4)^5(x+w^6)^5(x+w^8)^5\\&(x+w^{10})^5(x+w^{12})^5(x+w^{14})^5(x+w^{16})^5(x+w^{18})^5\\&(x+w^{20})^5(x+w^{22})^5(x+w^{24})^5
\end{align*} 
and 
\begin{align*}
    h(x)=& (x+1)^4(x+w^2)^4(x+w^4)^4(x+w^6)^4(x+w^8)^4\\&(x+w^{10})^4(x+w^{12})^4(x+w^{14})^4(x+w^{16})^4(x+w^{18})^4\\&(x+w^{20})^4(x+w^{22})^4(x+w^{24})^4.
\end{align*}
Now,
\begin{align*}
    h^\dagger(x)=& (x+1)^4(x+w^{-2})^4(x+w^{-4})^4(x+w^{-6})^4(x+w^{-8})^4(x+w^{-10})^4(x+w^{-12})^4(x+w^{-14})^4\\&(x+w^{-16})^4(x+w^{-18})^4(x+w^{-20})^4(x+w^{-22})^4(x+w^{-24})^4\\
    =&(x+1)^4(x+w^{24})^4(x+w^{22})^4(x+w^{20})^4(x+w^{18})^4(x+w^{16})^4(x+w^{14})^4(x+w^{12})^4(x+w^{10})^4\\&(x+w^{8})^4(x+w^{6})^4(x+w^{4})^4(x+w^{2})^4
\end{align*} 
Since $p^\eta=9$, 
$2a_l\geq 9$, $b_l+c_l\geq 9$. Hence, from Theorem \ref{th15}, $C\subseteq C^\perp$. Then $\mathcal{C}$ with parameters $[117,52,3]$ has $2$ consecutive spectral components in the parity spectrum. Also, $\mathcal{C}^\perp$ with parameters $[117,65,3]$ has $2$ consecutive spectral components in the parity spectrum. Now, from Theorem \ref{th15}, we obtain the  quantum constacyclic code with parameters $[[117,13,3]]$, which is far away from MDS ( $19< 119, ~~2d+k\leq n+2)$. It can be inferred that quantum constacyclic BCH codes are more efficient than repeated root constacyclic codes.
\end{example}
\begin{example}
  Let $L=3.8=24$, $p^\eta=3$ and $q=3^4$. Then in $\mathbb{F}_{81}[x]$, we have
\begin{align*}
    x^{24}-2=(x^{8}-2)^{3}=& (x+w^5)^3(x+w^{15})^3(x+w^{25})^3(x+w^{35})^3(x+w^{45})^3(x+w^{55})^3(x+w^{65})^3\\&(x+w^{75})^3.
\end{align*} 
We compute $C_5=\{5\}=C_{45}$, $C_{15}=\{7\}=C_{55}=C_7$, $C_{25}=\{1\}=C_{65}=C_1$ and  $C_{35}=\{3\}=C_{75}=C_3$. Here, $\{C_{5},C_{45}\}$, $\{C_{15},C_{55}\}$, $\{C_{25},C_{65}\}$ and $\{C_{35},C_{75}\}$ are asymmetric $q$-cyclotomic 
 cosets that appeared in pairs. Hence all $q$-cyclotomic cosets are asymmetric, as $n-r \notin C_r$. Let 
\begin{align*}
    g(x)=& (x+w^5)^3(x+w^{15})^3(x+w^{25})^3(x+w^{45})^9(x+w^{55})^3(x+w^{65})^3.
\end{align*}  
Then
\begin{align*}
    h(x)=& (x+w^{35})^3(x+w^{75})^3
\end{align*} 
and
\begin{align*}
    h(x)^\dagger=& (x+w^{-35})^3(x+w^{-75})^3\\
    =&(x+w^{55})^3(x+w^{5})^3.
\end{align*}  
Since $b_l+c_l\geq 3$ and $p^\eta=3$, from Theorem \ref{th15}, $C\subseteq C^\perp$. Then $\mathcal{C}$ with parameters $[24,6,6]$ has $5$ consecutive spectral components in the parity spectrum. Also, $\mathcal{C}^\perp$ with parameters $[24,18,3]$ has $2$ consecutive spectral components in the parity spectrum. Now, from Theorem \ref{th15}, we obtain the  quantum constacyclic code with parameters $[[24,12,3]]$, which is far away from MDS ( $18< 26, ~~2d+k\leq n+2)$. It can be inferred  that quantum constacyclic BCH codes are more efficient than repeated root constacyclic codes.  
\end{example}
    \begin{example}
  Let $L=2^2.5=20$. Then in $\mathbb{F}_{16}[x]$, we have
\begin{align*}
    x^{20}-1=(x^{5}-1)^{4}=& (x+1)^4(x+w^3)^4(x+w^{6})^4(x+w^{9})^4(x+w^{12})^4.
\end{align*} 
Here, $n=5$, $q^m=2^4=16$ and $n|q^m-1$.
Also,  $C_{0}=\{0\}$, $C_1=\{1,2,4,3\}=C_{2}=C_3=C_4$, all are symmetric $q$-cyclotomic cosets, as $n-r \in C_r$. Here $\xi=w^3$. Since 
\begin{align*}
    x^{L}-1=(x^{n}-1)^{p^\eta}=& 
    \mathcal{M}_{\xi^{r_1}}(x)^{p^\eta}\mathcal{M}_{\xi^{r_2}}(x)^{p^\eta}\cdots\mathcal{M}_{\xi^{r_\rho}}(x)^{p^\eta}
\end{align*} 
where $\mathcal{M}_{\xi^{r_i}}=\prod_{j\in \mathbb{C}_r}(x-\xi^j),$ $i=1,2\dots,\rho_1$. Then  in $\mathbb{F}_{2}[x]$, we have
\begin{align*}
    x^{20}-1=(x^{5}-1)^{4}=& \mathcal{M}_{\xi^{r_1}}(x)^{4}\mathcal{M}_{\xi^{r_2}}(x)^{4}\\
    =&\prod_{j\in \mathbb{C}_0}(x-\xi^j)\prod_{j\in \mathbb{C}_1}(x-\xi^j)\\
    =&(x-1)^4(x-\xi)^4(x-\xi^2)^4(x-\xi^3)^4(x-\xi^4)^4\\
    =&(x-1)^4(x-w^3)^4(x-w^6)^4(x-w^9)^4(x-w^4)^4\\
    =&(x+1)^4(x^4 + x^3 + x^2 + x + 1)^4.
\end{align*} 
Let
\begin{align*}
    g(x)=& (x+1)^3(x^4 + x^3 + x^2 + x + 1)^3.
\end{align*}  
Finally, we have
\begin{align*}
    h^\dagger(x)=h(x)=& (x+1)(x^4 + x^3 + x^2 + x + 1).
\end{align*} 
Since $a_l=3$, $2a_l\geq 4$ and $p^\eta=4$, from Theorem \ref{th15}, $C\subseteq C^\perp$. We have $\mathcal{C}$ with parameters $[20,5,4]$ having $3$ consecutive spectral components in the parity spectra. Also,  $\mathcal{C}^\perp$ with parameters $[20,15,2]$ has $1$ spectral component in the parity spectrum. Now, from Theorem \ref{th15}, we obtain the  quantum constacyclic code with parameters $[[20,5,2]]$, which is far away from MDS ( $9< 22, ~~2d+k\leq n+2)$. We infer that quantum cyclic BCH codes are more efficient than repeated root cyclic codes.  
\end{example}

\section{ Encoding of quantum constacyclic BCH code}\label{secA}
\subsection{Encoding architecture using Q(FFFT)}\label{sub51}
The encoding procedure of the constacyclic BCH 
code involves performing the inverse of the quantum version of the finite field Fourier transform ($\mathrm{FFFT}$) on the message state along with a few ancilla qudits.  Here, our framework is based on the idea used in \cite{nadkarni2021entanglement}. The procedure to construct the encoding operator comprises two steps: 
\begin{enumerate}
   \item To obtain the stabilizers of the initial
state.
\item To obtain the encoding operator as the Clifford
operator (Clifford operators, which convert one basis operator to another, are unitary operators that are necessary when working with basis operators \cite{farinholt2014ideal}) that transforms the code stabilizers to stabilizers of the initial state.
\end{enumerate}
From Lemma \ref{lem1}, the inverse finite field Fourier transform (IFFFT) on a vector $\boldsymbol{A}=(A_0,A_1,\dots,\\A_{n-1})$ is
\begin{align}
\mathrm{IFFFT}(\boldsymbol{A})=\!\! (a_0,a_1,\dots,a_{n-1}),\text{where }a_i\!=\!\!=\frac{1}{n\beta^i}\sum_{j=0}^{n-1}\xi^{-ij}A_j, \label{eqn:FFFT_inv}
\end{align}
where $\xi$ is the primitive element of $\mathbb{F}_{q^m}$.

Throughout, unless explicitly specified, $\mathrm{FFFT}$ and $\mathrm{IFFFT}$ are over $\mathbb{F}_{q^m}$ and operate on $n$ qudits. The operator $Q=\mathrm{Q}(\mathrm{FFFT})$ denotes the action of $\mathrm{FFFT}$ on the basis states of the qudits. For example, $\mathrm{Q}(\mathrm{FFFT})\ket{\boldsymbol{c}} = \ket{\mathrm{FFFT}(\boldsymbol{c})}$ for $\boldsymbol{c}\in\mathbb{F}_{q}^n$. Thus, we obtain the $\mathrm{Q}(\mathrm{FFFT})$ to be
\begin{align}
\mathrm{Q}(\mathrm{FFFT}) = \underset{\gamma\in\mathbb{F}_{q}^{n}}{\sum}\ket{\mathrm{FFFT}(\boldsymbol{\gamma})}\bra{\boldsymbol{\gamma}}. \label{eqn:QFFFT}
\end{align}

We note that the inverse of the operator $\mathrm{Q}(\mathrm{FFFT})$ is
\begin{align}
\mathrm{Q}(\mathrm{IFFFT}) =(\mathrm{Q}(\mathrm{FFFT}))^{-1} \!\!=\! (\mathrm{Q}(\mathrm{FFFT}))^{\dagger} \!= \!\!\underset{\gamma\in\mathbb{F}_{q}^{n}}{\sum}\!\!\!\ket{\boldsymbol{\gamma}}\!\bra{\mathrm{FFFT}(\boldsymbol{\gamma})}\!.\label{eqn:QFFFTinv}
\end{align}
 
We observe that the parity check matrix of the constacyclic BCH code in equation \eqref{eqn:H_b_RS} and the $\mathrm{FFFT}$ matrix are isomorphic. Thus, we can choose the encoding operator of the quantum constacyclic BCH code to be $\mathcal{E} = (\mathrm{Q}(\mathrm{FFFT}))^{-1}$ and obtain the initial state $\ket{\psi_0}$ as the state stabilized by the operators $(\mathrm{Q}(\mathrm{FFFT}))S_i((\mathrm{Q}(\mathrm{FFFT}))^{-1}$, where $S_i\in S$, the stabilizer group of the quantum constacyclic BCH code and $S_i$'s  are the minimal generators of $S$. The encoding procedure involves performing the encoding operation $(\mathrm{Q}(\mathrm{FFFT}))^{-1}$ on the transmitter-end qudits of $\ket{\psi_0}$. Further, $(\mathrm{Q}(\mathrm{FFFT}))S_i((\mathrm{Q}(\mathrm{FFFT}))^{-1}$ is obtained by conjugating the stabilizer generator $S_i$ with $(\mathrm{Q}(\mathrm{FFFT})^{-1}$. The encoding operator $(\mathrm{Q}(\mathrm{FFFT}))^{-1}$ operates only on the $n$ codeword qudits that are present at the transmitter end. The operators $(\mathrm{Q}(\mathrm{FFFT})S_i((\mathrm{Q}(\mathrm{FFFT}))^{-1}$, where $S_i\in S$, are called as the \emph{initial stabilizers} as they stabilize the initial state $\ket{\psi_0}$.

Moreover, to obtain the initial state $\ket{\psi_0}$ that is stabilized by the initial stabilizers, we perform the following two steps:
\begin{itemize}
\item First  we simplify the form of the initial stabilizers.
\item We obtain the state stabilized by the simplified form of the initial stabilizers. This state is the initial state $\ket{\psi_0}$. The initial state $\ket{\psi_0}$ of the codeword qudits obtained is a tensor product of an $(n-2k)$-qudit message state $\ket{\phi}$ with a few ancilla qudits. 
\end{itemize}
\subsection{Simplification of the initial stabilizers}\label{sub52}
Here, we obtain the simplified version of the initial stabilizers provided in Subsection \ref{sub51}. For our constructed quantum constacyclic BCH code, the stabilizer generator is $S_i$. Thus, the initial stabilizers are 
\begin{align}
(\mathrm{Q}(\mathrm{FFFT})S_i((\mathrm{Q}(\mathrm{FFFT}))^{-1}.\label{eqn:InitStab}
\end{align}
 
We recall that the matrix $\mathcal{H_{\mathrm{BCH}}}$ of the quantum constacyclic BCH code is of the form provided in equation \eqref{eqHBCH}. The operators corresponding to $\mathcal{H_{\mathrm{BCH}}}$ are
\begin{align}
S^{(\mathrm{X})}_{i,l} &\!= \mathrm{X}^{(p^{k'})}(\xi^l)\otimes \mathrm{X}^{(p^{k'})}(\xi^l(\beta\xi^{i}) \otimes \cdots \otimes \mathrm{X}^{(p^{k'})}(\xi^l(\beta\xi^i)^{(n-1)})= \underset{g=0}{\overset{n-1}{\otimes}}\mathrm{X}^{(p^{k'})}(\xi^l(\beta\xi^{i})^g),\label{eqn:X_Stab_RS}\\
S^{(\mathrm{Z})}_{j,m} &\!=\mathrm{Z}^{(p^{k'})}(\xi^m)\!\otimes\! \mathrm{Z}^{(p^{k'})}(\xi^m(\beta\xi^{j}))\! \otimes \!\cdots \!\otimes \!\mathrm{Z}^{(p^{k'})}(\xi^m(\beta\xi^{j})^{(n-1)})= \underset{g=0}{\overset{n-1}{\otimes}}\mathrm{Z}^{(p^{k'})}(\xi^m(\beta\xi^{j})^{g}),\label{eqn:Z_Stab_RS}
\end{align}
where $i\in\{b_1,\dots,b_1+\delta-2\}$, $j\in\{b_2,\dots,b_2+\delta-2\}$, and $l,m\in\{0,\dots,k'-1\}$. We note that the subscripts $l$ and $i$ in $S^{(\mathrm{X})}_{i,l}$, and the subscripts $m$ and $j$ in $S^{(\mathrm{Z})}_{j,m}$ correspond to the power of $\xi$ in $\xi^l(\beta\xi^{i})^{g}$ and $\xi^m(\beta\xi^{j})^{g}$, respectively.

For $\gamma\in\mathbb{F}_{p^{k'}}$, let $\mathrm{X}_{f}^{(p^{k'})}(\gamma)$ and $\mathrm{Z}_{f}^{(p^{k'})}(\gamma)$ denote the operators operating $\mathrm{X}^{(p^{k'})}(\gamma)$ and $\mathrm{Z}^{(p^{k'})}(\gamma)$ on the ${f}^{\mathrm{th}}$ qudit and the identity operator on the rest of the qudits, respectively, where $f\in\{1,\dots,n\}$. For $\gamma\in \mathbb{F}_{p^{k'}}$, we have the following relations whose proofs are provided in Appendix \ref{app:RelationProof}:
\begin{itemize}
\item[(N1)]$\mathrm{Q}(\mathrm{FFFT})\mathrm{X}_{f}^{(p^{k'})}(\gamma)(\mathrm{Q}(\mathrm{FFFT}))^{-1}\newline =\underset{j=0}{\overset{n-1}{\otimes}}\mathrm{X}^{(p^{k'})}(\gamma(\beta\xi^{j})^{(f-1)}),~\forall~{f}\in\{1,\dots,n\}$.
\item[(N2)] $\mathrm{Q}(\mathrm{FFFT})\mathrm{Z}_{f}^{(p^{k'})}(\gamma)(\mathrm{Q}(\mathrm{FFFT}))^{-1}\newline =\underset{j=0}{\overset{n-1}{\otimes}}\mathrm{Z}^{(p^{k'})}\left(\frac{1}{n\lambda}\gamma(\beta\xi^{j})^{(n-{f}+1)}\right),~\forall~{f}\in\{1,\dots,n\}$.
\item[(N3)] $\mathrm{Q(FFFT)}\left(\underset{f=0}{\overset{n-1}{\otimes}}\mathrm{X}^{(p^{k'})}(\gamma(\beta\xi^{j})^{{f}})\right)(\mathrm{Q(FFFT))}^{-1}\newline =\mathrm{X}_{(n-{j}+1)}^{(p^{k'})}(n\gamma),~\forall~{j}\in\{1,\dots,n\}$.
\item[(N4)] $\mathrm{Q(FFFT)}\left(\underset{f=0}{\overset{n-1}{\otimes}}\mathrm{Z}^{(p^{k'})}(\gamma(\beta\xi^{j})^{{f}})\right)(\mathrm{Q(FFFT))}^{-1}\newline =\mathrm{Z}_{({j}+1)}^{(p^{k'})}(\gamma),~\forall~{j}\in\{0,\dots,n-1\}$.
\end{itemize}

We consider $S^{(\mathrm{X})}_{i,l} = \underset{g=0}{\overset{n-1}{\otimes}}\mathrm{X}^{(p^{k'})}(\xi^l(\beta\xi^{i})^g)$ from equation \eqref{eqn:X_Stab_RS}, where $i\in\{b_1,\dots,b_1+\delta-2\}$ and $l\in\{0,\dots,k'-1\}$. We substitute $S_i$ with $S^{(\mathrm{X})}_{i,l}$ in equation \eqref{eqn:InitStab}, to obtain the initial stabilizer based on $S^{(\mathrm{X})}_{i,l}$ to be
\begin{align}
&(\mathrm{Q}(\mathrm{FFFT})S_i(\mathrm{Q}(\mathrm{FFFT}))^{-1})=\!\left(\!\mathrm{Q}(\mathrm{FFFT})\!\left(\underset{g=0}{\overset{n-1}{\otimes}}\mathrm{X}^{(p^{k'})}(\xi^l(\beta\xi^i)^{g})\!\right)\!\!(\mathrm{Q}(\mathrm{FFFT})^{-1})\!\right). \label{eqn:InitStaD_X_mid1}
\end{align}

Using the property (N3) in equation \eqref{eqn:InitStaD_X_mid1} with $\gamma=\xi^l$ and $f=g$, we obtain the initial stabilizer to be
\begin{align}
&\left(\!\mathrm{Q}(\mathrm{FFFT})\!\left(\underset{g=0}{\overset{n-1}{\otimes}}\mathrm{X}^{(p^{k'})}(\xi^l(\beta\xi^i)^{g})\!\right)\!\!(\mathrm{Q}(\mathrm{FFFT})^{-1})\!\right)= \mathrm{X}_{(n-i+1)}^{(p^{k'})}(n\xi^l). \label{eqn:InitStabMid_X}
\end{align}

Similarly, we consider $S^{(\mathrm{Z})}_{j,m} = \underset{g=0}{\overset{n-1}{\otimes}}\mathrm{Z}^{(p^{k'})}(\xi^m(\beta\alpha^{j})^{g})$ from equation \eqref{eqn:Z_Stab_RS}, where $j\in\{b_2,\dots,$ $b_2+\delta-2\}$ and $m\in\{0,\dots,k'-1\}$. We substitute $S_i$ with $S^{(\mathrm{Z})}_{j,m}$ in equation \eqref{eqn:InitStab}, to obtain the initial stabilizer based on $S^{(\mathrm{Z})}_{j,m}$ to be
\begin{align}
&(\mathrm{Q}(\mathrm{FFFT})S_i(\mathrm{Q}(\mathrm{FFFT}))^{-1})=\!\!\left(\!\!\mathrm{Q}(\mathrm{FFFT})\!\!\left(\underset{g=0}{\overset{n-1}{\otimes}}\mathrm{Z}^{(p^{k'})}(\xi^m(\beta\alpha^{j})^{g})\!\right)\!\!(\mathrm{Q}(\mathrm{FFFT}))^{-1}\!\!\right)\!. \label{eqn:InitStab_Z_mid1}
\end{align}

Using the property (N4) in equation \eqref{eqn:InitStab_Z_mid1} with $\gamma = \xi^m$ and $f=g$, we obtain the initial stabilizer to be
\begin{align}
&\!\!\left(\!\!\mathrm{Q}(\mathrm{FFFT})\!\!\left(\underset{g=0}{\overset{n-1}{\otimes}}\mathrm{Z}^{(p^{k'})}(\xi^m(\beta\alpha^{j})^{g})\!\right)\!\!(\mathrm{Q}(\mathrm{FFFT}))^{-1}\!\!\right)\!=\mathrm{Z}_{(j+1)}^{(p^{k'})}(\xi^m). \label{eqn:InitStabMid_Z}
\end{align}

As $i\in\{b_1,\dots,b_1+\delta-2\}$ and $j\in\{b_2,\dots,b_2+\delta-2\}$, we have
\begin{align*}
(n-i+1)&\in\{n-b_1-\delta+3,\dots,n-b_1+1\},\\
(j+1)&\in\{b_2+1,\dots,b_2+\delta-1\}.
\end{align*}
 Thus, considering $s=(n-i+1)$ and $t=(j+1)$ in equations \eqref{eqn:InitStabMid_X} and \eqref{eqn:InitStabMid_Z}, we obtain the initial stabilizers to be 
\begin{align}
 \underbrace{\mathrm{X}_{s}^{(p^{k'})}(n\xi^l)}_{\text{over }n\text{ qudits}}~~\text{and}~~\underbrace{\mathrm{Z}_{t}^{(p^{k'})}(\xi^m)}_{\text{over }n\text{ qudits}},\!\!\label{eqn:InitStabFinal}
 \end{align}
 where $s\in\{n-b_1-\delta+3,\dots,n-b_1+1\}$, $t\in\{b_2+1,\dots,b_2+\delta-1\}$, and $l,m\in\{0,\dots, k'-1\}$.
 
Based on the ranges of $s$ and $t$, we define the position index sets as follows:
\begin{align}
T&=\{1,\dots,n\},\label{eqn:P}\\
T_{\mathrm{X}}&=\{n-b_1-\delta+3,\dots,n-b_1+1\}, \label{eqn:P_X}\\
T_{\mathrm{Z}} &= \{b_2+1,\dots,b_2+\delta-1\}.\label{eqn:P_Z}
\end{align}

In equation \eqref{eqn:InitStabFinal}, $s$ and $t$ are the indices of the qudit on which the operators $\mathrm{X}_s^{(p^{k'})}(n\xi^l)$ and $\mathrm{Z}_t^{(p^{k'})}(\xi^m)$ operate $\mathrm{X}^{(p^{k'})}(n\xi^l)$ and $\mathrm{Z}^{(p^{k'})}(\xi^m)$, respectively. Thus, $s,t\in\{1,\dots,n\}$ and
\begin{align}
s \in T_{\mathrm{X}}\text{ and }t\in T_{\mathrm{Z}}. \label{eqn:st_PX_PZ}
\end{align}

 Also, the set $T$ contains the indices of all the $n$ qudits, ranging from $1$ to $n$ and all the elements in the sets in the equations \eqref{eqn:P}-\eqref{eqn:P_Z} are modulo $n$ and range from $1$ to $n$.  From equation \eqref{eqn:InitStabFinal}, we observe that $T_{\mathrm{X}}$ and $T_{\mathrm{Z}}$ represent the qudit indices where the initial stabilizers operate $\mathrm{X}^{(p^{k'})}(n\xi^l)$ and $\mathrm{Z}^{(p^{k'})}(\xi^m)$, respectively. 
\subsection{ Initial state $\ket{\Psi_o}$ of the codeword qudits}\label{sub53}
In this subsection, we find  the initial state $\ket{\psi_0}$ based on the obtained initial stabilizers. We first define the following set of operators and states that will be used to find the initial state:
\begin{itemize}
\item[1)] The discrete Fourier transform operator $\mathrm{DFT}_{p^{k'}}$  (\cite{grassl2003efficient,farinholt2014ideal}) is given by
\begin{align}
\mathrm{DFT}_{p^{k'}}:= \frac{1}{\sqrt{p^{k'}}}\underset{\theta,\nu\in\mathbb{F}_{p^{k'}}}{\sum}\omega^{\mathrm{Tr}_{p^{k'}/p}(\theta\nu)}\ket{\theta}\bra{\nu}, \label{eqn:DFT_EARS}
\end{align} 
where $\omega = \mathrm{e}^{\mathrm{i}\frac{2\pi}{p}}$.
\item[2)]The addition operator $\mathrm{ADD}_{p^{k'}}(1,2)$ (\cite{grassl2003efficient,farinholt2014ideal}) is defined as follows:
\begin{align}
\mathrm{ADD}_{p^{k'}}(1,2) :=\!\! \underset{\theta_1,\theta_2 \in \mathbb{F}_{p^{k'}}}{\sum}\!\!\!\ket{\theta_1}\bra{\theta_1} \otimes \ket{(\theta_1+\theta_2)}\bra{\theta_2}. \label{eqn:ADD_EARS}
\end{align}
\item[3)] A single qudit state $\ket{\epsilon}$:
\begin{align}
\ket{\epsilon} &:=\mathrm{DFT}_{p^{k'}}^{-1}\ket{0} = \mathrm{DFT}_{p^{k'}}^{\dagger}\ket{0}= \frac{1}{\sqrt{p^{k'}}}\underset{\theta,\nu\in\mathbb{F}_{p^{k'}}}{\sum}\omega^{-\mathrm{Tr}_{p^{k'}/p}(\theta\nu)}\ket{\theta}\bra{\nu}\ket{0},\nonumber\\
&= \frac{1}{\sqrt{p^{k'}}}\underset{\theta\in\mathbb{F}_{p^{k'}}}{\sum}\omega^{-\mathrm{Tr}_{p^{k'}/p}(\theta.0)}\ket{\theta}= \frac{1}{\sqrt{p^{k'}}}\underset{\theta\in\mathbb{F}_{p^{k'}}}{\sum}\ket{\theta}.\label{eqn:ket_epsilon}
\end{align}
\end{itemize}
 
We have defined the state $\ket{\epsilon}$ because the initial state $\ket{\psi_0}$ will have a few codeword qudits containing these states as they are stabilized by the initial stabilizers.

For $\xi\in\mathbb{F}_{p^{k'}}$, we have the following relations that will be used further to get the initial state:
\begin{itemize}
\item[(N5)]$\mathrm{X}^{(p^{k'})}(n\xi^l)\ket{\epsilon} = \ket{\epsilon},\text{ where }\ket{\epsilon}=\mathrm{DFT}_{p^{k'}}^{-1}\ket{0}$
.\item[(N6)]$\mathrm{Z}^{(p^{k'})}(\xi^m)\ket{0} = \ket{0}.$
\end{itemize}
The proofs of the relations (N5)-(N6) are provided in Appendix \ref{app:RelationProof}.

  Let $\ket{\phi}$ be the $(n-2k)$-qudit message state to be encoded. The initial stabilizers in \eqref{eqn:InitStabFinal} act as non-identity operators only on the qudits in $T_{\mathrm{X}} \cup T_{\mathrm{Z}}$ as from equation \eqref{eqn:st_PX_PZ}, $s \in T_\mathrm{X}$ and $t \in T_\mathrm{Z}$. Thus, all the initial stabilizers operate the identity operator on the qudits in $T\setminus (T_{\mathrm{X}}\cup T_{\mathrm{Z}})$. Since the size of $T\setminus (T_{\mathrm{X}}\cup T_{\mathrm{Z}})$ is $n-2k=(-n+2(\delta-1))$ as the size of $T$, $T_{\mathrm{X}}$ and $T_{\mathrm{Z}}$,  is $n$, $(\delta-1)$, 
 and $(\delta-1)$, respectively, and  the identity operator stabilizes all the quantum states. Thus, we store the message state $\ket{\phi}$ in the $(n-2k)$ qudits in $T\setminus (T_{\mathrm{X}}\cup T_{\mathrm{Z}})$ where $n-k=\delta-1$.
   
    Next, we find the initial state of the rest of the qudits based on the following two cases: \newline
 \textbf{\underline{Case I}: Index ${s \in T_{\mathrm{X}}}$}

 The initial stabilizers that operate on the $s^{\mathrm{th}}$ qudit in a non-identity fashion are $\{\mathrm{X}_{s}^{(p^{k'})}(n\xi^l)\}_{l=0}^{k'-1} $. From (N5), the initial state for each qudit $s$ is $\ket{\epsilon}$.\newline
 \textbf{\underline{Case 2}: Index ${s \in T_{\mathrm{Z}}}$}

 The initial stabilizers that operate on the $s^{\mathrm{th}}$ qudit in a non-identity fashion are $\{\mathrm{Z}_{s}^{(p^{k'})}(\xi^m)\}_{m=0}^{k'-1}$. From (N6), the initial state for each qudit $s$ is $\ket{0}$. 
 
 From equation \eqref{eqn:InitStabFinal}, the initial stabilizers that operate on the $s^{\mathrm{th}}$ qudit in a non-identity fashion are $\{\mathrm{X}_{s}^{(p^{k'})}(n\xi^l)\}_{l=0}^{k'-1}$ and $\{\mathrm{Z}_{s}^{(p^{k'})}(\xi^m)\}_{m=0}^{k'-1}$.

\subsection{Encoding algorithm of the quantum constacyclic BCH  code}
Here, we provide the encoding algorithm for the quantum constacyclic BCH code. Since quantum constacyclic BCH code is a CSS code, and hence
the stabilizer generators can be decomposed into two sets, one containing the shift operators $\mathrm{X}^{(p^{k'})}(\cdot)$, corresponding to $H_{b1},\dots,\xi^{k'-1}H_{b1}$ in equation \eqref{eqHBCH}, and the other containing the clock operators $\mathrm{Z}^{(p^{k'})}(\cdot)$, corresponding to $H_{b2},\dots,\xi^{k'-1}H_{b2}$ in equation \eqref{eqHBCH}. 

Similar to the classical case, here we consider $\mathrm{Q}(\mathrm{FFFT})\ket{{\boldsymbol{\gamma}}}$ as the spectrum of $\ket{{\boldsymbol{\gamma}}}$, where ${\boldsymbol{\gamma}} \in \mathbb{F}_{p^{k'}}^{n}$. Analogous to the null spectra in the classical case, we consider the $\mathrm{X}$-null spectrum and the $\mathrm{Z}$-null spectrum corresponding to the positions in $T_{\mathrm{X}}$ and $T_{\mathrm{Z}}$, respectively. These positions correspond to the ancilla qudits and hence can be viewed similar to the nulls in the classical constacyclic BCH code that correspond to the parity bits in the spectral domain. 

Based on $T_{\mathrm{X}}$ and $T_{\mathrm{Z}}$,  using (N5), (N6), and equation \eqref{eqn:InitStabFinal}, we split the $n$-qudit codeword block into $3$ subblocks as follows:
\begin{itemize}
\item Subblock $D_X$ corresponds to the qudits that correspond to the $\mathrm{X}$-null spectrum. As  $T_{\mathrm{X}}$ has $\delta-1$ elements. From (N5), this block contains the qudits to be in state $\ket{\epsilon}$ before encoding. We call the qudits in this block as the $\mathrm{X}$-ancilla qudits. (Case 1)
\item Subblock $D_Z$ corresponds to the qudits that correspond to the $\mathrm{Z}$-null spectrum. As $T_{\mathrm{Z}}$ has $\delta-1$ elements. From (N6), this block contains the qudits to be in state $\ket{0}$ before encoding. We call the qudits in this block as the $\mathrm{Z}$-ancilla qudits. (Case 2) 
\item Subblock $D_M$ corresponds to the message subblock comprising of the remaining $n_m = n-2k=2(\delta-1)-n$ qudits. This block contains the message $\ket{\phi}$ before encoding. We call the qudit in this block as the message qudits.
\end{itemize}
\begin{figure}
    \centering
    \includegraphics[scale=.5]{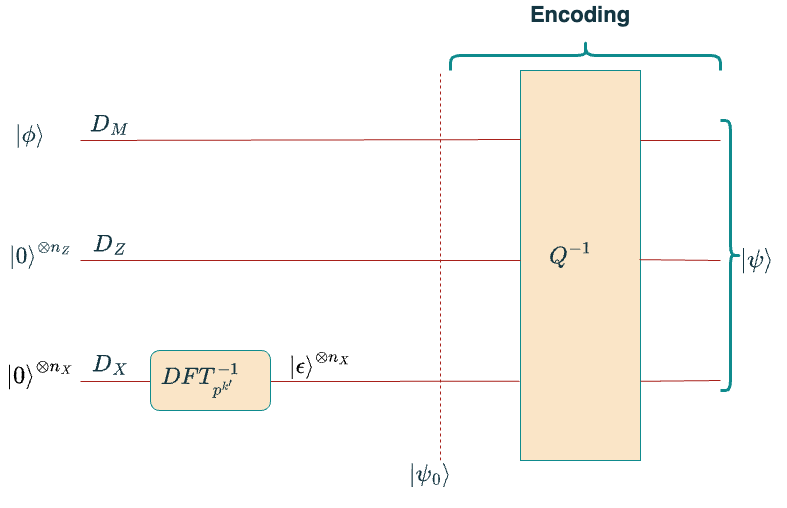}
    \caption{Encoding circuit for quantum constacyclic BCH codes. The encoding operator $\mathcal{E} = Q^{-1}$ is performed on the first $n$ qudits of the initial state $\ket{\psi_0}$ to obtain the codeword $\ket{\psi}$. $\ket{\epsilon}$ is obtained from $\ket{0}$ by performing $\mathrm{DFT}_{p^{k'}}^{-1}$. 
    Then the initial state $\ket{\psi_0}$ is obtained from message state $\ket{\phi}$, $\ket{0}$s, and $\ket{\epsilon}$s. }
    \label{EncodingBCH}
\end{figure}

 Figure \ref{EncodingBCH} provides the encoding circuit for the quantum constacyclic BCH  code. We note that $\ket{\epsilon} = \mathrm{DFT}_{p^{k'}}^{-1}\ket{0} = \mathrm{DFT}_{p^{k'}}\ket{0}$.

The initial state of the codeword is 
\begin{align}
\ket{\psi_0} = \ket{\phi}_{D_M}(\ket{0}^{\otimes n_{\mathrm{Z}}})_{D_Z}(\ket{\epsilon}^{\otimes n_{\mathrm{X}}})_{D_X},\label{eqn:InitRSCodeword}
\end{align}
where the subscripts $D_M$, $D_X$, and $D_Z$ denote the subblocks of the message qudits, $\mathrm{X}$-ancilla qudits, and $\mathrm{Z}$-ancilla qudits, respectively. The message qudits are initially in the state $\ket{\phi}$. The $\mathrm{Z}$-ancilla qudits are initially in the state $\ket{0}$. The $\mathrm{X}$-ancilla qudits are initially in the state $\ket{\epsilon}$. 

The encoding operation $\mathcal{E}$ is given as
\begin{align}
\mathcal{E} = (\mathrm{Q}(\mathrm{FFFT}))^{-1}, \label{eqn:EncodingOperator}
\end{align}
and the codeword is
\begin{align}
\ket{\psi} = \mathcal{E}\ket{\psi_0}. \label{eqn:Codeword}
\end{align}

\section{Decoding and error correction procedure of quantum constacyclic BCH code}\label{decoding}

The error correction procedure involves the following steps: 
\begin{itemize}
    \item[(a)] syndrome computation,
    \item[(b)] error deduction and recovery.
\end{itemize}

 The syndrome computation procedure involves obtaining the syndrome based on the error using the $\mathrm{Q}(\mathrm{FFFT})$, $\mathrm{DFT}_{p^{k'}}$ and $\mathrm{ADD}_{p^{k'}}$  operations defined in equations \eqref{eqn:QFFFT}, \eqref{eqn:DFT_EARS}, and \eqref{eqn:ADD_EARS}.
 
  The ancillary qudits used for syndrome computation called the \emph{syndrome qudits} store the syndrome after the syndrome computation procedure. These syndrome qudits are measured. 
Based on the measurement result, the syndrome is obtained and the error is deduced classically using the Berlekamp-Massey algorithm or the Euclidean algorithm or one can use our proposed algorithm in Section \ref{sec 3}. Based on the error deduced, the inverse error operation is performed on the received qudits to recover the codeword.

 We first provide our syndrome computation procedure to obtain the syndrome based on the error. We discuss the syndrome computation for the basis errors.
 Let the error $E = \mathrm{X}^{(p^{k'})}(\alpha_1)\mathrm{Z}^{(p^{k'})}\\(\gamma_1)\otimes \mathrm{X}^{(p^{k'})}(\alpha_2)\mathrm{Z}^{(p^{k'})}(\gamma_2)\otimes \dots \otimes \mathrm{X}^{(p^{k'})}(\alpha_n)\mathrm{Z}^{(p^{k'})}(\gamma_n)$ operate on the codeword $\ket{\psi}$ to obtain $E\ket{\psi}$. Let ${\boldsymbol{\alpha}} = [\alpha_1~\alpha_2~\dots ~\alpha_n]$ and ${\boldsymbol{\gamma}} = [\gamma_1~\gamma_2~\dots ~\gamma_n]$. Then $E=\mathrm{X}^{(p^{k'})}({\boldsymbol{\boldsymbol{\alpha}}})\mathrm{Z}^{(p^{k'})}({\boldsymbol{\gamma}}) \equiv[{\boldsymbol{\alpha}}|{\boldsymbol{\gamma}}]_{p^{k'}}$.
 
 For the stabilizer codes over $\mathbb{F}_{p^{k'}}$, in general, the elements of the syndrome of $E\ket{\psi}$ are the symplectic products of the stabilizer generators $S_i$s with the error $E$. Each syndrome element belongs to $\mathbb{F}_{p}$, and is stored in a subqudit. However, for our quantum constacyclic BCH code, we obtain the syndrome elements that belong to $\mathbb{F}_{p^{k'}}$, and is stored in a qudit. 
 
 In the check matrix $\mathcal{H_{\mathrm{BCH}}}$ in equation \eqref{eqHBCH}, $H_{b1}$ and $H_{b2}$ are used for the correction of $\mathrm{Z}^{(p^{k'})}(\cdot)$ and $\mathrm{X}^{(p^{k'})}(\cdot)$ errors, respectively. For the basis error $E$, we obtain the syndrome $\ket{s}$ similar to the classical case as 
 \begin{align}
 \ket{s} = \ket{{s_{\mathrm{X}}}} \otimes \ket{{s_{\mathrm{Z}}}}, \label{eqn:Synd_EARS}
 \end{align}
 where
 \begin{align}
 {s_{\mathrm{X}}} &= (H_{b2}{\boldsymbol{\alpha}}^{\mathrm{T}})^{\mathrm{T}} = {\boldsymbol{\alpha}}H_{b2}^{\mathrm{T}}\label{eqn:Synd_X_EARS}\\
  {s_{\mathrm{Z}}} &= (H_{b1}{\boldsymbol{\gamma}}^{\mathrm{T}})^{\mathrm{T}} = {\boldsymbol{\gamma}}H_{b1}^{\mathrm{T}}. \label{eqn:Synd_Z_EARS}
 \end{align}

We provide a syndrome computation circuit for the  quantum constacyclic BCH  code in Figure \ref{Syndrome}. It consists of the following operations:
\begin{itemize}
\item[(a)]$\mathrm{Q(FFFT)}$ on all the codewords qudits. Syndrome qudits are initially in state $\ket{0}$. The state of the qudits after the $\mathrm{Q}(\mathrm{FFFT})$ operation is performed is represented by $\ket{v_1}$ in Figure \ref{Syndrome}.
\item[(b)] $\mathrm{ADD}_{p^{k'}}$ operations are performed considering the qudits in $D_{\mathrm{Z}}$  as the control qudits and the first $(\delta-1)$ syndrome qudits as the target qudits. After these $\mathrm{ADD}_{p^{k'}}$ operations, the first $(\delta-1)$ syndrome qudits contain the syndrome $\ket{{s_{\mathrm{X}}}}$. The state of the qudits after the $\mathrm{ADD}_{p^{k'}}(1,2)$ operation is performed with the $\mathrm{Z}$-ancilla qudits in $D_Z$  and the syndrome qudits as the target is represented by $\ket{v_2}$ in Figure \ref{Syndrome}.
\item[(c)] $\mathrm{DFT}_{p^{k'}}$ operations are performed on all the qudits in blocks $D_{\mathrm{X}}$  to obtain the syndrome $\ket{{s_{\mathrm{Z}}}}$ over the qudits in these blocks. The state of the qudits after the $\mathrm{DFT}_{p^{k'}}$ operation is performed on the $\mathrm{X}$-ancilla qudits in $D_X$ is represented by $\ket{v_3}$ in Figure \ref{Syndrome}.
\item[(d)] $\mathrm{ADD}_{p^{k'}}$ operations are performed considering the qudits in $D_{\mathrm{X}}$ and the last $(\delta-1)$ syndrome qudits as the target qudits. After these $\mathrm{ADD}_{p^{k'}}$ operations, the last $(\delta-1)$ syndrome qudits contain the syndrome $\ket{{s_{\mathrm{Z}}}}$. The state of the qudits after the $\mathrm{ADD}_{p^{k'}}$ operation is performed with the $\mathrm{X}$-ancilla qudits in $D_X$  and the syndrome qudits as the target is represented by $\ket{v_4}$ in Figure \ref{Syndrome}.
\item[(e)] Perform the inverse operations $\mathrm{DFT}_{p^{k'}}^{-1}$ followed by $(\mathrm{Q(FFFT)})^{-1}$ on the corresponding qudits to obtain back $E\ket{\psi}$.
\end{itemize}

\begin{figure}
    \centering
    \includegraphics[scale=.38]{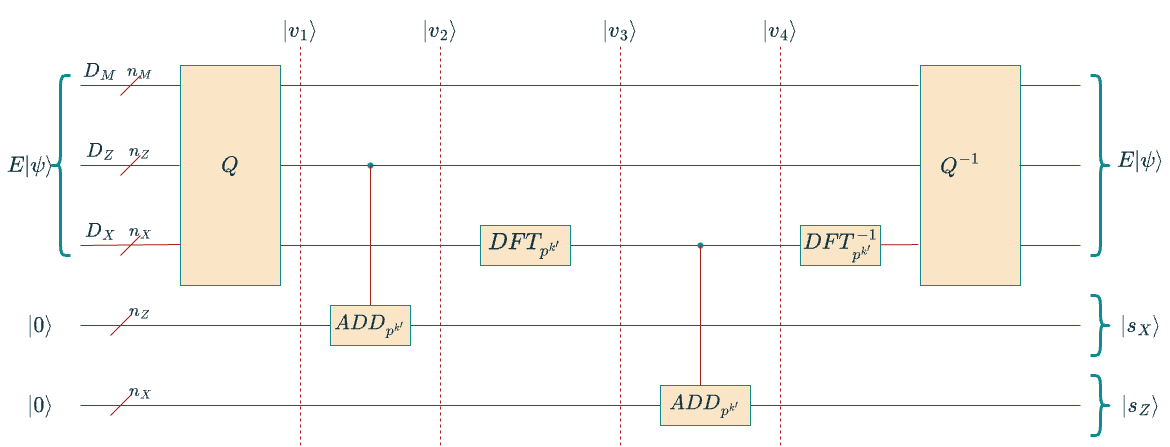}
    \caption{Syndrome computation circuit  for quantum constacyclic BCH codes.  The $n_m$, $n_{\mathrm{X}}$, and $n_{\mathrm{Z}}$ represent the number of qudits corresponding  to subblocks $D_M$, $D_{\mathrm{X}}$, and $D_{\mathrm{Z}}$, respectively. The syndrome computation involves performing $Q= \mathrm{Q}(\mathrm{FFFT})$, $\mathrm{ADD}_{p^{k'}}$, $\mathrm{DFT}_{p^{k'}}$, $\mathrm{DFT}_{p^{k'}}^{-1}$,  and $Q^{-1}$ operations on the codeword qudits along with $2(\delta-1)$ syndrome qudits that are initially in the state $\ket{0}$ each. At the end of the syndrome computation, the syndrome qudits are in the state $\ket{s} = \ket{s_{\mathrm{X}}}\otimes \ket{s_{\mathrm{Z}}}$, while the state of the codeword qudits is $E\ket{\psi}$.  $\ket{v_1}$, $\ket{v_2}$, $\ket{v_3}$, and $\ket{v_4}$ represent the intermediate states of the codeword qudits along with the syndrome qudits.}
    \label{Syndrome}
\end{figure}

Next, we will prove that the syndrome computation circuit transforms the state of the syndrome qudits into the syndrome state in equation \eqref{eqn:Synd_EARS} for the error of form $\mathrm{X}_i^{(p^{k'})}(\theta)\mathrm{Z}_i^{(p^{k'})}(\chi)$, where $i\in\{1,\dots,n\}$, $\theta,\chi\in\mathbb{F}_{p^{k'}}$, as all other basis errors can be written as the product of the errors of this form.  

For $E=\mathrm{X}_i^{(p^{k'})}(\theta)\mathrm{Z}_i^{(p^{k'})}(\chi)$, considering ${\boldsymbol{\alpha}} = \theta\mathbf{e}_i^{\mathrm{T}}$ and ${\boldsymbol{\gamma}} = \chi\mathbf{e}_i^{\mathrm{T}}$ in equations \eqref{eqn:Synd_EARS}, \eqref{eqn:Synd_X_EARS}, and \eqref{eqn:Synd_Z_EARS}, we obtain
\begin{align}
\ket{s} = \ket{{s_{\mathrm{X}}}}\otimes \ket{{s_{\mathrm{Z}}}},
\end{align}
where
\begin{align}
{s_{\mathrm{X}}} &= {\boldsymbol{\alpha}}H_{b2}^{\mathrm{T}} = \theta\mathbf{e}_i^{\mathrm{T}}H_{b2}^{\mathrm{T}}= [\theta(\beta\xi^{(b_2+j-1)})^{(i-1)}]_{j=1}^{\delta-1},\label{eqn:Synd_RS_X_BasisError}\\
{s_{\mathrm{Z}}} &= {\boldsymbol{\gamma}}H_{b1}^{\mathrm{T}} = \chi\mathbf{e}_i^{\mathrm{T}}H_{b1}^{\mathrm{T}}= [\chi(\beta\xi^{(b_1+j-1)})^{(i-1)}]_{j=1}^{\delta-1}.\label{eqn:Synd_RS_Z_BasisError}
\end{align}

From equation \eqref{eqn:Codeword}, the erroneous state is
 \begin{align}
 E\ket{\psi} \!\!=\! \mathrm{X}_i^{(p^{k'})}\!(\theta)\mathrm{Z}_i^{(p^{k'})}\!(\chi)((\mathrm{Q}(\mathrm{FFFT}))^{-1}\!)\!\ket{\psi_0}\!. \label{eqn:E_psi}
 \end{align}
  
The $n_s = 2(\delta-1)$ syndrome qudits are  initially in the state $\ket{0}$ each. After applying $Q = \mathrm{Q}(\mathrm{FFFT})$ on the first $n$ qudits of $E\ket{\psi}$, from equation \eqref{eqn:E_psi} and Figure \ref{Syndrome}, we obtain $\ket{v_1}$ to be
\begin{align}
&\ket{v_1}\nonumber\\
 &= (\mathrm{Q}(\mathrm{FFFT}) \otimes \mathrm{I}_{p^{k'}}^{\otimes n_s})(E\ket{\psi} \otimes \ket{0}^{\otimes n_s}),\nonumber\\
  &= ((\mathrm{Q}(\mathrm{FFFT}) )E\ket{\psi}) \otimes \ket{0}^{\otimes n_s},\nonumber\\
&= ((\mathrm{Q}(\mathrm{FFFT})\mathrm{X}_i^{(p^{k'})}\!(\theta)\mathrm{Z}_i^{(p^{k'})}\!(\chi)(\mathrm{Q}(\mathrm{FFFT}))^{-1}\!)\!\ket{\psi_0})\otimes\! \ket{0}^{\!\otimes n_s}\!,\nonumber\\
&= (((\mathrm{Q}(\mathrm{FFFT})\mathrm{X}_i^{(p^{k'})}\!(\theta)(\mathrm{Q}(\mathrm{FFFT}))^{-1})(\mathrm{Q}(\mathrm{FFFT})\mathrm{Z}_i^{(p^{k'})}\!(\chi)(\mathrm{Q}(\mathrm{FFFT}))^{-1})\!)\ket{\psi_0})\! \otimes\! \ket{0}^{\!\otimes n_s}. \label{eqn:ECC_mid1}
\end{align}
 Substituting relations (N1) and (N2) in equation \eqref{eqn:ECC_mid1}, we get
 \begin{align}
\ket{v_1}&= \bigg(\bigg(\left(\underset{j=0}{\overset{n-1}{\otimes}}\mathrm{X}^{(p^{k'})}(\theta(\beta\xi^{j})^{(i-1)})\right)\left(\underset{l=0}{\overset{n-1}{\otimes}}\mathrm{Z}^{(p^{k'})}(\frac{1}{n\lambda}\chi(\beta\xi^l)^{(n-i+1)})\right)\bigg)\ket{\psi_0}\bigg)\! \otimes\! \ket{0}^{\!\otimes n_s},\nonumber\\
&= \bigg(\bigg(\underset{j=0}{\overset{n-1}{\otimes}}\mathrm{X}^{(p^{k'})}(\theta(\beta\xi^{j})^{(i-1)})\mathrm{Z}^{(p^{k'})}(\frac{1}{n\lambda}\chi(\beta\xi^{j})^{(n-i+1)})\bigg)\ket{\psi_0}\bigg)\! \otimes\! \ket{0}^{\!\otimes n_s},\nonumber\\
&=\!\! \bigg(\!\!\bigg(\underset{g=1}{\overset{n}{\otimes}}\mathrm{X}^{(p^{k'})}(\theta(\beta\xi^{(g-1)})^{(i-1)})\mathrm{Z}^{(p^{k'})}(\frac{1}{n\lambda}\chi(\beta\xi^{(g-1)})^{(n-i+1)} \bigg)\ket{\psi_0}\bigg)\! \otimes\! \ket{0}^{\!\otimes n_s}, \label{eqn:ECC_mid2}
 \end{align}
where $g=(j+1)$.

From equation \eqref{eqn:InitRSCodeword}, $\ket{\psi_0}$ is a tensor product of states in different blocks of the codeword. The operator 
\begin{align}
E_{\mathrm{F}} :=\underset{g=1}{\overset{n}{\otimes}}\mathrm{X}^{(p^{k'})}(\theta(\beta\xi^{(g-1)})^{(i-1)})\mathrm{Z}^{(p^{k'})}(\frac{1}{n\lambda}\chi(\beta\xi^{(g-1)})^{(n-i+1)}\label{eqn:E_F}
\end{align}
 is also a tensor product of $n$ single qudit operators. Thus, in equation \eqref{eqn:ECC_mid2}, we can separately consider the message, $\mathrm{X}$-ancilla, and $\mathrm{Z}$-ancilla, qudits of $\ket{v_1}$ in blocks $D_M$, $D_X$, and $D_Z$, respectively,  to obtain
\begin{align}
\ket{v_1} = \ket{m_1}_{D_M}\ket{x_1}_{D_X}\ket{z_1}_{D_Z}\ket{0}^{\otimes n_s},\label{eqn:vartheta_1_state1}
\end{align}
where the subscript denotes the blocks where the states physically reside.
\subsection{Computation of the syndrome $\ket{{s_{\mathrm{X}}}}$}\label{sub61}

First, we show that the state of the $\mathrm{Z}$-ancilla qudits in block $D_Z$ for $\ket{v_1}$ are also components of the $\ket{{s_{\mathrm{X}}}}$. For the block $D_Z$, from Figure \ref{Syndrome}, we consider the state in $\ket{v_1}$. Thus, the state of a qudit in $D_Z$ in $\ket{v_1}$ is obtained by applying the part of operator $E_F$ from equation \eqref{eqn:E_F} corresponding to the $\mathrm{Z}$-ancilla qudits in the block $D_Z$ on  the initial state $\ket{0}$ of these qudits. Thus, for $m\in T_{\mathrm{Z}}$, we obtain the state $\ket{z^{(m)}}$ of the $m^{\mathrm{th}}$ qudit in $\ket{v_1}$ to be
\begin{align}
\ket{z^{(m)}} &=\mathrm{X}^{(p^{k'})}(\theta(\beta\xi^{(m-1)})^{(i-1)})\mathrm{Z}^{(p^{k'})}(\frac{1}{n\lambda}\chi(\beta\xi^{(m-1)})^{(n-i+1)} \ket{0}\\
&= \mathrm{X}^{(p^{k'})}(\theta(\beta\xi^{(m-1)})^{(i-1)})\ket{0}=\ket{\theta(\beta\xi^{(m-1)})^{(i-1)}}.\label{eqn:s_X_BZ}
\end{align}

From equation \eqref{eqn:P_Z}, $m \in T_{\mathrm{Z}} = \{b_2+1,\dots ,b_2+\delta-1\}$ and $(m-1)\in\{b_2,\dots ,b_2+\delta-2\}$. From equation \eqref{eqn:Synd_RS_X_BasisError}, the state 
$\ket{\theta(\beta\xi^{(m-1)})^{(i-1)}}$ over the $\mathrm{Z}$-ancilla qudit in $D_Z$ in equation \eqref{eqn:s_X_BZ} is a component of the $\ket{{s_{\mathrm{X}}}}$ part of the syndrome $\ket{s}$. The size of $D_Z$ containing the $\mathrm{Z}$-ancilla qudits is $n_\mathrm{Z} = \delta-1$. Thus, from  $D_Z$, we obtained all the $\delta-1$  qudits of $\ket{{s_{\mathrm{X}}}}$. 

The $\mathrm{ADD}_{p^{k'}}$ operations are performed considering the qudits in blocks $D_Z$  as the control qudits and the first $(\delta-1)$ syndrome qudits as the target qudits as shown in Figure \ref{Syndrome} to transform the state of these syndrome qudits to $\ket{{s_{\mathrm{X}}}}$. The state of all the qudits after these $\mathrm{ADD}_{p^{k'}}$ operations is $\ket{v_2}$.
\subsection{Computation of the syndrome $\ket{{s_{\mathrm{Z}}}}$}\label{sub62}
In this subsection, we obtain the syndrome $\ket{{s_{\mathrm{Z}}}}$ by performing $\mathrm{DFT}_{p^{k'}}$ operation on all the $\mathrm{X}$-ancilla qudits in block $D_X$  of $\ket{v_2}$. 

We introduce the following two relations based on $\mathrm{DFT}_{p^{k'}}$ that are proved in Appendix \ref{app:RelationProof}.
\begin{itemize}
\item[(N7)] $\mathrm{DFT}_{p^{k'}}\mathrm{X}^{(p^{k'})}(\gamma)\mathrm{DFT}_{p^{k'}}^{-1} =\mathrm{Z}^{(p^{k'})}(\gamma)$.
\item[(N8)] $\mathrm{DFT}_{p^{k'}}~\mathrm{Z}^{(p^{k'})}(\gamma)\mathrm{DFT}_{p^{k'}}^{-1} =\mathrm{X}^{(p^{k'})}(-\gamma)$.
\end{itemize}

For block $D_X$ that contains the $\mathrm{X}$-ancilla qudits, from Figure \ref{Syndrome}, the state in $\ket{v_2}$ and $\ket{v_1}$ are the same as no operation is performed on block $D_X$ in between. Thus, the state of an $\mathrm{X}$-ancilla qudit in $D_X$ in $\ket{v_2}$ is obtained by applying the part of operator $E_F$ in equation \eqref{eqn:E_F} corresponding to block $D_X$ on to the initial state $\ket{\epsilon}$ of these $\mathrm{X}$-ancilla qudits. Thus, for $m\in T_{\mathrm{X}}$, we obtain the state $\ket{x_2^{(m)}}$ of the $m^{\mathrm{th}}$ qudit in $\ket{v_2}$ to be
\begin{align}
\ket{x_2^{(m)}}&=\mathrm{X}^{(p^{k'})}(\theta(\beta\xi^{(m-1)})^{(i-1)})\mathrm{Z}^{(p^{k'})}(\frac{1}{n\lambda}\chi(\beta\xi^{(m-1)})^{(n-i+1)}) \ket{\epsilon},\nonumber\\
&=\mathrm{X}^{(p^{k'})}(\theta(\beta\xi^{(m-1)})^{(i-1)})\mathrm{Z}^{(p^{k'})}(\frac{1}{n\lambda}\chi(\beta\xi^{(m-1)})^{(n-i+1)} )\mathrm{DFT}_{p^{k'}}^{-1}\!\ket{0}\!.\label{eqn:s_Z_BX_mid1}
\end{align}

On applying $\mathrm{DFT}_{p^{k'}}$ on the $m^{\mathrm{th}}$ qudit, we obtain the state $\ket{x_3^{(m)}}$ in $\ket{v_3}$ to be
\allowdisplaybreaks
\begin{align}
\ket{x_3^{(m)}}
&=\mathrm{DFT}_{p^{k'}}\mathrm{X}^{(p^{k'})}(\theta(\beta\xi^{(m-1)})^{(i-1)})\mathrm{Z}^{(p^{k'})}(\frac{1}{n\lambda}\chi(\beta\xi^{(m-1)})^{(n-i+1)} )\mathrm{DFT}_{p^{k'}}^{-1}\!\ket{0}\!,\nonumber\\
&=(\mathrm{DFT}_{p^{k'}}\mathrm{X}^{(p^{k'})}(\theta(\beta\xi^{(m-1)})^{(i-1)}))\mathrm{DFT}_{p^{k'}}^{-1})~(\mathrm{DFT}_{p^{k'}}\mathrm{Z}^{(p^{k'})}(\frac{1}{n\lambda}\chi(\beta\xi^{(m-1)})^{(n-i+1)} ) \mathrm{DFT}_{p^{k'}}^{-1})\ket{0},\nonumber\\
&=\mathrm{Z}^{(p^{k'})}(\theta(\beta\xi^{(m-1)})^{(i-1)})\mathrm{X}^{(p^{k'})}(-\frac{1}{n\lambda}\chi(\beta\xi^{(m-1)})^{(n-i+1)})\ket{0},(\text{From (N7) and (N8)})\nonumber\\
&=\mathrm{Z}^{(p^{k'})}(\theta(\beta\xi^{(m-1)})^{(i-1)})\ket{-\frac{1}{n\lambda}\chi(\beta\xi^{(m-1)})^{(n-i+1)}},\nonumber\\
&=\omega^{\mathrm{Tr}_{p^{k'}/p}(\theta(\beta\xi^{(m-1)})^{(i-1)})(-\frac{1}{n\lambda}\chi(\beta\xi^{(m-1)})^{(n-i+1)}))}\ket{-\frac{1}{n\lambda}\chi(\beta\xi^{(m-1)})^{(n-i+1)}},\nonumber\\
&=\omega^{\mathrm{Tr}_{p^{k'}/p}(-\frac{1}{n\lambda}\theta\chi\beta^{n})}\ket{-\frac{1}{n\lambda}\chi(\beta\xi^{(m-1)})^{(n-i+1)}},\nonumber\\
&=\omega^{\mathrm{Tr}_{p^{k'}/p}(-\frac{\lambda}{n\lambda}\theta\chi)}\ket{-\frac{1}{n\lambda}\chi(\beta\xi^{(m-1)})^{(n-i+1)}},~(\because ~\beta^n=\lambda)\nonumber\\
&= \omega^{\mathrm{Tr}_{p^{k'}/p}(-\frac{1}{n}\theta\chi)}\ket{-\frac{\lambda}{n\lambda}\chi(\beta\xi^{(m-1)})^{-(i-1)}},~(\because ~\xi^n=1)\nonumber\\
&= \omega^{\mathrm{Tr}_{p^{k'}/p}(-\frac{1}{n}\theta\chi)}\ket{-\frac{1}{n}\chi(\beta\xi^{-(m-1)})^{(i-1)}},
\label{eqn:s_Z_BR}
\end{align}
as $\beta^{-1}=\beta$ and $\xi^{(-m+1)(i-1)} = \xi^{(n-m+1)(i-1)}$.

This gives $m\in T_\mathrm{X}$. Thus, from equation \eqref{eqn:P_X}, $m\in\{n-b_1-\delta+3,\dots,n-b_1+1\}$ and $(n-m+1)\in\{b_1,\dots, b_1+\delta-2\}$. From equation \eqref{eqn:Synd_RS_Z_BasisError}, the state 
$\ket{-\frac{1}{n}\chi(\beta\xi^{(n-m+1)})^{(i-1)}}$ over the $\mathrm{X}$-ancilla qudit in $D_X$ in equation \eqref{eqn:s_Z_BR} is a component of the $\ket{{s_{\mathrm{Z}}}}$ part of the syndrome $\ket{s}$, neglecting the phase $\omega^{\mathrm{Tr}_{p^{k'}/p}(-\frac{1}{n}\theta\chi)}$. Since the size of $D_X$ that contains the $\mathrm{X}$-ancilla qudits is $n_\mathrm{X} = \delta-1$. Thus, from  $D_X$, we obtained all the $\delta-1$ qudits of $\ket{{s_{\mathrm{Z}}}}$.

 The $\mathrm{ADD}_{p^{k'}}$ operations are performed considering the qudits in blocks $D_X$ as the control qudits and the last $(\delta-1)$ syndrome qudits as the target qudits as shown in Figure \ref{Syndrome} to transform the state of these syndrome qudits to $\ket{{s_{\mathrm{Z}}}}$. Thus, we have obtained the entire syndrome over the $2(\delta-1)$ syndrome qudits.  Moreover, the $\mathrm{X}$-ancilla qudits and the $\mathrm{Z}$-ancilla qudits in blocks $D_X$ and $D_Z$ contribute towards the syndrome corresponding to the $\mathrm{Z}^{(p^{k'})}(\cdot)$ and $\mathrm{X}^{(p^{k'})}(\cdot)$ errors, respectively. This follows from the fact that the syndromes for $\mathrm{Z}^{(p^{k'})}(\cdot)$ and $\mathrm{X}^{(p^{k'})}(\cdot)$ errors are based on $H_{b1}$ and $H_{b2}$, and the blocks $D_X$ and $D_Z$ were also defined based on these matrices in Subsection \ref{BCH}. Finally, we apply the inverse operations $\mathrm{DFT}_{p^{k'}}^{-1}$,  and $(\mathrm{Q}(\mathrm{FFFT}))^{-1}$ to transform the state of the codeword qudits back to $E\ket{\psi}$.

\subsection{Error deduction and recovery}\label{sub63}
After the syndrome computation procedure is performed, the syndrome qudits are measured. The measurement result from the first $n_{\mathrm{Z}}$ syndrome qudits is ${s_{\mathrm{X}}}$ and the measurement result from the rest $n_{\mathrm{X}}$ syndrome qudits is ${s_{\mathrm{Z}}}$.

 For an error $E = [{\boldsymbol{\alpha}}|\boldsymbol{{\gamma}}]$ on the first $n$ qudits of the code, where ${\boldsymbol{\alpha}}\boldsymbol{,{\gamma}} \in \mathbb{F}_{p^{k'}}^{n}$, from equations \eqref{eqn:Synd_X_EARS} and \eqref{eqn:Synd_Z_EARS}, ${s_{\mathrm{X}}}$ and ${s_{\mathrm{Z}}}$ are the syndromes of the errors ${\boldsymbol{\alpha}}$ and ${\boldsymbol{\gamma}}$ with respect to the classical constacyclic BCH codes whose parity check matrices are $H_{b2}$ and $H_{b1}$, respectively. Using the classical Berlekamp-Massey algorithm, Euclidean algorithm or our proposed classical constacyclic BCH  decoding algorithm in Section \ref{sec 3} or any other classical decoding algorithm on ${s_{\mathrm{X}}}$ and ${s_{\mathrm{Z}}}$, separately, the errors ${\boldsymbol{\alpha}}$ and ${\boldsymbol{\gamma}}$ are obtained. Let ${\boldsymbol{\alpha}} = [e_1~e_2~\dots~e_n]$ and ${\boldsymbol{\gamma}} = [f_1~f_2~\dots~f_n]$. Then, the corresponding quantum errors on the first $n$ qudits of the codeword are $\underset{i=1}{\overset{n}{\otimes}}\mathrm{X}^{(p^{k'})}(e_i)\mathrm{Z}^{(p^{k'})}(f_i)$. The inverse of these quantum errors are applied to the codeword qudits to recover $\ket{\psi}$.
\section{Conclusions}\label{conclu}
 Constacyclic codes over finite fields are characterized in the transform domain using finite field Fourier transform defined over an appropriate extension
field. We proposed an efficient spectral decoding for systematic constacyclic codes. The computational complexity of our proposed method for some parameters is smaller than the computational complexity of the classical syndrome-based decoding algorithms by a linear factor. We have constructed QECCs from constacyclic BCH codes in the spectral domain and provided encoding, syndrome computation and decoding circuits for these codes. Also, we constructively established through several examples that quantum constacyclic BCH codes are more efficient than repeated-root constacyclic codes over all cases of asymmetric and symmetric cyclotomic cosets. 
\section*{Acknowledgement}
S. Patel is supported under IISc-IoE Postdoctoral fellowship.
\bibliography{funbibfile}
\appendix
\section{Appendix}\label{app:RelationProof}
    In this Appendix, we derive a few relations that have been used in the main body of the paper. For $\gamma \in \mathbb{F}_{p^{k'}}$, the following relations hold true.
    \begin{remark}
    The following identities (N1)-(N4) are used to obtain the  initial stabilizers $\mathrm{X}_s^{(p^{k'})}(n\xi^l)$ and $\mathrm{Z}_t^{(p^{k'})}(\xi^m)$ in the simplest form.
\end{remark}

\begin{itemize}
\item[\textbf{(N1)}] \normalsize{${\mathrm{Q}(\mathrm{FFFT})\mathrm{X}_i^{(p^{k'})}(\gamma)(\mathrm{Q}(\mathrm{FFFT}))}^{{-1}}{=}\underset{{j=0}}{\overset{{n-1}}{{\otimes}}}{\mathrm{X}}^{{(p^{k'})}}({\gamma(\beta\xi^j)^{(i-1)}),}~{\forall}~{i\in}\{{1,}\dots,{n}\}$ }.
\begin{proof}
For $i\in\{1,\dots,n\}$ and $\gamma \in \mathbb{F}_{p^{k'}}$, we have
\allowdisplaybreaks
\begin{align}
\mathrm{Q}(\mathrm{FFFT})\mathrm{X}_i^{(p^{k'})}(\gamma)(\mathrm{Q}(\mathrm{FFFT}))^{-1}
&=\!\!\!\left(\!\!\underset{{\boldsymbol{\theta_1}}\in\mathbb{F}_{p^{k'}}^{n}}{\sum}\hspace{-0.01cm}\!\!\!\!\ket{\mathrm{FFFT}({\boldsymbol{\theta_1}})}\!\hspace{-0.02cm}\bra{{\boldsymbol{\theta_1}}}\!\!\hspace{-0.01cm}\right)\!\!\mathrm{X}_i^{(p^{k'})}\!(\gamma)\hspace{-0.01cm}\!\!\left(\!\underset{{\boldsymbol{\theta_2}}\in\mathbb{F}_{p^{k'}}^{n}}{\sum}\!\!\!\!\ket{{\boldsymbol{\theta_2}}}\!\hspace{-0.02cm}\bra{\mathrm{FFFT}({\boldsymbol{\theta_2}})\hspace{-0.01cm}}\!\!\!\right)\!\!, \nonumber\\
&=\underset{{\boldsymbol{\theta_1}},{\boldsymbol{\theta_2}}\in\mathbb{F}_{p^{k'}}^{n}}{\sum}\!\!\!\!\ket{\mathrm{FFFT}({\boldsymbol{\theta_1}})}\!\bra{\boldsymbol{\theta_1}}\mathrm{X}_i^{(p^{k'})}\!(\gamma)\ket{\boldsymbol{\theta_2}}\!\bra{\mathrm{FFFT}({\boldsymbol{\theta_2}})}, \nonumber\\
&=\underset{{\boldsymbol{\theta_1}},{\boldsymbol{\theta_2}}\in\mathbb{F}_{p^{k'}}^{n}}{\sum}\!\!\!\!\ket{\mathrm{FFFT}({\boldsymbol{\theta_1}})}\!\bra{\boldsymbol{\theta_1}}\ket{{\boldsymbol{\theta_2}}+\gamma\mathbf{e}_i^{\mathrm{T}}}\!\bra{\mathrm{FFFT}({\boldsymbol{\theta_2}})}, \nonumber\\
&=\underset{{\boldsymbol{\theta_1}},{\boldsymbol{\theta_2}}\in\mathbb{F}_{p^{k'}}^{n}}{\sum}\!\!\!\!\ket{\mathrm{FFFT}({\boldsymbol{\theta_1}})}\!\delta_{{\boldsymbol{\theta_1}},({\boldsymbol{\theta_2}}+\gamma\mathbf{e}_i^{\mathrm{T}})}\!\bra{\mathrm{FFFT}({\boldsymbol{\theta_2}})}, \nonumber\\
&=\underset{{\boldsymbol{\theta_2}}\in\mathbb{F}_{p^{k'}}^{n}}{\sum}\!\!\!\!\ket{\mathrm{FFFT}({\boldsymbol{\theta_2}}+\gamma\mathbf{e}_i^{\mathrm{T}})}\bra{\mathrm{FFFT}({\boldsymbol{\theta_2}})}, \nonumber\\
&=\underset{{\boldsymbol{\theta_2}}\in\mathbb{F}_{p^{k'}}^{n}}{\sum}\!\!\!\!\ket{\mathrm{FFFT}({\boldsymbol{\theta_2}})+\mathrm{FFFT}(\gamma\mathbf{e}_i^{\mathrm{T}})}\bra{\mathrm{FFFT}({\boldsymbol{\theta_2}})}, \label{eqn:FXF-1_mid1}
\end{align}
as $\mathrm{FFFT}(\cdot)$ is a linear function.\vspace{0.1cm}

Let ${\boldsymbol{\theta_2}'} = \mathrm{FFFT}({\boldsymbol{\theta_2}})$, then ${\boldsymbol{\theta_2}} = \mathrm{FFFT}^{-1}({\boldsymbol{\theta_2}'})$. As $\mathrm{FFFT}$ is an invertible function, for every ${\boldsymbol{\theta_2}} \in \mathbb{F}_{p^{k'}}^n$ we obtain a unique ${\boldsymbol{\theta_2}'} = \mathrm{FFFT}({\boldsymbol{\theta_2}}) \in \mathbb{F}_{p^{k'}}^n$. Thus, the set $\{{\boldsymbol{\theta_2}'} = \mathrm{FFFT}({\boldsymbol{\theta_2}})\}_{{\boldsymbol{\theta_2}}\in\mathbb{F}_{p^{k'}}^{n}}$ contains $p^{k'n}$ unique elements in $\mathbb{F}_{p^{k'}}^{n}$ as the set $\{{\boldsymbol{\theta_2}}\}_{{\boldsymbol{\theta_2}}\in\mathbb{F}_{p^{k'}}^{n}}$ contains $p^{k'n}$ elements. As the field $\mathbb{F}_{p^{k'}}^n$ has only $p^{k'n}$ elements, the set $\{{\boldsymbol{\theta_2}'} = \mathrm{FFFT}({\boldsymbol{\theta_2}})\}_{{\boldsymbol{\theta_2}}\in\mathbb{F}_{p^{k'}}^{n}}$ contains all the elements of $\mathbb{F}_{p^{k'}}^{n}$.\vspace{0.1cm}

 By substituting ${\boldsymbol{\theta_2}'}$ in equation \eqref{eqn:FXF-1_mid1}, the summation over ${\boldsymbol{\theta_2}}$ taking all the values of $\mathbb{F}_{p^{k'}}^n$ transforms into a summation over ${\boldsymbol{\theta_2}'}$ taking all the values of $\mathbb{F}_{p^{k'}}^n$. We will use this idea throughout to transform summations when we change summation variables. Since $\mathrm{FFFT}(\gamma\mathbf{e}_i^{\mathrm{T}}) = [\gamma\beta^{i-1}~\gamma(\beta\xi)^{(i-1)}~\dots~\gamma(\beta\xi^{(n-1)})^{(i-1)}]$. From equation \eqref{eqn:FXF-1_mid1}, we get
\begin{align}
&\mathrm{Q}(\mathrm{FFFT})\mathrm{X}_i^{(p^{k'})}(\beta)(\mathrm{Q}(\mathrm{FFFT}))^{-1}\nonumber\\
&=\underset{{\boldsymbol{\theta_2}'}\in\mathbb{F}_{p^{k'}}^{n}}{\sum}\!\!\!\!\ket{{\boldsymbol{\theta_2}'}+[\gamma\beta^{i-1}~\gamma(\beta\xi)^{(i-1)}~\dots~\gamma(\beta\xi^{(n-1)})^{(i-1)}]}\bra{{\boldsymbol{\theta_2}'}},\nonumber\\
&=\underset{{\boldsymbol{\theta_2}'}\in\mathbb{F}_{p^{k'}}^{n}}{\sum}\!\!\!\!\left(\underset{j=0}{\overset{n-1}{\otimes}}\mathrm{X}^{(p^{k'})}(\gamma(\beta\xi^j)^{(i-1)})\right) \ket{{\boldsymbol{\theta_2}'}}\bra{{\boldsymbol{\theta_2}'}},\nonumber\\ 
&=\left(\underset{j=0}{\overset{n-1}{\otimes}}\mathrm{X}^{(p^{k'})}(\gamma(\beta\xi^j)^{(i-1)}\right)\underset{{\boldsymbol{\theta_2}'}\in\mathbb{F}_{p^{k'}}^{n}}{\sum}\!\!\!\! \ket{{\boldsymbol{\theta_2}'}}\bra{{\boldsymbol{\theta_2}'}}.\label{eqn:FXF-1_mid2} 
\end{align} 
By using the completeness relation of quantum mechanics, $\underset{{\boldsymbol{\theta_2}'}\in\mathbb{F}_{p^{k'}}^{n}}{\sum}\!\!\!\! \ket{{\boldsymbol{\theta_2}'}}\bra{{\boldsymbol{\theta_2}'}} = \mathrm{I}_{p^{k'}}^{\otimes n}$, we have
\begin{align*}
&\mathrm{Q}(\mathrm{FFFT})\mathrm{X}_i^{(p^{k'})}(\gamma)(\mathrm{Q}(\mathrm{FFFT}))^{-1}=\underset{j=0}{\overset{n-1}{\otimes}}\mathrm{X}^{(p^{k'})}(\gamma(\beta\xi^j)^{(i-1)}).
\end{align*}
\end{proof}
\item[\textbf{(N2)}] \normalsize{${\mathrm{Q}(\mathrm{FFFT})\mathrm{Z}_i^{(p^{k'})}(\gamma)(\mathrm{Q}(\mathrm{FFFT}))}^{{-1}}{=}\left(\underset{{j=0}}{\overset{{n-1}}{{\otimes}}}{\mathrm{Z}}^{{(p^{k'})}}{(\frac{1}{n\lambda}\gamma(\beta\xi^j}^{{(n-i+1)}}))\right){,}~{\forall}~{i\in}\{{1,}\dots,{n}\}$}.
\begin{proof}
For $i\in\{1,\dots,n\}$, $\gamma \in \mathbb{F}_{p^{k'}}$ ,and ${\boldsymbol{\theta_2}}\in\mathbb{F}_{p^{k'}}^n$, we note that
\begin{align}
\mathrm{Z}_i^{(p^{k'})}\!(\gamma)\ket{{\boldsymbol{\theta_2}}} = \omega^{\mathrm{Tr}_{p^{k'}/p}({\boldsymbol{\theta_2}}\gamma\mathbf{e}_i^{\mathrm{T}})}\ket{{\boldsymbol{\theta_2}}}, \label{eqn:R2_mid1_EARS}
\end{align}
where $\omega^{\mathrm{Tr}_{p^{k'}/p}({\boldsymbol{\theta_2}}\gamma\mathbf{e}_i^{\mathrm{T}})}$ is obtained as the operator $\mathrm{Z}_i^{(p^{k'})}\!(\gamma)$ performs $\mathrm{Z}^{(p^{k'})}\!(\gamma)$ on the $i^{\mathrm{th}}$ qudit.

For $i\in\{1,\dots,n\}$, we consider
\begin{align}
&\mathrm{Q}(\mathrm{FFFT})\mathrm{Z}_i^{(p^{k'})}(\gamma)(\mathrm{Q}(\mathrm{FFFT}))^{-1}\nonumber\\
&=\!\left(\!\underset{{\boldsymbol{\theta_1}}\in\mathbb{F}_{p^{k'}}^{n}}{\sum}\!\!\!\!\ket{\mathrm{FFFT}({\boldsymbol{\theta_1}})}\!\!\bra{{\boldsymbol{\theta_1}}}\!\!\right)\!\!\mathrm{Z}_i^{(p^{k'})}\!(\gamma)\!\!\left(\!\underset{{\boldsymbol{\theta_2}}\in\mathbb{F}_{p^{k'}}^{n}}{\sum}\!\!\!\!\ket{{\boldsymbol{\theta_2}}}\!\!\bra{\mathrm{FFFT}({\boldsymbol{\theta_2}})}\!\!\right)\!\!, \nonumber\\
&=\!\underset{{\boldsymbol{\theta_1}},{\boldsymbol{\theta_2}}\in\mathbb{F}_{p^{k'}}^{n}}{\sum}\!\!\!\!\!\!\!\ket{\mathrm{FFFT}({\boldsymbol{\theta_1}})}\!\bra{\boldsymbol{\theta_1}}\mathrm{Z}_i^{(p^{k'})}\!(\gamma)\ket{\boldsymbol{\theta_2}}\!\bra{\mathrm{FFFT}({\boldsymbol{\theta_2}})}\!.\label{eqn:R2_mid2_EARS}
\end{align}

Substituting equation \eqref{eqn:R2_mid1_EARS} in equation \eqref{eqn:R2_mid2_EARS}, we obtain
\begin{align}
&\mathrm{Q}(\mathrm{FFFT})\mathrm{Z}_i^{(p^{k'})}(\gamma)(\mathrm{Q}(\mathrm{FFFT}))^{-1}\nonumber\\
&=\underset{{\boldsymbol{\theta_1}},{\boldsymbol{\theta_2}}\in\mathbb{F}_{p^{k'}}^{n}}{\sum}\!\!\!\!\!\!\omega^{\mathrm{Tr}_{p^{k'}/p}({\boldsymbol{\theta_2}}\gamma\mathbf{e}_i^{\mathrm{T}})}\ket{\mathrm{FFFT}({\boldsymbol{\theta_1}})}\!\bra{\boldsymbol{\theta_1}}\ket{\boldsymbol{\theta_2}}\!\bra{\mathrm{FFFT}({\boldsymbol{\theta_2}})}, \nonumber\\
&=\underset{{\boldsymbol{\theta_1}},{\boldsymbol{\theta_2}}\in\mathbb{F}_{p^{k'}}^{n}}{\sum}\!\!\!\!\!\!\omega^{\mathrm{Tr}_{p^{k'}/p}({\boldsymbol{\theta_2}}\gamma\mathbf{e}_i^{\mathrm{T}})}\ket{\mathrm{FFFT}({\boldsymbol{\theta_1}})}\!\delta_{{\boldsymbol{\theta_1}},{\boldsymbol{\theta_2}}}\!\bra{\mathrm{FFFT}({\boldsymbol{\theta_2}})}, \nonumber\\
&=\underset{{\boldsymbol{\theta_2}}\in\mathbb{F}_{p^{k'}}^{n}}{\sum}\!\!\omega^{\mathrm{Tr}_{p^{k'}/p}({\boldsymbol{\theta_2}}\gamma\mathbf{e}_i^{\mathrm{T}})}\ket{\mathrm{FFFT}{\boldsymbol{\theta_2}}}\bra{\mathrm{FFFT}({\boldsymbol{\theta_2}})}, \nonumber\\
&=\underset{{\boldsymbol{\theta_2}}\in\mathbb{F}_{p^{k'}}^{n}}{\sum}\!\!\omega^{\mathrm{Tr}_{p^{k'}/p}({\boldsymbol{\theta_2}}\gamma\mathbf{e}_i^{\mathrm{T}})}\ket{\mathrm{FFFT}({\boldsymbol{\theta_2}})}\bra{\mathrm{FFFT}({\boldsymbol{\theta_2}})}. \label{eqn:FZF-1_mid1}
\end{align}
Let ${\boldsymbol{\theta_2}'} = \mathrm{FFFT}({\boldsymbol{\theta_2}})$. Let ${\boldsymbol{\theta_2}'}=[T_0~T_1~\dots~T_{n-1}]$. From equation \eqref{eqn:FFFT_inv}, we have
\begin{align}
&{\boldsymbol{\theta_2}} \!=\! \mathrm{FFFT}^{-1}({\boldsymbol{\theta_2}'})\!=
\!\!=\! \left[\frac{1}{n\beta^n}\underset{j=0}{\overset{n-1}{\sum}}T_j(\beta\xi^j)^{(n-m)}\right]_{m=0}^{n-1}.\label{eqn:R2_mid3_EARS}
\end{align}
Let $l = (m+1)$, then $l$ ranges from $1$ to $n$ as $m$ ranges from $0$ to $n-1$. Substituting $l=(m+1)$ in equation \eqref{eqn:R2_mid3_EARS},  we get
\begin{align}
&{\boldsymbol{\theta_2}} = \mathrm{FFFT}^{-1}({\boldsymbol{\theta_2}'})= \left[\frac{1}{n\lambda}\underset{j=0}{\overset{n-1}{\sum}}T_j(\beta\xi^j)^{(n-l+1)}\right]_{l=1}^{n}, ~\text{where} \beta^n=\lambda \nonumber\\
\Rightarrow &\mathrm{Tr}_{p^{k'}/p}({\boldsymbol{\theta_2}}\gamma\mathbf{e}_i^{\mathrm{T}}) = \mathrm{Tr}_{p^{k'}/p}(\mathrm{FFFT}^{-1}({\boldsymbol{\theta_2}'})\gamma\mathbf{e}_i^{\mathrm{T}}),\nonumber\\
& =\mathrm{Tr}_{p^{k'}/p}\left(\left[\frac{1}{n\lambda }\underset{j=0}{\overset{n-1}{\sum}}T_j(\beta\xi^j)^{(n-l+1)}\right]_{l=1}^{n}\!\!\!\gamma\mathbf{e}_i^{\mathrm{T}}\right),\nonumber\\
& =\mathrm{Tr}_{p^{k'}/p}\left(\frac{1}{n\lambda }\underset{j=0}{\overset{n-1}{\sum}}\gamma T_j(\beta\xi^j)^{(n-i+1)}\right).\label{eqn:FZF-1_mid2}
\end{align}
 Substituting ${\boldsymbol{\theta_2}'} = \mathrm{FFFT}({\boldsymbol{\theta_2}})$ and equation \eqref{eqn:FZF-1_mid2} in equation \eqref{eqn:FZF-1_mid1}, we obtain
\begin{align}
&\mathrm{Q}(\mathrm{FFFT})\mathrm{Z}_i^{(p^{k'})}(\gamma)(\mathrm{Q}(\mathrm{FFFT}))^{-1}\nonumber\\
&=\underset{{\boldsymbol{\theta_2}'}\in\mathbb{F}_{p^{k'}}^{n}}{\sum}\!\!\!\!\omega^{\mathrm{Tr}_{p^{k'}/p}\left(\frac{1}{n\lambda }\underset{j=0}{\overset{n-1}{\sum}}\gamma T_j(\beta\xi^j)^{(n-i+1)}\right)}\ket{{\boldsymbol{\theta_2}'}}\bra{{\boldsymbol{\theta_2}'}},\nonumber\\
&=\!\!\!\!\!\!\underset{T_0,\dots,T_{n-1}\in\mathbb{F}_{p^{k'}}}{\sum}\!\!\!\!\!\!\!\!\!\!\omega^{\!\mathrm{Tr}_{\!p^{k'}\!/\!p}\!\left(\frac{1}{n\lambda }\underset{j=0}{\overset{n-1}{\sum}}\gamma T_j(\beta\xi^j)^{(n-i+1)}\!\!\right)}\!\!\ket{T_0\! \dots\! T_{n-1}}\!\bra{T_0\! \dots\! T_{n-1}},\nonumber\\
&=\!\!\!\!\!\!\underset{T_0,\dots,T_{n-\!1}\in\mathbb{F}_{p^{k'}}}{\sum}\!\!\!\!\left(\underset{j=0}{\overset{n-1}{\otimes}}\mathrm{Z}^{(p^{k'})}\!(\frac{1}{n\lambda }\gamma (\beta\xi^j)^{(n-i+1)}\!)\!\!\right)\!\!\ket{T_0\! \dots\! T_{n-1}}\!\bra{T_0\! \dots\! T_{n-1}},\nonumber\\
&=\!\!\left(\underset{j=0}{\overset{n-1}{\otimes}}\mathrm{Z}^{(p^{k'})}\!(\frac{1}{n\lambda }\gamma (\beta\xi^j)^{(n-i+1)}\!)\!\!\right)\!\!\!\!\underset{T_0,\dots,T_{n-\!1}\in\mathbb{F}_{p^{k'}}}{\sum}\!\!\!\!\!\!\!\!\ket{T_0\! \dots\! T_{n-1}}\!\bra{T_0\! \dots\! T_{n-1}},\nonumber\\
&=\left(\underset{j=0}{\overset{n-1}{\otimes}}\mathrm{Z}^{(p^{k'})}\!(\frac{1}{n\lambda }\gamma (\beta\xi^j)^{(n-i+1)}\!)\right)\underset{{\boldsymbol{\theta_2}'}\in\mathbb{F}_{p^{k'}}^{n}}{\sum}\!\!\!\! \ket{{\boldsymbol{\theta_2}'}}\bra{{\boldsymbol{\theta_2}'}}.\label{eqn:FZF-1_mid3} 
\end{align} 
By using the completeness relation of quantum mechanics, $\underset{{\boldsymbol{\theta_2}'}\in\mathbb{F}_{p^{k'}}^{n}}{\sum}\!\!\!\! \ket{{\boldsymbol{\theta_2}'}}\bra{{\boldsymbol{\theta_2}'}} = \mathrm{I}_{p^{k'}}^{\otimes n}$,. Thus, from equation \eqref{eqn:FZF-1_mid3}, we get
\begin{align*}
&\mathrm{Q}(\mathrm{FFFT})\mathrm{Z}_i^{(p^{k'})}\!(\gamma)(\mathrm{Q}(\mathrm{FFFT}))^{\!-1}\!\!=\!\!\left(\underset{j=0}{\overset{n-1}{\otimes}}\mathrm{Z}^{(p^{k'})}\!(\frac{1}{n\lambda}\gamma (\beta\xi^j)^{(n-i+1)})\!\!\right)\!.
\end{align*}
\end{proof}
\item[\textbf{(N3)}] $\mathrm{Q(FFFT)}\left(\underset{{i=0}}{\overset{{n-1}}{{\otimes}}}{\mathrm{X}}^{{(p^{k'})}}{(\gamma(\beta\xi^j)^i})\right){(\mathrm{Q(FFFT))}}^{{-1}}={\mathrm{X}}_{{(n-j+1)}}^{{(p^{k'})}}{(n\gamma),}~{\forall}~{j\in}\{{1,}\dots,{n}\}$.
\begin{proof}

For $j\in\{1,\dots, n\}$ and $\gamma \in \mathbb{F}_{p^{k'}}$, we have
\begin{align}
\underset{i=0}{\overset{n-1}{\otimes}}\mathrm{X}^{(p^{k'})}(\gamma(\beta\xi^j)^i) &= \mathrm{X}^{(p^{k'})}(\gamma) \otimes \mathrm{X}^{(p^{k'})}(\gamma(\beta\xi^j)) \otimes \cdots \otimes \mathrm{X}^{(p^{k'})}(\gamma(\beta\xi^{j})^{n-1}).\label{eqn:R3_mid2_EARS}
\end{align}

From equation \eqref{eqn:R3_mid2_EARS}, the operator $\underset{i=0}{\overset{n-1}{\otimes}}\mathrm{X}^{(p^{k'})}(\gamma(\beta\xi^j)^i)$ performs the operation $\mathrm{X}^{(p^{k'})}(\gamma(\beta\xi^j)^i)$ on the $(i+1)^{\mathrm{th}}$ qudit, for all $i\in\{0,\dots,n-1\}$. For $s\in\{1,\dots, n\}$ and $\gamma\in\mathbb{F}_{p^{k'}}$, the operator $\mathrm{X}^{(p^{k'})}_s(\gamma)$ performs $\mathrm{X}^{(p^{k'})}(\gamma)$ on the $s^{\mathrm{th}}$ qudit. Thus, from equation \eqref{eqn:R3_mid2_EARS}, we have
\begin{align}
\underset{i=0}{\overset{n-1}{\otimes}}\mathrm{X}^{(p^{k'})}(\gamma(\beta\xi^j)^i) = \overset{n-1}{\underset{i=0}{\prod}} \mathrm{X}^{(p^{k'})}_{(i+1)}(\gamma(\beta\xi^j)^i). \label{eqn:R3_mid3_EARS}
\end{align}

Using equation \eqref{eqn:R3_mid3_EARS}, we have
\allowdisplaybreaks
\begin{align}
&\mathrm{Q(FFFT)}\left(\underset{i=0}{\overset{n-1}{\otimes}}\mathrm{X}^{(p^{k'})}(\gamma(\beta\xi^j)^i)\right)(\mathrm{Q(FFFT))}^{-1}\nonumber\\
&=\mathrm{Q(FFFT)}\left(\overset{n-1}{\underset{i=0}{\prod}} \mathrm{X}^{(p^{k'})}_{(i+1)}(\gamma(\beta\xi^j)^i)\right)(\mathrm{Q(FFFT))}^{-1},\nonumber\\
&=\mathrm{Q}(\mathrm{FFFT})\mathrm{X}_1^{(p^{k'})}(\gamma)\cdots\mathrm{X}_n^{(p^{k'})}\!(\gamma(\beta\xi^{j})^{(n-1)})(\mathrm{Q}(\mathrm{FFFT}))^{\!-1}\!\!,\nonumber\\
&=\mathrm{Q}(\mathrm{FFFT})\mathrm{X}_1^{(p^{k'})}(\gamma)(\mathrm{Q}(\mathrm{FFFT}))^{\!-1}\mathrm{Q}(\mathrm{FFFT})\mathrm{X}_2^{(p^{k'})}(\gamma(\beta\xi^j))\nonumber\\
&~~~(\mathrm{Q}(\mathrm{FFFT}))^{\!-1}\!\!\cdots\mathrm{Q}(\mathrm{FFFT})\mathrm{X}_n^{(p^{k'})}\!(\gamma(\beta\xi^{j})^{(n-1)})(\mathrm{Q}(\mathrm{FFFT}))^{\!-1}\!\!,\nonumber\\
&= \underset{i=0}{\overset{n-1}{\prod}}\mathrm{Q(FFFT)}\mathrm{X}_{(i+1)}^{(p^{k'})}(\gamma(\beta\xi^j)^i)(\mathrm{Q(FFFT)})^{-1},\nonumber\\
&= \underset{i=0}{\overset{n-1}{\prod}}\underset{m=0}{\overset{n-1}{\otimes}}\mathrm{X}^{(p^{k'})}(\gamma(\beta\xi^j)^i(\beta\xi^m)^{(i+1-1)}),~~~(\text{From}~(N1))\nonumber\\
&= \underset{i=0}{\overset{n-1}{\prod}}\underset{m=0}{\overset{n-1}{\otimes}}\mathrm{X}^{(p^{k'})}(\gamma\beta^{i}(\beta\xi^{j+m})^i),\nonumber\\
&= \underset{i=0}{\overset{n-1}{\prod}}\underset{m=0}{\overset{n-1}{\otimes}}\mathrm{X}^{(p^{k'})}(\gamma\beta^{n-i-1}\beta^{-n+1}(\beta\xi^{j+m})^i),~~~(\text{since}~ \beta=\beta^{-1})\nonumber\\
&= \underset{m=0}{\overset{n-1}{\otimes}}\underset{i=0}{\overset{n-1}{\prod}}\mathrm{X}^{(p^{k'})}(\gamma\beta^{-n+1}\beta^{n-i-1}(\beta\xi^{j+m})^i),\nonumber\\
&= \underset{m=0}{\overset{n-1}{\otimes}}\mathrm{X}^{(p^{k'})}\left(\underset{i=0}{\overset{n-1}{\sum}}\gamma\beta^{-n+1}\beta^{n-i-1}(\beta\xi^{j+m})^i\right),\nonumber\\
&~~~~(\because \mathrm{X}^{(p^{k'})}(\alpha)\mathrm{X}^{(p^{k'})}(\gamma) = \mathrm{X}^{(p^{k'})}(\alpha+\gamma))\nonumber\\
&= \underset{g=1}{\overset{n}{\otimes}}\mathrm{X}^{(p^{k'})}\left(\underset{i=0}{\overset{n-1}{\sum}}\gamma\beta^{-n+1}\beta^{n-i-1}(\beta\xi^{j+g-1})^i\right),\text{ where }g = (m+1).\label{eqn:R3_mid4_EARS}
\end{align}

As $\xi^n = 1$, we note that
\begin{align}
\underset{i=0}{\overset{n-1}{\sum}}\beta^{n-i-1}(\beta\xi^{j+g-1})^i&= \begin{cases}
\underset{i=0}{\overset{n-1}{\sum}}\beta^{n-1}\xi^{0} & (g+j-1)~\mathrm{mod}~n=0\\
\underset{i=0}{\overset{n-1}{\sum}}\beta^{n-i-1}(\beta\xi^{j+g-1})^i & (g+j-1)~\mathrm{mod}~n\neq 0
\end{cases},\nonumber\\
&= \begin{cases}
n\beta^{n-1} & (g+j-1)~\mathrm{mod}~n=0\\
0 & (g+j-1)~\mathrm{mod}~n\neq 0\end{cases},\label{eqn:R3_mid5_EARS}
\end{align}
as $\xi^{n}=1$ and the roots of unity sum to $0$ (from Lemma \ref{lem1}). From equation \eqref{eqn:R3_mid5_EARS}, we have
\begin{align}
&\underset{j=0}{\overset{n-1}{\sum}}\beta^{n-i-1}(\beta\xi^{j+g-1})^i)= n\beta^{n-1}\delta_{(g+j-1),n} = n\beta^{n-1}\delta_{g,(n-j+1)}.\label{eqn:R3_mid6_EARS}
\end{align}

Substituting equation \eqref{eqn:R3_mid6_EARS} in equation \eqref{eqn:R3_mid4_EARS}, we obtain

\begin{align}
&\mathrm{Q(FFFT)}\left(\underset{i=0}{\overset{n-1}{\otimes}}\mathrm{X}^{(p^{k'})}(\gamma(\beta\xi^j)^{i})\right)(\mathrm{Q(FFFT))}^{-1}\nonumber\\
&= \underset{g=1}{\overset{n}{\otimes}}\mathrm{X}^{(p^{k'})}\left(\gamma\beta^{-n+1} n\beta^{n-1}\delta_{g,n-j+1}\right)= \mathrm{X}_{(n-j+1)}^{(p^{k'})}\left(n\gamma\right),\nonumber
\end{align}
as $j \in \{1,\dots,n\}$, this gives $(n-j+1)\in\{1,\dots,n\}$, and in the tensor product $\underset{g=1}{\overset{n}{\otimes}}\mathrm{X}^{(p^{k'})}\left(\gamma n\delta_{g,n-j+1}\right)$, only the $(n-j+1)^{\mathrm{st}}$ element is non-identity ($\because g= (n-j+1)$), while the rest become identity ($\because g\neq  (n-j+1) \Rightarrow \mathrm{X}^{(p^{k'})}(0)=\mathrm{I}_{p^{k'}}$). 
\end{proof}

\item[\textbf{(N4)}] $\mathrm{Q(FFFT)}\left(\underset{{i=0}}{\overset{{n-1}}{{\otimes}}}{\mathrm{Z}}^{{(p^{k'})}}{(\gamma(\beta\xi^j)^{i})}\right){(\mathrm{Q(FFFT))}}^{{-1}} ={\mathrm{Z}}_{{(j+1)}}^{{(p^{k'})}}{(\gamma),}~{\forall}~{j\in}\{{0,}\dots,{n-1}\}$.
\begin{proof}
For $j\in\{0,\dots, n-1\}$ and $\gamma \in \mathbb{F}_{p^{k'}}$, we have
\begin{align}
&\underset{i=0}{\overset{n-1}{\otimes}}\mathrm{Z}^{(p^{k'})}(\gamma(\beta\xi^j)^{i})= \mathrm{Z}^{(p^{k'})}(\gamma )\otimes \mathrm{Z}^{(p^{k'})}(\gamma(\beta\xi^j)) \otimes \cdots \otimes \mathrm{Z}^{(p^{k'})}(\gamma(\beta\xi^{j})^{(n-1)}).\label{eqn:R4_mid2_EARS}
\end{align}

From equation \eqref{eqn:R4_mid2_EARS}, the operator $\underset{j=0}{\overset{n-1}{\otimes}}\mathrm{Z}^{(p^{k'})}(\gamma(\beta\xi^j)^{i})$ performs the operation $\mathrm{Z}^{(p^{k'})}(\gamma(\beta\xi^j)^{i})$ on the $(i+1)^{\mathrm{st}}$ qudit, for all $i\in\{0,\dots,n-1\}$. For $s\in\{1,\dots, n\}$ and $\gamma\in\mathbb{F}_{p^{k'}}$, the operator $\mathrm{Z}^{(p^{k'})}_s(\gamma)$ performs $\mathrm{Z}^{(p^{k'})}(\gamma)$ on the $s^{\mathrm{th}}$ qudit. Thus, from equation \eqref{eqn:R3_mid2_EARS}, we have
\begin{align}
\underset{i=0}{\overset{n-1}{\otimes}}\mathrm{Z}^{(p^{k'})}(\gamma(\beta\xi^j)^{i}) = \overset{n-1}{\underset{i=0}{\prod}} \mathrm{Z}^{(p^{k'})}_{(i+1)}(\gamma(\beta\xi^j)^{i}). \label{eqn:R4_mid3_EARS}
\end{align}

Using equation \eqref{eqn:R4_mid3_EARS}, we have
\allowdisplaybreaks
\begin{align}
&\mathrm{Q(FFFT)}\left(\underset{i=0}{\overset{n-1}{\otimes}}\mathrm{Z}^{(p^{k'})}(\gamma(\beta\xi^j)^{i})\right)(\mathrm{Q(FFFT))}^{-1}\nonumber\\
&=\mathrm{Q(FFFT)}\left(\overset{n-1}{\underset{i=0}{\prod}} \mathrm{Z}^{(p^{k'})}_{(i+1)}(\gamma(\beta\xi^j)^{i})\right)(\mathrm{Q(FFFT))}^{-1},\nonumber\\
&=\mathrm{Q}(\mathrm{FFFT})\mathrm{Z}_1^{(p^{k'})}(\gamma)\cdots\mathrm{Z}_n^{(p^{k'})}\!(\gamma(\beta\xi^{j})^{n-1}\!)(\mathrm{Q}(\mathrm{FFFT}))^{\!-1}\!\!,\nonumber\\
&=\mathrm{Q}(\mathrm{FFFT})\mathrm{Z}_1^{(p^{k'})}(\gamma)(\mathrm{Q}(\mathrm{FFFT}))^{\!-1}\mathrm{Q}(\mathrm{FFFT})\mathrm{Z}_2^{(p^{k'})}(\gamma(\beta\xi^{j}))\nonumber\\
&~~~(\mathrm{Q}(\mathrm{FFFT}))^{\!-1}\!\!\cdots\mathrm{Q}(\mathrm{FFFT})\mathrm{Z}_n^{(p^{k'})}\!(\gamma(\beta\xi^{j})^{n-1}\!)(\mathrm{Q}(\mathrm{FFFT}))^{\!-1}\!\!,\nonumber\\
&= \underset{i=0}{\overset{n-1}{\prod}}\mathrm{Q(FFFT)}\mathrm{Z}_{(i+1)}^{(p^{k'})}(\gamma(\beta\xi^{j})^{i})(\mathrm{Q(FFFT)})^{-1},\nonumber\\
&= \underset{i=0}{\overset{n-1}{\prod}}\underset{m=0}{\overset{n-1}{\otimes}}\!\!\mathrm{Z}^{(p^{k'})}(\frac{1}{n\lambda}\gamma(\beta\xi^{j})^{i}(\beta\xi^{m})^{(n-l+1)}),(\text{From}~(N2)~\text{and}~i+1=l)\nonumber\\
&= \underset{i=0}{\overset{n-1}{\prod}}\underset{m=0}{\overset{n-1}{\otimes}}\mathrm{Z}^{(p^{k'})}(\frac{1}{n\lambda}\gamma\beta^{(n-i)}(\beta\xi^{(j-m)})^i),~(\because \xi^n=1, ~\beta=\beta^{-1})\nonumber\\
&= \underset{m=0}{\overset{n-1}{\otimes}}\underset{i=0}{\overset{n-1}{\prod}}\mathrm{Z}^{(p^{k'})}(\frac{1}{n\lambda}\beta\gamma\beta^{(n-i-1)}(\beta\xi^{(j-m)})^i),\nonumber\\
&= \underset{m=0}{\overset{n-1}{\otimes}}\mathrm{Z}^{(p^{k'})}\left(\underset{i=0}{\overset{n-1}{\sum}}\frac{1}{n\lambda}\beta\gamma\beta^{(n-i-1)}(\beta\xi^{(j-m)})^i\right),\nonumber\\
&~~~~(\because \mathrm{Z}^{(p^{k'})}(\alpha)\mathrm{Z}^{(p^{k'})}(\gamma) = \mathrm{Z}^{(p^{k'})}(\alpha+\gamma))\nonumber\\
&= \underset{g=1}{\overset{n}{\otimes}}\mathrm{Z}^{(p^{k'})}\!\!\left(\underset{j=0}{\overset{n-1}{\sum}}\frac{1}{n\lambda}\beta\gamma\beta^{(n-i-1)}(\beta\xi^{(j-(g-1))})^i\!\!\right)\!\!,\text{where }g = (m\!+\!1),\nonumber\\
&= \underset{g=1}{\overset{n}{\otimes}}\mathrm{Z}^{(p^{k'})}\left(\frac{1}{n\lambda}\beta\gamma\underset{j=0}{\overset{n-1}{\sum}}\beta^{(n-i-1)}(\beta\xi^{(j-(g-1))})^i\right).\label{eqn:R4_mid4_EARS}
\end{align}

As $\xi^n = 1$, we note that
\begin{align}
\underset{j=0}{\overset{n-1}{\sum}}\beta^{(n-i-1)}(\beta\xi^{(j-(g-1))})^i
 &= \begin{cases}
\underset{j=0}{\overset{n-1}{\sum}}\beta^{n-1}\xi^{0} & (j-g+1)~\mathrm{mod}~n=0\\
\underset{j=0}{\overset{n-1}{\sum}}\beta^{(n-i-1)}(\beta\xi^{(j-(g-1))})^i &  (j-g+1)~\mathrm{mod}~n\neq 0
\end{cases},\nonumber\\
&= \begin{cases}
n \beta^{n-1}& (j-g+1)~\mathrm{mod}~n=0\\
0 & (j-g+1)~\mathrm{mod}~n\neq 0\end{cases},\label{eqn:R4_mid5_EARS}
\end{align}
as $\xi^{n}=1$ and the roots of unity sum to $0$ (From Lemma \ref{lem1}). From equation \eqref{eqn:R4_mid5_EARS}, we have
\begin{align}
&\underset{j=0}{\overset{n-1}{\sum}}\beta^{(n-i-1)}(\beta\xi^{(j-(g-1))})^i= n\beta^{n-1}\delta_{(j-g+1),n} = n\delta_{g,(n+j+1)} = n\delta_{g,(j+1)},\label{eqn:R4_mid6_EARS}
\end{align}
as $g\in \{1,\dots,n\}$ and $(n+j+1) \equiv (j+1)$ in the range $1$ to $n$ ($\because (j+1) \in \{1,\dots ,n\}$).

Substituting equation \eqref{eqn:R4_mid6_EARS} in equation \eqref{eqn:R4_mid4_EARS}, we obtain
\begin{align}
&\mathrm{Q(FFFT)}\left(\underset{i=0}{\overset{n-1}{\otimes}}\mathrm{Z}^{(p^{k'})}(\gamma(\beta\alpha^{j})^{i})\right)(\mathrm{Q(FFFT))}^{-1}\nonumber\\
&= \underset{g=1}{\overset{n}{\otimes}}\mathrm{Z}^{(p^{k'})}\left(\frac{1}{n\lambda}\beta\gamma n\beta^{n-1}\delta_{g,j+1}\right)= \mathrm{Z}_{(j+1)}^{(p^{k'})}\left(\gamma\right),\nonumber
\end{align}
as $\beta^n=\lambda$ and in the tensor product $\underset{g=1}{\overset{n}{\otimes}}\mathrm{Z}^{(p^{k'})}\left(\frac{1}{n}\beta\gamma n\beta^{n-1}\delta_{g,j+1}\right)$, only the $(j+1)^{\mathrm{st}}$ element is non-identity ($\because g= (j+1)$), while the rest become identity ($\because g\neq  (j+1) \Rightarrow \mathrm{Z}^{(p^{k'})}(0)=\mathrm{I}_{p^{k'}}$). 
\end{proof}
\end{itemize}
\begin{remark}
    The following identities (N5) and (N6) are used to obtain the initial state $\ket{\psi_0}$ of the codeword qudits from the obtained initial stabilizers $\mathrm{X}_s^{(p^{k'})}(n\xi^l)$ and $\mathrm{Z}_t^{(p^{k'})}(\xi^m)$.
\end{remark}
\begin{itemize}
\item[\textbf{(N5)}]${\mathrm{X}}^{{(p^{k'})}}{(n\xi^l)}{\ket{\epsilon} = \ket{\epsilon},}\!\text{ where }{\ket{\epsilon}=\mathrm{DFT}}^{{-1}}{\ket{0}}$
\begin{proof}
From equation \eqref{eqn:ket_epsilon},
\begin{align}
\ket{\epsilon}&=\mathrm{DFT}^{-1}\ket{0}
= \frac{1}{\sqrt{p^{k'}}}\underset{\theta_1\in\mathbb{F}_{p^{k'}}}{\sum}\ket{\theta_1}.\label{eqn:DFT-1ket0}
\end{align}
For $\xi \in \mathbb{F}_{p^{k'}}$, we analyze the action of $\mathrm{X}^{(p^{k'})}(n\xi^l)$ on $\ket{\epsilon}$ from equation \eqref{eqn:DFT-1ket0} as follows:
\begin{align}
\mathrm{X}^{(p^{k'})}(n\xi^l )\ket{\epsilon}=& \mathrm{X}^{(p^{k'})}(n\xi^l )\frac{1}{\sqrt{p^{k'}}}\underset{\theta_1\in\mathbb{F}_{p^{k'}}}{\sum}\ket{\theta_1}= \frac{1}{\sqrt{p^{k'}}}\underset{\theta_1\in\mathbb{F}_{p^{k'}}}{\sum}\mathrm{X}^{(p^{k'})}(n\xi^l)\ket{\theta_1}\nonumber\\=& \frac{1}{\sqrt{p^{k'}}}\underset{\theta_1\in\mathbb{F}_{p^{k'}}}{\sum}\ket{\theta_1+n\xi^l }. \label{eqn:X_epsilon_mid1}
\end{align}
Let $\theta_1'=\theta_1+n\xi^l $. Performing the change of variable from $\theta_1$ to $\theta_1'$ in equation \eqref{eqn:X_epsilon_mid1}, the summation over all $\theta_1$ in $\mathbb{F}_{p^{k'}}$ changes to the summation over all $\theta_1'$ in $\mathbb{F}_{p^{k'}}$ by the closure property of $\mathbb{F}_{p^{k'}}$. Thus, performing change of variable in equation \eqref{eqn:X_epsilon_mid1}, we obtain
\begin{align}
\mathrm{X}^{(p^{k'})}(n\xi^l )\ket{\epsilon}&= \frac{1}{\sqrt{p^{k'}}}\underset{\theta_1'\in\mathbb{F}_{p^{k'}}}{\sum}\ket{\theta_1'}= \ket{\epsilon}. ~~~(\text{From equation }\eqref{eqn:DFT-1ket0})
\end{align}
\end{proof}
\item[\textbf{(N6)}]${\mathrm{Z}}^{{(p^{k'})}}{(\xi^m)}{\ket{0} = \ket{0}.}$
\begin{proof}
For $\xi  \in \mathbb{F}_{p^{k'}}$, we have
\begin{align}
\mathrm{Z}^{(p^{k'})}(\xi^m )\ket{0}&= \omega^{\mathrm{Tr}_{p^{k'}/p}(\xi^m.0)}\ket{0} = \ket{0}.
\end{align}
\end{proof}
\end{itemize}
\begin{remark}
    The following identities (N7) and (N8) are used for computing syndrome $\ket{s_{\mathrm{Z}}}$.
\end{remark}
\begin{itemize}
\item[\textbf{(N7)}] ${\mathrm{DFT}_{{p^{k'}}}~\mathrm{X}}^{{(p^{k'})}}{(\gamma )}{\mathrm{DFT}}_{{p^{k'}}}^{{-1}} {=\mathrm{Z}}^{{(p^{k'})}}{(\gamma )}$
\begin{proof}
For $\gamma  \in \mathbb{F}_{p^{k'}}$, from equation \eqref{eqn:DFT_EARS}, we have
\allowdisplaybreaks
\begin{align*}
\mathrm{DFT}_{p^{k'}}~\mathrm{X}^{(p^{k'})}(\gamma )\mathrm{DFT}_{p^{k'}}^{-1}
&= \!\left(\frac{1}{\sqrt{p^{k'}}}\underset{\theta_1,\theta_2\in \mathbb{F}_{p^{k'}}}{\sum}\!\!\!\!\!\omega^{\mathrm{Tr}_{p^{k'}/p}(\theta_1\theta_2)}\ket{\theta_1}\bra{\theta_2}\right)\mathrm{X}^{(p^{k'})}(\gamma )\nonumber\\
&~~~~\left(\frac{1}{\sqrt{p^{k'}}}\underset{\theta_3,\theta_4\in \mathbb{F}_{p^{k'}}}{\sum}\omega^{\mathrm{Tr}_{p^{k'}/p}(-\theta_3\theta_4)}\ket{\theta_3}\bra{\theta_4}\right),\nonumber\\
&= \!\frac{1}{p^{k'}}\underset{\theta_1,\theta_2,\theta_3,\theta_4\in \mathbb{F}_{p^{k'}}}{\sum}\!\!\!\!\!\!\!\!\!\!\!\omega^{\!\mathrm{Tr}_{p^{k'}\!/\!p}(\theta_1\theta_2-\theta_3\theta_4)\!\!}\ket{\theta_1}\bra{\theta_2}\ket{\mathrm{X}^{(p^{k'}}(\gamma )}\ket{\theta_3}\bra{\theta_4},\nonumber\\
&= \!\frac{1}{p^{k'}}\underset{\theta_1,\theta_2,\theta_3,\theta_4\in \mathbb{F}_{p^{k'}}}{\sum}\!\!\!\!\!\!\!\!\!\!\!\omega^{\!\mathrm{Tr}_{p^{k'}\!/\!p}(\theta_1\theta_2-\theta_3\theta_4)\!\!}\ket{\theta_1}\!\bra{\theta_2}\ket{\theta_3+\gamma}\!\bra{\theta_4},\nonumber\\
&= \!\frac{1}{p^{k'}}\underset{\theta_1,\theta_2,\theta_3,\theta_4\in \mathbb{F}_{p^{k'}}}{\sum}\!\!\!\!\!\!\!\!\!\!\!\omega^{\!\mathrm{Tr}_{p^{k'}\!/\!p}(\theta_1\theta_2-\theta_3\theta_4)\!\!}\ket{\theta_1}\!\delta_{\theta_2,(\theta_3+\gamma )}\!\bra{\theta_4},\nonumber\\
&= \!\frac{1}{p^{k'}}\underset{\theta_1,\theta_3,\theta_4\in \mathbb{F}_{p^{k'}}}{\sum}\!\!\!\!\!\!\!\omega^{\!\mathrm{Tr}_{p^{k'}\!/\!p}(\theta_1(\theta_3+\gamma )-\theta_3\theta_4)\!\!}\ket{\theta_1}\!\bra{\theta_4},\nonumber\\
&= \!\frac{1}{p^{k'}}\underset{\theta_1,\theta_3,\theta_4\in \mathbb{F}_{p^{k'}}}{\sum}\!\!\!\!\!\!\!\omega^{\!\mathrm{Tr}_{p^{k'}\!/\!p}(\theta_1\theta_3+\theta_1\gamma -\theta_3\theta_4)\!\!}\ket{\theta_1}\!\bra{\theta_4},\nonumber\\
&= \!\frac{1}{p^{k'}}\underset{\theta_1,\theta_4\in \mathbb{F}_{p^{k'}}}{\sum}\!\!\!\!\!\omega^{\!\mathrm{Tr}_{p^{k'}\!/\!p}(\theta_1\gamma )\!\!}\ket{\theta_1}\!\bra{\theta_4}\underset{\theta_3\in \mathbb{F}_{p^{k'}}}{\sum}\omega^{\!\mathrm{Tr}_{p^{k'}\!/\!p}(\theta_3(\theta_1-\theta_4))\!\!},\nonumber\\
&= \!\!\frac{1}{p^{k'}}\!\underset{\theta_1,\theta_4\in \mathbb{F}_{p^{k'}}}{\sum}\!\!\!\!\!\omega^{\!\mathrm{Tr}_{p^{k'}\!/\!p}(\theta_1\gamma )\!\!}\ket{\theta_1}\!\bra{\theta_4}p^{k'}\delta_{\theta_1-\theta_4,0},\!(\text{From Lemma \ref{lem3}})\nonumber\\
&= \!\frac{1}{p^{k'}}\underset{\theta_1,\theta_4\in \mathbb{F}_{p^{k'}}}{\sum}\!\!\!\!\!\omega^{\!\mathrm{Tr}_{p^{k'}\!/\!p}(\theta_1\gamma )\!\!}\ket{\theta_1}\!\bra{\theta_4}p^{k'}\delta_{\theta_1,\theta_4},\nonumber\\
&= \underset{\theta_1\in \mathbb{F}_{p^{k'}}}{\sum}\!\!\!\omega^{\!\mathrm{Tr}_{p^{k'}\!/\!p}(\theta_1\gamma )\!\!}\ket{\theta_1}\!\bra{\theta_1},\nonumber\\
&= \underset{\theta_1\in \mathbb{F}_{p^{k'}}}{\sum}\!\!\!\mathrm{Z}^{(p^{k'})}(\gamma )\ket{\theta_1}\!\bra{\theta_1}
= \mathrm{Z}^{(p^{k'})}(\gamma )\underset{\theta_1\in \mathbb{F}_{p^{k'}}}{\sum}\!\!\!\ket{\theta_1}\!\bra{\theta_1}= \mathrm{Z}^{(p^{k'})}(\gamma ),
\end{align*}
because $\underset{\theta_1\in \mathbb{F}_{p^{k'}}}{\sum}\!\!\!\ket{\theta_1}\!\bra{\theta_1} = \mathrm{I}_{p^{k'}}$ by the completeness relation.
\end{proof}
\item[\textbf{(N8)}] ${\mathrm{DFT}_{p^{k'}}~\mathrm{Z}}^{{(p^{k'})}}{(\gamma)}{\mathrm{DFT}}_{p^{k'}}^{{-1}} {=\mathrm{X}}^{{(p^{k'})}}{(-\gamma)}$
\begin{proof}
For $\gamma \in \mathbb{F}_{p^{k'}}$, from equation \eqref{eqn:DFT_EARS}, we have
\allowdisplaybreaks
\begin{align*}
\mathrm{DFT}_{p^{k'}}~\mathrm{Z}{(p^{k'})}(\gamma)\mathrm{DFT}_{p^{k'}}^{-1}
&= \!\left(\frac{1}{\sqrt{p^{k'}}}\underset{\theta_1,\theta_2\in \mathbb{F}_{p^{k'}}}{\sum}\!\!\!\!\!\omega^{\mathrm{Tr}_{p^{k'}/p}(\theta_1\theta_2)}\ket{\theta_1}\bra{\theta_2}\right)\mathrm{Z}^{(p^{k'})}(\gamma)\nonumber\\
&~~~~\left(\frac{1}{\sqrt{p^{k'}}}\underset{\theta_3,\theta_4\in \mathbb{F}_{p^{k'}}}{\sum}\omega^{\mathrm{Tr}_{p^{k'}/p}(-\theta_3\theta_4)}\ket{\theta_3}\bra{\theta_4}\right),\nonumber\\
&= \!\frac{1}{p^{k'}}\underset{\theta_1,\theta_2,\theta_3,\theta_4\in \mathbb{F}_{p^{k'}}}{\sum}\!\!\!\!\!\!\!\!\!\!\!\omega^{\!\mathrm{Tr}_{p^{k'}\!/\!p}(\theta_1\theta_2-\theta_3\theta_4)\!\!}\ket{\theta_1}\!\bra{\theta_2}\ket{\mathrm{Z}^{(p^{k'})}\!(\gamma)}\ket{\theta_3}\!\bra{\theta_4},\nonumber\\
&= \!\frac{1}{p^{k'}}\underset{\theta_1,\theta_2,\theta_3,\theta_4\in \mathbb{F}_{p^{k'}}}{\sum}\!\!\!\!\!\!\!\!\!\!\!\omega^{\!\mathrm{Tr}_{p^{k'}\!/\!p}(\theta_1\theta_2-\theta_3\theta_4+\theta_3\gamma)\!\!}\ket{\theta_1}\!\bra{\theta_2}\ket{\theta_3}\!\bra{\theta_4},\nonumber\\
&= \!\frac{1}{p^{k'}}\underset{\theta_1,\theta_2,\theta_3,\theta_4\in \mathbb{F}_{p^{k'}}}{\sum}\!\!\!\!\!\!\!\!\!\!\!\omega^{\!\mathrm{Tr}_{p^{k'}\!/\!p}(\theta_1\theta_2-\theta_3\theta_4+\theta_3\gamma)\!\!}\ket{\theta_1}\!\delta_{\theta_2,\theta_3}\!\bra{\theta_4},\nonumber\\
&= \!\frac{1}{p^{k'}}\underset{\theta_1,\theta_3,\theta_4\in \mathbb{F}_{p^{k'}}}{\sum}\!\!\!\!\!\!\!\omega^{\!\mathrm{Tr}_{p^{k'}\!/\!p}(\theta_1\theta_3-\theta_3\theta_4+\theta_3\gamma)\!\!}\ket{\theta_1}\!\bra{\theta_4},\nonumber\\
&= \!\frac{1}{p^{k'}}\underset{\theta_1,\theta_4\in \mathbb{F}_{p^{k'}}}{\sum}\!\!\!\!\!\ket{\theta_1}\!\bra{\theta_4}\underset{\theta_3\in \mathbb{F}_{p^{k'}}}{\sum}\omega^{\!\mathrm{Tr}_{p^{k'}\!/\!p}(\theta_3(\theta_1-\theta_4+\gamma))\!\!},\nonumber\\
&= \!\!\frac{1}{p^{k'}}\!\underset{\theta_1,\theta_4\in \mathbb{F}_{p^{k'}}}{\sum}\!\!\!\!\!\ket{\theta_1}\!\bra{\theta_4}p^{k'}\delta_{\theta_1-\theta_4+\gamma,0},\!(\text{From Lemma \ref{lem3}})\nonumber\\
&= \!\frac{1}{p^{k'}}\underset{\theta_1,\theta_4\in \mathbb{F}_{p^{k'}}}{\sum}\!\!\!\!\!\ket{\theta_1}\!\bra{\theta_4}p^{k'}\delta_{\theta_1,(\theta_4-\gamma)},\nonumber\\
&= \underset{\theta_4\in \mathbb{F}_{p^{k'}}}{\sum}\!\ket{\theta_4-\gamma}\!\bra{\theta_4},\nonumber\\
&= \mathrm{X}^{(p^{k'})}(-\gamma )\underset{\theta_1\in \mathbb{F}_{p^{k'}}}{\sum}\!\!\!\ket{\theta_4}\!\bra{\theta_4}
= \mathrm{X}^{(p^{k'})}(-\gamma ),
\end{align*}
\end{proof}
\end{itemize}
 \end{document}